\documentclass[lettersize,journal]{IEEEtran}
\usepackage{cite}
\usepackage{amsmath,amssymb,amsfonts, bm, algorithm, subcaption, adjustbox}
\usepackage{amsthm}
\usepackage{algorithmic}
\usepackage{graphicx}
\usepackage{verbatim}
\usepackage{textcomp}
\usepackage{xcolor}
\usepackage{url}
\usepackage{stfloats}
\usepackage{pifont}
\usepackage{caption}
\usepackage{subcaption}
\usepackage{graphicx}
\usepackage{lipsum} 
 
\def\hlb{\color{blue}}
\def\hlb{\color{black}}

\newtheorem{theorem}{Theorem}
\newtheorem{lemma}{Lemma}
\newtheorem{corollary}{Corollary}
\def\BibTeX{{\rm B\kern-.05em{\sc i\kern-.025em b}\kern-.08em
		T\kern-.1667em\lower.7ex\hbox{E}\kern-.125emX}}
\begin{document}

\title{Hybrid Hierarchical Federated Learning over 5G/NextG Wireless Networking}
 
\author{Haiyun Liu, Jiahao Xue,
	Jie Xu,~\IEEEmembership{Senior Member,~IEEE,}
	Yao Liu,~\IEEEmembership{Senior Member,~IEEE,}
    and Zhuo Lu,~\IEEEmembership{Senior Member,~IEEE}
\thanks{Haiyun Liu and Yao Liu are with the Department of Computer Science and Engineering,
    	University of South Florida, Tampa, FL 33620 USA (e-mail: haiyunliu@usf.edu; yliu21@usf.edu).}
\thanks{Jie Xu is with the Department of Electrical and Computer Engineering, University of Florida, Gainesville, FL 32611 USA (e-mail: jie.xu@ufl.edu).}
\thanks{Jiahao Xue and Zhuo Lu are with the Department of Electrical Engineering,
University of South Florida, Tampa, FL 33620 USA (e-mail: jiahao@usf.edu; zhuolu@usf.edu).}%

}

\maketitle

\begin{abstract}
	Today’s 5G and NextG wireless networks are moving toward using the coordinated multi-point (CoMP) transmission and reception technique, where a client can be simultaneously served by multiple base stations (BSs) for better communication performance. However, traditional hierarchical federated learning (HFL) architectures impose the constraint that each client can be associated with only one edge server (ES) at a time. If we keep using the traditional HFL architectures in modern hierarchical networks for model training, the benefits of the CoMP technique would remain unexploited and leave room for further improvements in training efficiency. To address this issue, we propose hybrid hierarchical federated learning (HHFL), which allows clients in overlapping regions to simultaneously communicate with multiple edge servers (ESs) for model aggregation. HHFL is able to enhance inter-ES knowledge sharing, thereby mitigating model divergence and improving training efficiency. We provide a rigorous theoretical convergence analysis with a convergence upper bound to validate its effectiveness. Experimental results show that HHFL outperforms traditional HFL, particularly when the data across different ESs is not independent and identically distributed (non-IID). For example, when each ES is dominated by only two of the ten classes and 15 out of the 57 clients can connect to multiple ESs, HHFL achieves up to 2× faster convergence under the same configuration. These results demonstrate that HHFL provides a scalable and efficient solution for FL model training in today’s and NextG wireless networks.
	
\end{abstract}

\begin{IEEEkeywords}
	Federated learning (FL); Hierarchical Federated Learning (HFL); Multi-server; Training efficiency.
\end{IEEEkeywords}

\section{Introduction}
\IEEEPARstart{F}{ederated} learning (FL) in wireless networking is able to effectively utilize data locally stored by wireless users for machine learning tasks with data privacy \cite{brik2020federated,nguyen2021federated,yin2020fedloc}. In real-world applications, the design of the FL architecture, which defines how clients and servers communicate and collaborate during the learning process, must account for the network topology, which describes how devices are physically or logically connected within the wireless system, as their alignment directly impacts training effectiveness and efficiency.

In small-scale wireless networks, the star network topology is widely adopted, where a centralized server device directly connects to multiple distributed client devices to finish a training task \cite{mehmet1988traffic,petrek2001large,goratti2015nacrp,barranco2006active,day2002comparative}. As networks continue to expand and become increasingly complex, as exemplified by the Internet of Things (IoT), large-scale vehicular networks (V2X), and edge-cloud computing frameworks, the escalating demands for communication and computation often exceed the capacity of a single server device. To address these challenges, the hierarchical network topology is employed, where intermediate-layer servers are designed to reduce the load on the previously centralized server \cite{lessmann2007parameterized,xu2021adaptive,ravasz2003hierarchical}. To fit this common topology, hierarchical federated learning (HFL) has been proposed, as illustrated in Fig.~\ref{Fig:basickeyframework}(a). {\hlb In this architecture, training proceeds periodically. The cloud server (CS) disseminates the global model to all edge servers (ESs), which broadcast it to their clients for initialization. Clients perform local updates and upload their models to the connectable ES(s) for aggregation; the aggregated ES model is then broadcast back to clients. After several client–ES rounds, ESs upload their models to the CS for aggregation to update the global model.}

\begin{figure}[t!]
	\centering
	\includegraphics[width=0.5\textwidth]{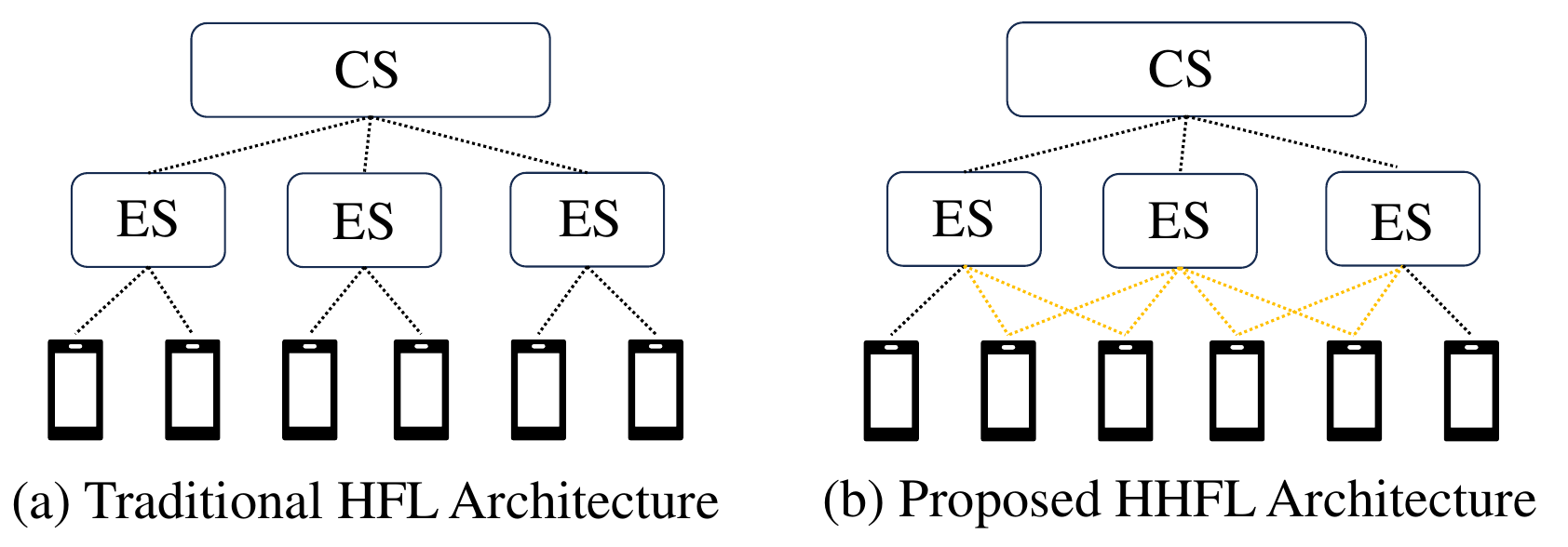}
	\caption{The traditional HFL architecture and our proposed architecture.}
	\label{Fig:basickeyframework}
\end{figure}

As user demands for high data rates, low latency, and strong reliability continue to grow, the multi-point connection mode, i.e., coordinated multi-point (CoMP), has been proposed as a more advanced alternative to the traditional single-point approach in hierarchical network architectures \cite{jiao2023performance, ye2023multi}. In this mode, clients can concurrently connect to multiple intermediate-layer base stations (BSs, which correspond to ESs in FL), which better utilizes network resources and enhances communication performance. This connection mode has already seen widespread adoption in 5G networks and is expected to play an even more critical role in the development of 6G \cite{xiao2024space,xu2020research,yang2022federated,chowdhury20206g,akyildiz20206g,yang20196g}.
Theoretically, we could continue using the traditional HFL architecture in today’s and future hierarchical networks that support CoMP. In this case, each client is assigned to communicate with a single BS, either randomly or based on specific communication parameters such as received signal strength (RSS). However, this design does not fully harness the advantages of multi-point mode to benefit FL training. {\hlb In particular, under the HFL framework, the models of different ESs perform updates in isolation for extended periods between two consecutive aggregations at the CS; their updated models tend to accumulate biases from heterogeneous local gradients and drift toward localized optima. This inconsistency hinders global convergence, especially when the data distribution across ESs is not independent and identically distributed (non-IID).} 

To address this issue, we propose a novel hybrid hierarchical federated learning (HHFL) architecture, illustrated in Fig.~\ref{Fig:basickeyframework}(b). {\hlb In HHFL, each client obtains the model parameters from its connected ESs and aggregates them via simple arithmetic averaging (i.e., summing the received models and dividing by the number of connected ESs). The client then performs local updates and periodically uploads the updated models back to the connected ESs for further aggregation.} After a fixed number of aggregations, the ESs transmit their aggregated models to the CS for further aggregation, resulting in an updated global model. Then, this global model is sent back to the ESs. {\hlb Compared to traditional HFL, a distinct advantage of HHFL is that each client in overlapping areas functions as a knowledge bridge between neighboring ESs. Specifically, by aggregating models from multiple neighboring ESs and subsequently uploading the updated model to all associated servers, the client effectively fuses knowledge derived from distinct ESs and propagates this fused knowledge across cell boundaries. As a result, clients under a given ES can indirectly benefit from knowledge carried by clients covered by other ESs, which mitigates overly localized model updates at the ES level, makes the aggregated updates more consistent with the global optimization objective, and ultimately accelerates convergence. }

Through the rigorous analysis presented in Section~\ref{Sub:convergence}, we demonstrate that HHFL achieves effective convergence as training progresses, with its convergence upper bound formally established therein.
The experimental results in Section~\ref{Sec:evaluation} show that when the data distribution across ESs is non-IID and some clients are located in overlapping regions covered by multiple ESs, HHFL exhibits a faster convergence speed compared to traditional HFL. Specifically, in a representative case where the data held by each ES is dominated by only two out of the ten classes, and 15 out of the 57 clients are located in overlapping regions, HHFL achieves approximately a $100\%$ improvement in convergence efficiency over traditional HFL. Our contributions are as follows.

\begin{enumerate}
	\item {\hlb We propose the novel HHFL architecture based on the CoMP technology. This architecture enables clients located in overlapping areas to act as bridges and facilitates knowledge sharing across ESs.}
	\item {\hlb We provide a rigorous theoretical convergence analysis of HHFL, and discuss its practical deployment considerations to confirm its feasibility in modern 5G/NextG networks.}
	\item {\hlb We conduct extensive experiments to demonstrate the superior efficiency of HHFL over traditional HFL.}
\end{enumerate}

The rest of this paper is organized as follows: Section~\ref{Sec:preliminaries} introduces the preliminaries of FL. Section~\ref{Sec:design} introduces the HHFL architecture in detail, along with a well-designed algorithm. Section~\ref{Sec:performance} analyzes the performance of our HHFL architecture. Section~\ref{Sec:evaluation} conducts experimental simulations to validate the high efficiency of HHFL.  
Section~\ref{Sec:relatedwork} summarizes related works, followed by the conclusion in Section~\ref{Sec:conclusion}.

\section{Preliminaries and Models}\label{Sec:preliminaries}
In this section, we briefly introduce the classical FL architectures and their corresponding algorithms.

\subsection{Basic FL Architecture}

A basic FL architecture includes $K$ participating clients that directly communicate with a CS that enables collaborative model training by aggregating locally computed updates from all clients \cite{mcmahan2017communication,kairouz2021advances,bonawitz2017practical,li2020federated,smith2017federated}. Specifically, in each communication round, the CS broadcasts the current global model to clients. Each client then performs local training based on its private data for $E$ steps to get the model update. Finally, each  client returns its update to the CS, which aggregates these model updates to refine the global model and send them in the next round.

One of the most representative algorithms for this basic architecture is \textit{FedAvg}~\cite{mcmahan2017communication}, in which the CS updates the global model by computing a weighted average of the client models.

Formally, let $t$ denote the index of local update steps, and $t_0$ be the starting step of a communication round. Then, the \textit{FedAvg} algorithm proceeds iteratively as follows:

\begin{enumerate}
	\item The CS broadcasts the current global model $w^{t_0}$ to all $K$ clients, and each client $i$ initializes its local model as $w_i^{t_0} = w^{t_0}$.
	\item Each client $i$ performs $E$ steps of local updates as:
	\begin{equation}\label{Eq:sgd}
		\begin{split}
			w_i^{{t_0}+e+1}=w_i^{{t_0}+e}-\eta_{{t_0}+e}\nabla F_i(w_i^{{t_0}+e}),
		\end{split}
	\end{equation}
	where $e = 0, \dots, E-1$, $\eta_{{t_0}+e}$ is the learning rate at step ${t_0}+e$, and $F_i(w_i^{{t_0}+e})$ is the average loss of client $i$ with respect to its local model $w_i^{{t_0}+e}$. After  $E$ local updates, the  model of client $i$ becomes $w_i^{t_0+E}$.
	\item The CS aggregates all the client models to get an updated global model as:
	\begin{equation}\label{Eq:aggregation}
		\begin{split}
			w^{{t_0}+E}=\sum_{i=1}^{K}p_iw_i^{{t_0}+E}.
		\end{split}
	\end{equation}
	where $p_i$ denotes the proportion of data held by client $i$ relative to the total amount of data across all clients.
\end{enumerate}
These three procedures are repeated until either a predefined number of communication rounds is reached or the model meets a specified convergence criterion.

\subsection{The HFL Architecture}\label{SubSec:HFL}
As communication networks scale and user demands continue to grow, the basic architecture of FL increasingly falls short of meeting practical requirements. To address these issues, the HFL architecture has been introduced \cite{liu2020client}. Unlike the basic client-server setup, HFL places ESs between clients and the CS, forming a multi-layer structure of client-ES-CS, as shown in Fig.~\ref{Fig:basickeyframework}(a). In each global communication round, the global model is first transmitted from the CS to all ESs, which then forward it to their respective clients. Upon receiving the model, each client performs $E$ steps of local updates and sends the updated model to its associated ES for \textit{edge aggregation}.  Each ES performs $G$ rounds of edge aggregation with its connected clients. Finally, the updated models from all ESs are sent back to the CS, which performs \textit{cloud aggregation} to update the global model.

A representative HFL algorithm is \textit{Hier-FedAvg}, which extends the classical \textit{FedAvg} algorithm to the multi-layer setting \cite{liu2020client}. Suppose there are $N$ ESs, denoted by $ES_1, ES_2, \dots, ES_N$. For each $ES_n$, let $C_n$ denote the set of indices of all clients that it covers. Similarly, for each client $i$ with $i \in \{1, \dots, K\}$, let $S_i$ denote the indices of ES to which client $i$ can connect to.
Then, the \textit{Hier-FedAvg} algorithm proceeds iteratively as follows:

\begin{enumerate}
	\item  CS broadcasts its global model $w^{t_0}$ to all ESs, and each $ES_n$ sets its current model  $w_{(n)}^{{t_0}}$ to $w^{t_0}$.
	\item Each client $i$ initializes its model as:
	\begin{equation}\label{Eq:clientinitial}
		\begin{split}
			w_{i}^{{t_0}}=w_{(S_i)}^{{t_0}}.
		\end{split}
	\end{equation}
	\item Each client $i$ performs $E$ steps of local updates as:
	\begin{equation}\label{Eq:localupdate}
		\begin{split}
			w_{i}^{{t_0}+e+1}=w_{i}^{{t_0}+e}-\eta_{{t_0}+e}\nabla F_i(w_i^{{t_0}+e}).
		\end{split}
	\end{equation}
	After $E$ local updates, the model of client $i$ is $w_i^{{t_0}+E}$.
	
	\item Each client $i$ sends its model $w_i^{{t_0}+E}$ to the connected ES for {edge aggregation}. For each $ES_n$, its aggregated model is as:
	\begin{equation}\label{Eq:ES aggregationfor}
		\begin{split}
			w_{(n)}^{{t_0}+E}=\sum_{i\in C_n}\frac{N_i}{N_n}w_i^{{t_0}+E},
		\end{split}
	\end{equation}
	where $N_i$ is the number of local data held by client $i$, and $N_n = \sum_{i \in \mathcal{C}_n} N_i$ is the number of data covered by $ES_n$.
	\item After repeating 2), 3), and 4) for $G$ rounds, each $ES_n$ sends its model to CS for {cloud aggregation}. Then, the updated global model is as:
	\begin{equation}\label{Eq:globalm}
		\begin{split}
			w^{{t_0}+GE}=\sum_{n=1}^{N} \frac{N_n}{N_{total}}w_{(n)}^{{t_0}+GE}.
		\end{split}
	\end{equation}
	where $N_{total}= \sum_{n=1}^{N} N_n$ is the number of data covered by all ESs.
\end{enumerate}
These procedures are repeated until either a specified number of training steps is reached or the model meets a specified convergence criterion.

\section{Architecture Design} \label{Sec:design}
In this section, we first present the motivation behind our HHFL architecture and describe the overall design. Then, we design the algorithm for this architecture.

\subsection{Motivation and Architecture Design}
With the increasing densification of wireless infrastructures and the growing demand for seamless coverage in 5G and NextG communication systems, overlapping regions among BSs have become increasingly common \cite{bhushan2014network,borralho2021survey,ullah2023survey,rehman2023survey}. The coordinated multi-point (CoMP) transmission and reception technique has been proposed \cite{shamai2001enhancing} as an enhanced technique over traditional signal-point model (i.e., a client can communicate with only one BS). As shown in Fig.~\ref{Fig:CoMP}, compared to the traditional single-point mode, CoMP allows a client to simultaneously transmit to and receive data from multiple BSs, thereby enhancing system performance by improving spectral efficiency, reducing inter-cell interference, and significantly boosting the quality of service for users \cite{shamai2001enhancing,qamar2017comprehensive,irram2020coordinated,shen2022comp}. CoMP has been widely adopted in the 5G era and is expected to play an even more significant role in NextG networks \cite{irram2020coordinated,shen2022comp,solaija2021generalized,jiang2021road}.

\begin{figure}[t!]
	\centering
	\includegraphics[width=0.35\textwidth]{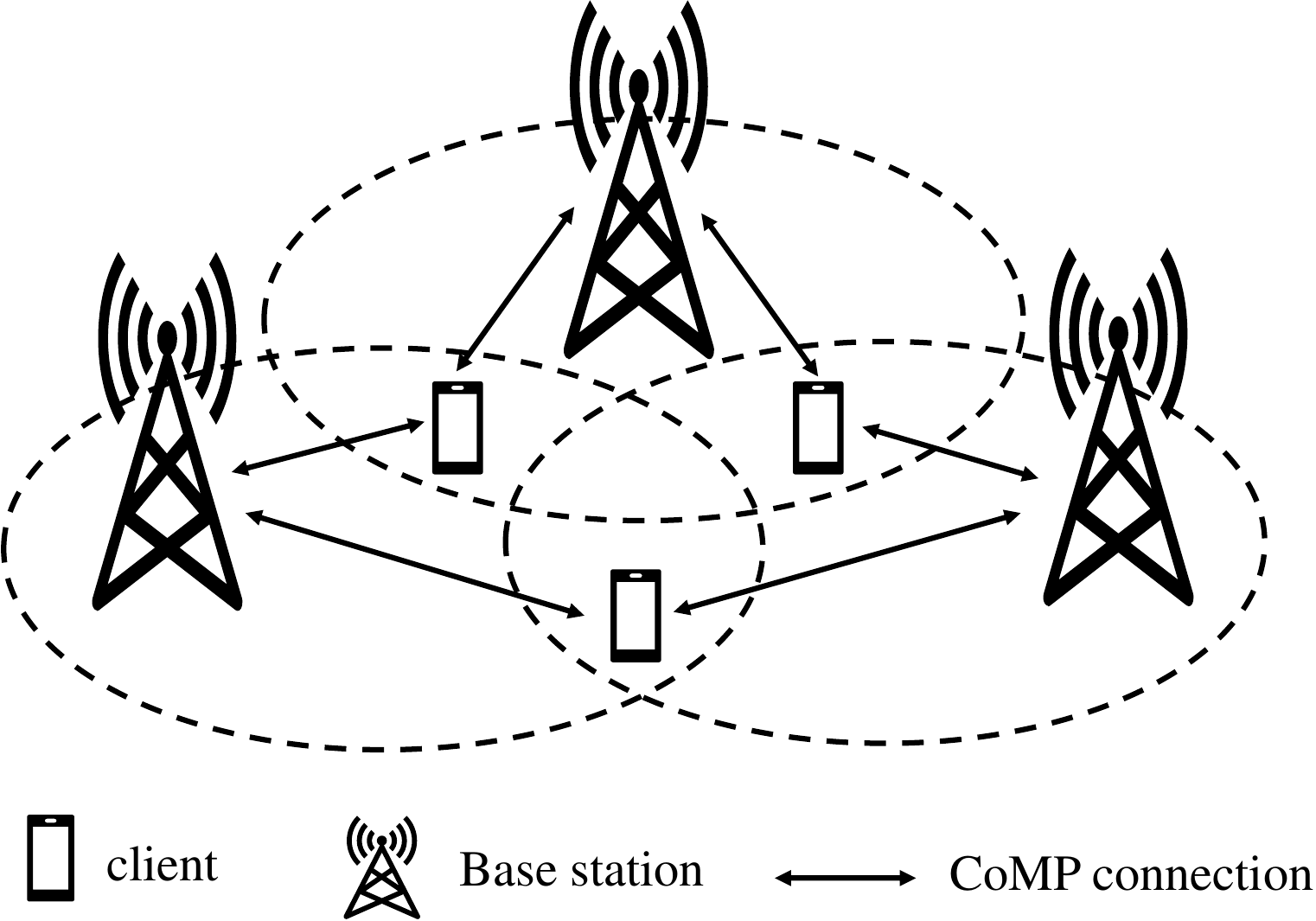}
	\caption{Coordinated multi-point technique.}
	\label{Fig:CoMP}
\end{figure}

Although some existing studies on HFL have considered making use of the overlapping regions and allow clients to be associated with multiple ESs, such as enabling dynamic edge association \cite{luo2020hfel}, leveraging statistical similarities across clients \cite{ma2024personalized}, optimizing staleness-aware aggregation schedules \cite{wu2023hiflash}, analyzing security implications of multi-edge connectivity \cite{alqattan2024security}, or modeling asynchronous communication and dynamic client–ES interactions \cite{mitra2023timely}, {\hlb they essentially still follow the conventional HFL paradigm, where multi-connectivity in a CoMP-capable network is not utilized. }

To address this under-utilization, we propose the new HHFL to effectively leverage the CoMP technique to benefit the FL training. Specifically, the FL training runs in a hybrid mode: for those clients located in overlapping regions, they are able to simultaneously communicate with multiple ESs through CoMP techniques. In contrast, clients that fall within the coverage of only a single ES still communicate with the ES alone. {\hlb In each global communication round, the CS first disseminates the global model to all ESs. Then, each ES transmits its current model to the clients within its coverage to initialize their local models. For clients in any overlapping area, CoMP-coordinated scheduling (CoMP-CS) allocates frequency-orthogonal physical resource blocks (PRBs) to their accessible ESs via the Xn interface (a standard interface for inter-BS/ES communication) for downlink transmission. This mechanism prevents inter-cell interference, enabling the simultaneous reception of multiple models. Upon receiving these models, the clients in overlapping areas perform \textit{client aggregation} to initialize their local models.
	Next, each client conducts local training for $E$ steps and uploads its updated model to all associated ESs for {edge aggregation}. 
	For each client in the overlapping area, a single uplink transmission of its local model can be received and processed in parallel by multiple ESs via CoMP-capable uplink multi-point reception. Meanwhile, different clients in the same area are scheduled on frequency-orthogonal PRBs to avoid mutual interference. After $G$ rounds of {edge aggregation}, each ES sends its current model to the CS for {cloud aggregation}, which updates the global model.}

{\hlb It is worth noting that employing CoMP inevitably introduces additional synchronization requirements, which increase implementation complexity. Nevertheless, our proposed HHFL architecture mitigates these costs through a lightweight and practical CoMP design. Specifically, rather than relying on traditional CoMP joint reception, multiple ESs independently receive and process the uplink data transmitted by the same client. This bypasses the need for complex inter-BS cooperation, such as coherent joint decoding or soft combining, which typically necessitates stringent sub-symbol phase synchronization and heavy backhaul data exchange. Moreover, other fundamental synchronization requirements, such as basic frame and slot timing alignment to absorb marginal propagation delays, strictly align with standard network operations. These macro-level synchronization mechanisms are inherently mandated in all modern cellular systems to maintain basic connectivity and prevent cross-link interference, meaning our architecture successfully accommodates multi-point reception without imposing the prohibitive physical-layer synchronization overheads typically associated with traditional CoMP.
	Finally, to avoid the latency of per-round dynamic negotiation among ESs for allocating PRBs to the clients in overlapping areas, HHFL exploits the inherently predictable and periodic traffic patterns of FL. Because the model payload size remains constant and uploads occur at regular intervals (i.e., exactly after $E$ local updates), the network can leverage 5G configured grants (CG) \cite{liu2020analyzing} to pre-allocate periodic uplink resources, effectively bypassing dynamic resource requests. 
}

{\hlb Building upon this practical multi-connectivity communication foundation, from an algorithmic perspective, HHFL differs from traditional HFL in that it involves two levels of aggregation: {client aggregation} and {edge aggregation}. Together, these mechanisms allow clients located in overlapping areas to act as knowledge {bridges} between neighboring ESs. Specifically, by aggregating models from multiple ESs, these clients fuse knowledge derived from distinct ESs; subsequently, by uploading their locally trained updates back to all associated ESs, they propagate this fused knowledge across cell boundaries. This mechanism enables clients under a given ES to indirectly benefit from data distributions predominantly located in neighboring ES coverage areas, thereby mitigating overly localized model updates at the ES level. When ES-level non-IID is pronounced, this extra inter-ES sharing makes their aggregated updates more consistent with the global optimization direction, thereby accelerating convergence.}

\subsection{HHFL Algorithm Design}
Let $S_i$ denote the set of indices of all ESs that client $i$ can connect to, rather than a single ES as previously defined. Similar to the Hier-FedAvg algorithm for HFL described in Section~\ref{SubSec:HFL}, our HHFL algorithm operates iteratively:
\begin{enumerate}
	\item  CS broadcasts the global model $w^{t_0}$ to all ESs, and each $ES_n$ sets its current model parameter $w_{(n)}^{{t_0}}$ to $w^{t_0}$.
	\item Each client $i$ executes {client aggregation} to initialize its client model as:
	
	\begin{equation}\label{Eq:wforES}
		\begin{split}
			w_{i}^{t_0}=\frac{1}{|S_i|}\sum_{n\in S_i}{w_{(n)}^{t_0}},
		\end{split}
	\end{equation}
	where $|S_i|$ is the number of elements in set $S_i$.
	\item Each client $i$ performs $E$ steps of local updates as:
	\begin{equation}\label{Eq:localupdateforproposed}
		\begin{split}
			w_{i}^{{t_0}+e+1}=w_{i}^{{t_0}+e}-\eta_{{t_0}+e}\nabla F_i(w_i^{{t_0}+e}).
		\end{split}
	\end{equation}
	After $E$ local updates, the model of client $i$ is $w_i^{{t_0}+E}$.
	
	\item Each client $i$ sends its model $w_i^{{t_0}+E}$ to the ESs it can connect to for {edge aggregation}. For each $ES_n$, its aggregated model is as:
	\begin{equation}\label{Eq:ESaggregationforproposed}
		\begin{split}
			w_{(n)}^{{t_0}+E}=\sum_{i\in C_n}\frac{1}{\lozenge_n}\frac{p_i}{|S_i|}w_i^{{t_0}+E},
		\end{split}
	\end{equation}
	where $\lozenge_n$ is determined by
	\begin{equation}\label{Eq:lozenge}
		\begin{split}
			\lozenge_n=\sum_{i\in C_n}\frac{p_i}{|S_i|}.
		\end{split}
	\end{equation}
	\item After repeating procedures 2), 3), and 4) for $G$ rounds, each $ES_n$ sends its model  to CS for {cloud aggregation}. Then, the updated global model is as:
	\begin{equation}\label{Eq:globalforproposed}
		\begin{split}
			w^{t+GE}=\sum_{n=1}^{N}\lozenge_n w_{(n)}^{t+GE}.
		\end{split}
	\end{equation}
	
\end{enumerate}
These procedures will be repeated until either the specified overall training steps are reached or the model meets a specified convergence criterion. The pseudo-code of the algorithm is shown in Algorithm~\ref{ALG:HHFL}.

\begin{algorithm}[t]
	\caption{The Proposed HHFL Algorithm}\label{ALG:HHFL}
	\begin{algorithmic}[1]
		\STATE \textbf{Input:} {Initialized parameter $w^0$ and total training steps $T$}
		\STATE \textbf{Output:} {Final global parameter $w^T$}
		\STATE Set $w_{i}^0=w^0$ for all clients $i=1,...,K$
		\FOR{$t = 0, 1, \ldots, T-1$}
		\FORALL{client $i=1,...,K$ \textbf{in parallel}}
		\STATE $v_{i}^{t+1}=w_{i}^{t}-\eta_{t}\nabla F_i(w_i^{t})$  // Model after one local update
		
		\ENDFOR
		\IF{$E \nmid t+1$}
		\FORALL{client $i=1,...,K$ \textbf{in parallel}}
		\STATE $w_{i}^{t+1}=v_{i}^{t+1}$
		
		\ENDFOR
		\ENDIF
		\IF{$E \mid t+1$}
		\FORALL{edge server $ES_n$, $n=1,...,N$ \textbf{in parallel}}
		\STATE $w_{(n)}^{t+1}=\sum_{i\in C_n}\frac{1}{\lozenge_n}\frac{p_i}{|S_i|}v_i^{t+1}$ // Edge aggregation
		\ENDFOR
		\IF{ $GE \mid t+1$}
		\STATE 	$w^{t+1}=\sum_{n=1}^{N}\lozenge_n w_{(n)}^{t+1}$ // Cloud aggregation
		\FORALL{edge server $ES_n$ with $n=1,...,N$ \textbf{in parallel}}
		\STATE $w_{(n)}^{t+1}=w^{t+1}$ // CS broadcasts its parameter to all ESs
		\ENDFOR
		\ENDIF
		
		\FORALL{client $i=1,...,K$ \textbf{in parallel}}
		\STATE $w_{i}^{t+1}=\frac{1}{|S_i|}\sum_{n\in S_i}{w_{(n)}^{t+1}}$ // Client aggregation
		\ENDFOR
		\ENDIF
		\ENDFOR
		\STATE $w^{T}=\sum_{n=1}^{N}\left(\lozenge_n \sum_{i\in C_n}\frac{1}{\lozenge_n}\frac{p_i}{|S_i|}w_i^{T}\right)$ // Final global parameter obtained through edge aggregations and cloud aggregation for all client models at step $T$.
	\end{algorithmic}
\end{algorithm}

\section{Performance Analysis}\label{Sec:performance}
In this section, we analyze the convergence properties of the HHFL architecture and derive an upper bound on its convergence to demonstrate its effectiveness. Then, we compare the proposed architecture with the traditional HFL architecture to demonstrate its efficiency advantages.

\subsection{Convergence Analysis}\label{Sub:convergence}
\subsubsection{Evolution Path of Client Models}
In the HHFL architecture, the global model is closely related to the local models maintained by individual clients. To investigate whether and how the global model converges as training progresses, we first need to understand how client models evolve after every local update.
We use the state vectors $\mathbf{w}^t = [w_1^t, \dots, w_K^t]^\top$ and $\mathbf{w}^{t+1} = [w_1^{t+1}, \dots, w_K^{t+1}]^\top$ to represent the collection of local models of all clients at step $t$ and step $t+1$, respectively.
Next, we analyze the evolution path from $\mathbf{w}^t$ to $\mathbf{w}^{t+1}$, denoted by $\mathbf{w}^t \rightarrow \mathbf{w}^{t+1}$.

We first introduce an intermediate state vector $\mathbf{v}^{t+1} = [v_1^{t+1}, \dots, v_K^{t+1}]^\top$. This vector represents the collection of all clients’ local updates at step $t$, obtained immediately after local training and before aggregation (if any). 
Then, we can decompose the evolution path $ \mathbf{w}^t \rightarrow \mathbf{w}^{t+1} $ into two successive transitions: 
$ \mathbf{w}^t \rightarrow \mathbf{v}^{t+1} $ (local training) and 
$ \mathbf{v}^{t+1} \rightarrow \mathbf{w}^{t+1} $ (aggregation, if performed).
This decomposition enables us to examine the effects of computation and aggregation separately, providing a clearer view of how model evolves during training.
In other words, understanding the evolution $\mathbf{w}^t \rightarrow \mathbf{w}^{t+1}$ requires analyzing both $\mathbf{w}^t \rightarrow \mathbf{v}^{t+1}$ and $\mathbf{v}^{t+1} \rightarrow \mathbf{w}^{t+1}$.

Since $v_i^{t+1}$ is obtained by client $i$ after one local update based on the model $w_i^t$, we have:
\begin{equation}\label{Eq:w_t2v_t+1}
	\begin{split}
		v_i^{t+1}=w_i^{t}-\eta_{t}\nabla F_i(w_i^{t}),
	\end{split}
\end{equation}
which characterizes the first transition $\mathbf{w}^t \rightarrow \mathbf{v}^{t+1}$.
For the transition $ \mathbf{v}^{t+1} \rightarrow \mathbf{w}^{t+1} $, we analyze the process under three distinct cases based on the value of step $ t+1 $:

\noindent \textbf{Case 1}: $E \nmid t+1$. At this step, there is no any aggregation, and we have
\begin{equation}\label{Eq:case1}
	\begin{split}
		w_i^{t+1} = v_i^{t+1}.
	\end{split}
\end{equation}

\noindent \textbf{Case 2}: $EG\mid t+1$. At this step, {edge aggregation}, {cloud aggregation}, and {client aggregation} perform sequentially. Then, we have
\begin{equation}\label{Eq:wforMultiEG1}
	\begin{split}
		w_{i}^{t+1}&\stackrel{(a)}{=}\frac{1}{|S_i|}\sum_{n\in S_i}{w_{(n)}^{t+1}}\stackrel{(b)}{=}\frac{1}{|S_i|}\sum_{n\in S_i}{w^{t+1}}
		=w^{t+1}\\&\stackrel{(c)}{=}\sum_{n=1}^{N}\lozenge_n w_{(n)}^{t+1}\stackrel{(d)}{=}\sum_{n=1}^{N}\left(\lozenge_n  \sum_{i\in C_n}\frac{1}{\lozenge_n}\frac{p_i}{|S_i|}v_i^{t+1}\right)\\
		&=\sum_{n=1}^{N}\sum_{i\in C_n}\frac{p_i}{|S_i|} v_i^{t+1}\\&=\sum_{i\in C_1}\frac{p_i}{|S_i|} v_i^{t+1}+...+\sum_{i\in C_N}\frac{p_i}{|S_i|} v_i^{t+1},
	\end{split}
\end{equation}
where $\stackrel{(a)}{=}$ comes from the {client aggregation}, $\stackrel{(b)}{=}$ comes from the CS broadcasting, $\stackrel{(c)}{=}$ comes from the {cloud aggregation}, and $\stackrel{(d)}{=}$ comes from the {edge aggregation}.
For all $i \in \{1, \dots, K\}$, the term $\frac{p_i}{|S_i|} v_i^{t+1}$ would appear $|S_i|$ times in $\sum_{i \in C_1} \frac{p_i}{|S_i|} v_i^{t+1} + \cdots + \sum_{i \in C_N} \frac{p_i}{|S_i|} v_i^{t+1}$, and thus we have
\begin{equation}\label{Eq:wforMultiEG2}
	\begin{split}
		&\sum_{i \in C_1} \frac{p_i}{|S_i|} v_i^{t+1} + \cdots + \sum_{i \in C_N} \frac{p_i}{|S_i|} v_i^{t+1}\\
		&=\sum_{i=1}^{K}\frac{1}{|S_i|} p_i v_i^{t+1} |S_i|=\sum_{i=1}^{K} p_i v_i^{t+1}.
	\end{split}
\end{equation}
By combining \eqref{Eq:wforMultiEG1} with \eqref{Eq:wforMultiEG2}, we can conclude 
\begin{equation}\label{Eq:case2}
	\begin{split}
		w_{i}^{t+1}=\sum_{i=1}^{K} p_i v_i^{t+1}.
	\end{split}
\end{equation}

\noindent \textbf{Case 3}: $E\mid t+1$ and $EG\nmid t+1$. At this step, {edge aggregation} and {client aggregation} are performed sequentially, whereas {cloud aggregation} does not occur.  Since $w_{i}^{t+1}=\frac{1}{|S_i|}\sum_{n\in S_i}{w_{(n)}^{t+1}}$ and  $w_{(n)}^{t+1}=\sum_{i\in C_n}\frac{1}{\lozenge_n}\frac{p_i}{|S_i|}v_i^{t+1}$, we have

\begin{equation}\label{case3}
	\begin{split}
		w_{i}^{t+1}=\frac{1}{|S_i|}\sum_{n\in S_i}\sum_{i\in C_n}\frac{1}{\lozenge_n}\frac{p_i}{|S_i|}v_i^{t+1}.
	\end{split}
\end{equation}
In summary, the model of client $i$ at step $t+1$ is represented as follows:
\begin{equation}\label{Eq:paraforwMsc}
	\begin{split}
		w_i^{t+1} =
		\begin{cases}
			v_i^{t+1}, & \text{if } E\nmid t+1\\
			\sum_{i=1}^{K}p_iv_i^{t+1}, & \text{if } EG\mid t+1\\
			\frac{1}{|S_i|}\sum_{n\in S_i}\sum_{i\in C_n}\frac{1}{\lozenge_n}\frac{p_i}{|S_i|}v_i^{t+1}, & \text{otherwise}\\
		\end{cases}
	\end{split}
\end{equation}
which characterizes the second transition $\mathbf{v}^{t+1} \rightarrow \mathbf{w}^{t+1}$.
Now, the evolution path $\mathbf{w}^t \rightarrow \mathbf{w}^{t+1}$ can be clearly obtained.
\subsubsection{Virtual Global Model and Client Contributions} \label{fairness}
Since the global model is updated only once every $EG$ local training steps, $\mathbf{w}^t$ cannot reflect the training performance at step $t$ in real time. To address this issue, we define a virtual global model $\overline{w^t}$, which is obtained by hypothetically performing {edge aggregation} followed by {cloud aggregation} over all client models.
This virtual global model can be expressed as:

\begin{equation}\label{Eq:virw}
	\begin{split}
		\overline{w^t} = \sum_{n=1}^{N} \left( \lozenge_n \sum_{i \in C_n} \frac{1}{\lozenge_n} \frac{p_i}{|S_i|} w_i^{t} \right).
	\end{split}
\end{equation}
Furthermore, based on \eqref{Eq:wforMultiEG1} and~\eqref{Eq:wforMultiEG2}, the above $\overline{w^t}$ can be rewritten  as follows:
\begin{equation}\label{Eq:virw2}
	\begin{split}
		\overline{w^t} = \sum_{i=1}^K p_i w_i^t.
	\end{split}
\end{equation}
Similar to  $\overline{w^t}$, we define model $\overline{v^t}$ corresponding to $\mathbf{v}^t$ as:
\begin{equation}
	\begin{split}
		\overline{v^t} = \sum_{i=1}^K p_i v_i^t.
	\end{split}
\end{equation}
{\hlb It can be seen that each client has the same overall contribution weight in the global model, which is crucial for maintaining client-level fairness and stable training performance. In particular, clients in overlapping areas are not made more influential in the global optimization simply because they participate in multiple edge aggregations. This prevents their updates from being overweighted and helps avoid amplifying any adverse impact caused by potential data bias on those clients.}
\subsubsection{Upper Bound for the Convergence} \label{bound}
We first introduce four assumptions commonly used in FL \cite{han2021fedmes,fang2024submodel,li2019convergence,fang2022communication,wang2019adaptive,nguyen2020fast,amiri2021convergence,reisizadeh2020fedpaq} and distributed optimization \cite{zhang2012communication,stich2018local}: 

\begin{itemize}
	\item \textit{Assumption 1:} For all $i\in\{1,...,K\}$, $F_i(w)$ is a $L$-smooth function, i.e., for any $\mathbf w$ and $\mathbf {w'}$, $F_i(\mathbf{w}) \leq F_i(\mathbf{w}') + (\mathbf{w} - \mathbf{w}') \nabla F_i(\mathbf{w}') + \frac{L}{2} \| \mathbf{w} - \mathbf{w}' \|^2$.
	
	\item \textit{Assumption 2:} For all $i\in\{1,...,K\}$, $F_i(w)$ is a $\mu$-strongly convex function, i.e., for any $\mathbf w$ and $\mathbf {w'}$,
	$F_i(\mathbf{w}) \geq F_i(\mathbf{w}') + (\mathbf{w} - \mathbf{w}') \nabla F_i(\mathbf{w}') + \frac{\mu}{2} \| \mathbf{w} - \mathbf{w}' \|^2$.
	
	\item \textit{Assumption 3:} For all $i\in\{1,...,K\}$, the variance of stochastic gradients in each client is bounded, i.e., for any $\mathbf \xi_i^t\sim \mathbf D_i$, $\mathbb{E} \left\| \nabla F_i\left( \mathbf w_i^t, \mathbf \xi_i^t \right) - \nabla F_i\left( \mathbf w_i^t,\mathbf D_i \right) \right\|^2 \leq \sigma_i^2$.
	
	\item \textit{Assumption 4:} For all $i\in\{1,...,K\}$, the expected squared norm of stochastic gradients is uniformly bounded, i.e., for any $\mathbf \xi_i^t\sim \mathbf D_i$ and $\mathbf w_i^t$, $\mathbb{E} \left\| \nabla F_i\left( \mathbf w_i^t, \mathbf \xi_i^t \right) \right\|^2 \leq H^2 \quad$.
\end{itemize}

With these $4$ assumptions and $2$ definitions in place, we now proceed to present the following conclusions:

\begin{lemma}[Lemma 1 in \cite{li2019convergence}]\label{lemma1}
	Let Assumptions 1 and 2 hold, and $L$, $\mu$ be defined therein. If $\eta_t\leq\frac{1}{4L}$, we have
	\begin{equation}\label{Evminusw}
		\begin{split}
			\mathbb{E} \left\| \overline{{v}^{t+1}} - w^\star \right\|^2
			&\leq (1 - \eta_t \mu) \mathbb{E} \left\| \overline{{w}^t} - w^\star \right\|^2
			+ \eta_t^2 \mathbb{E} \left\| g_t - \overline{{g}^t} \right\|^2 \\
			&\;\;\;\;+ 6L \eta_t^2 \Gamma
			+ 2 \mathbb{E} \sum_{i=1}^K p_i \left\| \overline{{w}^t} - w_i^t \right\|^2,
		\end{split}
	\end{equation}
	where $w^\star$ denotes the optimal global model corresponding to the objective function $F(w)=\sum_{i=1}^K p_i F_i (w)$, $\overline{{g}^t} = \sum_{i=1}^K p_i \nabla F_i(w_i^t,D_i)$ is the full-batch gradient, $g^t = \sum_{i=1}^K p_i \nabla F_i(w_i^t, \xi_i)$ is the stochastic gradient, and   $\Gamma = F^\star - \sum_{k=1}^N p_k F_k^\star\geq 0$ is the optimality gap with $F^\star$ is the minimum value of function $F$ and $F_k^\star$ is the minimum value of function $F_k$ \cite{li2019convergence}.
\end{lemma}

\begin{lemma}\label{lemma:v=w}
	In our proposed algorithm, for any $t \in \{1, \ldots, T\}$, we have $\overline{v^{t}} = \overline{w^{t}}$.
	
	{\hlb	\textit{Proof:} See Appendix~A.}
	~\hfill$\blacksquare$
\end{lemma}

\begin{corollary} \label{corollary1}
	Let Assumptions 1 and 2 hold, and $L$, $\mu$ be defined therein. Let $\Delta_t$ be defined by $\mathbb{E} \left\| \overline{{w}^{t}} - w^\star \right\|^2$. If $\eta_t\leq\frac{1}{4L}$, we have
	\begin{equation}\label{EforDelta}
		\begin{split}
			\Delta_{t+1}
			&\leq (1 - \eta_t \mu) \Delta_{t}
			+ \eta_t^2 \mathbb{E} \left\| g_t - \overline{{g}^t} \right\|^2 \\
			&\;\;\;\;+ 6L \eta_t^2 \Gamma
			+ 2 \mathbb{E} \sum_{i=1}^K p_i \left\| \overline{{w}^t} - w_i^t \right\|^2,
		\end{split}
	\end{equation}
	
	\textit{Proof:} Based on Lemma \ref{lemma:v=w}, we have $\overline{v^{t}} = \overline{w^{t}}$, and thus $\mathbb{E} \left\| \overline{{v}^{t+1}} - w^\star \right\|^2 = \mathbb{E} \left\| \overline{{w}^{t+1}} - w^\star \right\|^2 = \Delta_{t+1}$.
	According to (\ref{Evminusw}) in Lemma \ref{lemma1}, it is clear that $\Delta_{t+1}
	\leq (1 - \eta_t \mu) \Delta_{t}
	+ \eta_t^2 \mathbb{E} \left\| g_t - \overline{{g}^t} \right\|^2
	+ 6L \eta_t^2 \Gamma
	+ 2 \mathbb{E} \sum_{i=1}^K p_i \left\| \overline{{w}^t} - w_i^t \right\|^2$.
	~\hfill$\blacksquare$
\end{corollary}

\noindent \textbf{Remark 1.} For $\Delta_{t+1}$, it is bounded by four terms: $\Delta_{t}$, $\mathbb{E} \left\| g_t - \overline{g^t} \right\|^2$, $\Gamma$, and $\mathbb{E} \sum_{i=1}^K p_i \left\| \overline{w^t} - w_i^t \right\|^2$. The first term, $\Delta_{t}$, represents the overall training outcome at the previous step. The second term, $\mathbb{E} \left\| g_t - \overline{g^t} \right\|^2$, indicates the data heterogeneity within each client. The third term, $\Gamma$, reflects the data heterogeneity across different clients,  which is associated with the non-IID property. The final term, $\mathbb{E} \sum_{i=1}^K p_i \left\| \overline{w^t} - w_i^t \right\|^2$, is determined by the specific topology and parameters of the training model.

\begin{lemma}\label{lemma3}
	Let Assumption 4 holds and $H$ be defined therein. Assume $\eta_{t}$ is non-increasing and $\eta_{t}\leq2 \eta_{t+E}$, we have
	\begin{equation}
		\begin{split}\label{Eq:coefficient}
			&\mathbb{E} \sum_{i=1}^K p_i \left\| \overline{w^t} - w_i^t \right\|^2\\
			&\leq\eta_{t}^2 4^{G}H^2(GE+G-2)\left[(GE-1)+E^2(G-1)\right],
		\end{split}
	\end{equation}
	
	{\hlb	\textit{Proof:}  See Appendix~B.}
	~\hfill$\blacksquare$
\end{lemma}
\noindent \textbf{Remark 2.} For the term $\mathbb{E} \sum_{i=1}^K p_i \left\| \overline{w^t} - w_i^t \right\|^2$, its upper bound is mainly determined by three factors: the topology of the model, the data distribution across different clients, and the current step $t$. 

\begin{lemma}\label{lemma4}
	Assume Assumption 3 holds. It follows that
	\begin{equation}
		\begin{split}
			\mathbb{E} \left\| g_t - \overline{{g}^t} \right\|^2\leq\sum_{i=1}^K {p_i}^2 \sigma_i^2.
		\end{split}
	\end{equation}
	
	\textit{Proof:} The proof is given in \cite{li2019convergence}.
	~\hfill$\blacksquare$
\end{lemma}

\begin{corollary} \label{corollary2}
	Let Assumptions 1 to 4 hold, and $L$, $\mu$, $\sigma_i^2$ and $H$ be defined therein. Assume $\eta_{t}$ is non-increasing, $\eta_t\leq\frac{1}{4L}$, and $\eta_{t}\leq2 \eta_{t+E}$. Let $\Delta_t$ be defined by $\mathbb{E} \left\| \overline{{w}^{t}} - w^\star \right\|^2$, then we have
	\begin{equation} \label{Eq:corollary}
		\begin{split}
			\Delta_{t+1}
			&\leq (1 - \eta_t \mu) \Delta_{t}
			+ \eta_t^2 X,
		\end{split}
	\end{equation}
	where $ X = \sum_{i=1}^K p_i^2 \sigma_i^2 + 6L(F^\star - \sum_{k=1}^N p_k F_k^\star) + 2^{2G+1}H^2(GE+G-2)[(GE-1)+E^2(G-1)]$.
	
	{\hlb	\textit{Proof:} See Appendix~C.~\hfill$\blacksquare$}
\end{corollary}
\begin{theorem} \label{Th:theorem1}
	Let Assumptions 1 to 4 hold, and $L$, $\mu$, $\sigma_i^2$ and $H$ be defined therein. Set the learning rate $\eta_t=\frac{\beta}{t+\alpha}$, where $\beta>\frac{1}{\mu}$, $\alpha\geq E>0$, and $\eta_0\leq\min\{\frac{1}{\mu},\frac{1}{4L}\}$. Then, we have
	\begin{equation}\label{convergence:HHFL}
		\begin{split}
			\mathbb{E}[F(\overline{w^t})] - F^* \leq \frac{L}{2} \frac{Z}{t+\alpha },
		\end{split}
	\end{equation}
	where $Z=\max\{\frac{\beta^2 X}{\beta\mu+1},\Delta_0 \alpha\}$ {\hlb with $X = \sum_{i=1}^K p_i^2 \sigma_i^2 + 6L(F^\star - \sum_{k=1}^N p_k F_k^\star) + 2^{2G+1}H^2(GE+G-2)[(GE-1)+E^2(G-1)]$ and $\Delta_0= \left\| {{w}^{0}} - w^\star \right\|^2$.} 
	
	{\hlb	\textit{Proof:} See Appendix~D.~\hfill$\blacksquare$}
	
\end{theorem}

\noindent \textbf{Remark 3.} As $t$ increases, $\mathbb{E}[F(\overline{w}^{\,t})]-F^\star$ gradually approaches zero. From an expectation perspective, $\overline{w}^{\,t}$ converges to the optimum, and the system performance improves accordingly.

{\hlb\noindent \textbf{Remark 4.} From the upper bound in Theorem~1, the drift-related term
	$2^{2G+1}H^2(GE+G-2)\big[(GE-1)+E^2(G-1)\big]$ in $X$
	generally increases with larger $E$ and $G$, suggesting that when edge aggregation (controlled by $E$) and cloud aggregation (controlled by $G$) are less frequent, local-update drift becomes more pronounced, thereby slowing down early-stage convergence.
	Therefore, using smaller $E$ and $G$ at the beginning of training to enable more frequent edge/cloud aggregation can help mitigate drift and accelerate performance improvement.
	As training proceeds, the model typically approaches a stationary point and the learning rate decays, so the marginal performance gain per communication round diminishes; in this regime, maintaining very frequent aggregation often yields limited additional benefit but incurs substantial communication time overhead, thus reducing the end-to-end training efficiency.
	Accordingly, gradually increasing $E$ and $G$ in later stages to lower the aggregation and communication frequency is an efficient practical strategy.}

{\hlb\noindent \textbf{Remark 5.} For HFL, when the four assumptions and the system parameters are set to be the same as those of HHFL, its convergence bound takes a similar form to (\ref{convergence:HHFL}), i.e., $\mathbb{E}[F(\overline{w^t})] - F^* \leq \frac{L}{2} \frac{Z'}{t+\alpha }$. The only difference lies in the constant term (see Appendix~E for the proof). Specifically, when the data distribution across ESs is non-IID, we show that $Z' > Z$, and the gap $Z' - Z$ increases as $E$ and $G$ grow. In other words, under non-IID ES-level data distribution, conventional HFL converges more slowly than HHFL, and the advantage of HHFL becomes more pronounced for larger $E$ and $G$. This conclusion is also consistent with the experimental results reported later (e.g., Figs.~12--13). }

{\hlb It is worth noting that, while the four assumptions are commonly adopted in the literature, they may not cover all practical scenarios, since different learning models can exhibit different properties. For example, \textit{Assumption~2} does not apply to neural networks with non-convex loss functions. Nevertheless, the insights obtained from our analysis remain valuable for more general settings. In Section~\ref{Sec:evaluation}, we further validate this point through experiments on a variety of learning models.}

\subsection{Comparison between HFL and HHFL} \label{Sec:comparison}
We compare the performance of HHFL and HFL in terms of the number of steps and the time required to reach convergence.

\subsubsection{Steps to Achieve Convergence} \label{subsec:comparison_steps}
We assume that the algorithmic parameters $E$ and $G$ are identical for both HFL and HHFL.
In HFL, each client connects to a single ES, and its updates remain confined within isolated ES domains for $EG$ local updates before {cloud aggregation} is performed. When the data distributions across ESs are non-IID, this isolation exacerbates model divergence and inconsistency across ES models as training progresses. As a result, local models initialized from different ESs become increasingly divergent over time, leading to higher variance in local gradient updates and slower convergence.

In contrast, HHFL effectively mitigates this issue. Each ES in HHFL receives updates from a larger and more diverse set of clients, making its model more representative of the global data distribution and reducing divergence among ES models. In addition, the client aggregation mechanism in HHFL enables clients located in overlapping regions to initialize their models using multiple ES models. Together, these two mechanisms reduce the variance in client initializations.
As a result, the local updates computed as
$v_i^{t+1} = w_i^t - \eta_t \nabla F_i(w_i^t)$
are based on better-informed starting points. Since $w_i^t$ incorporates more global data, the corresponding gradient $\nabla F_i(w_i^t)$ is more aligned with the global direction $\nabla F(w^t)$, which reduces the variance of local gradients, represented as
$\sigma^2 = \mathbb{E}_i \left[\|\nabla F_i(w_i^t) - \nabla F(w^t)\|^2\right].
$ {\hlb This observation is consistent with Remark~4, and more theoretical details are provided in Appendix~E.}

{\hlb According to \cite{stich2018local}, standard stochastic optimization bounds suggest that
	$\mathbb{E}[F(\overline{{w}^T)}]-F(w^\star)\leq \mathcal{O}\!\left(\frac{1+\sigma^2}{\sqrt{T}}\right)$,
	indicating that a smaller gradient variance $\sigma^2$ leads to faster convergence. Consistently, the tighter bound in (\ref{convergence:HHFL}) further supports this effect. Therefore, compared with HFL, HHFL typically requires fewer training steps to reach a target convergence level.}


\subsubsection{Time to Achieve Convergence} \label{Sec:time}
The {overall time} (i.e., the total wall-clock time) mainly involves the computing time and the wireless transmission time.
We assume that the two types of time do not overlap. This assumption is widely adopted in FL, as it helps reduce scheduling and synchronization complexity, especially on resource-constrained devices \cite{kairouz2021advances,li2020federated,bonawitz2019towards,liu2020client,zhao2018federated,wang2019adaptive,lyu2020threats,smith2017federated}.
Therefore, the overall time is the sum of the computing time and the transmission time.
For the {computing time}, it is primarily used for local updates and model aggregations. Since model aggregation only involves computing a weighted average, its computational cost is significantly lower than that of a local update. Moreover, aggregation occurs much less frequently than local updates. Therefore, it is reasonable to ignore the aggregation time and approximate the computing time using only the time spent on local updates.
Let $c$ denote the number of CPU cycles required to compute one bit, $f$ the CPU frequency, and $b_0$ the number of computed bits required for one step of a local update. Then, the local update computation time for executing $T$ updates is given by
\begin{equation}
	\begin{split}
		T_\text{comp}=\frac{cb_0}{f}T.
	\end{split}
\end{equation}
For the transmission time, the total time for $T$ local update steps can be expressed as
\begin{equation}\label{Eq: transmission}
	\begin{split}
		T_\text{trans}&=\left \lceil\frac{T}{E}\right \rceil  T_\text{CES}+\left \lceil\frac{T}{EG}\right \rceil T_\text{ECS},
	\end{split}
\end{equation}
{\hlb where $\lceil\cdot\rceil$ denotes the ceiling function. $T_{\mathrm{CES}}$ denotes the end-to-end wall-clock duration (e.g., airtime and protocol overhead) of one {ES--client} round-trip model exchange, including the ES-to-client downlink model dissemination and the client-to-ES uplink upload of updated models. $T_{\mathrm{ECS}}$ denotes the end-to-end wall-clock duration of one {CS--ES} round-trip model exchange, including the CS-to-ES global-model dissemination and the ES-to-CS uplink upload of ES models. Moreover, $\left \lceil\frac{T}{E}\right \rceil$ and $\left \lceil\frac{T}{EG}\right \rceil$ represent the numbers of {ES--client} and CS--ES round-trip exchanges required to complete $T$ training steps, respectively.}
Then, the overall time for both architectures to complete $T$ local update steps are represented as
\begin{equation}\label{Eq: overall}
	\begin{split}
		T_\text{overall}&=\frac{cb_0}{f}T+\left \lceil\frac{T}{E}\right \rceil  T_\text{CES}+\left \lceil\frac{T}{EG}\right \rceil T_\text{ECS}.
	\end{split}
\end{equation}

{ \hlb 
	In practical deployments, model dissemination from ESs to clients can be carried out in parallel over orthogonal PRBs, and a client’s model upload can be received and decoded by multiple associated ESs via CoMP-capable uplink multi-point reception. Moreover, contemporary 5G/NextG systems are engineered with substantial parallel {receive-side} baseband-processing capability at both BSs and user devices to decode multiple PRB allocations within the same carrier, and BSs are typically provisioned with sufficient {transmit-side} baseband-processing resources to map and process multiple PRB allocations in parallel within a carrier, thereby serving multiple users concurrently on orthogonal PRBs~\cite{standard2017final}. While such parallel processing capabilities effectively eliminate transmission bottlenecks, enabling these multi-point connections does introduce extra synchronization overhead, as ESs must coordinate PRB allocations via the Xn interface. However, this coordination cost has a negligible impact on $T_{CES}$. First, because ESs operate independently without data-level cooperation, the Xn interface exchanges only lightweight scheduling signals, thereby bypassing heavy backhaul data transfers. Second, leveraging 5G CG for PRB pre-allocation amortizes the one-time negotiation latency across the entire training process, effectively eliminating per-round dynamic scheduling delays. Consequently, under our latency model, the additional ES-client links do not necessarily increase $T_{CES}$, and since the network infrastructure remains static, $T_{\mathrm{ECS}}$ is also unaffected.
	As a result, with the same client hardware capabilities and the same $E$ and $G$, HHFL and HFL have equivalent overall training-time expressions in terms of the total number of local update steps, as shown in \eqref{Eq: overall}. As discussed in Section~\ref{subsec:comparison_steps}, HHFL requires fewer steps to reach convergence compared to HFL.
	Therefore, HHFL requires a shorter overall time to converge, making it a more efficient choice than HFL.}
{\hlb \subsubsection{Resource Consumption} Given the negligible computational cost of model parameter aggregation, the primary resource-related drawback of HHFL compared to standard HFL is the greater per-round resource consumption such as higher PRB occupancy and energy costs during model transmission between ESs and clients. This is because clients situated in overlapping areas can simultaneously connect to multiple ESs for data exchange. This creates more ES--client communication links, thereby increasing the volume of data transmitted per round and requiring more resources for transceiving and signal processing. However, as HHFL typically requires fewer training steps to reach convergence due to the knowledge-sharing benefits provided by multi-connectivity links, it yields three primary advantages: i) reduced total computational overhead for local client training; ii) fewer communication rounds between the CS and ESs; and iii) a lower total number of communication rounds between ESs and clients. While the first two advantages provide direct resource savings for the entire system including clients, ESs, and the CS, the third involves a tradeoff between the reduction in communication rounds and the increased per-round transmission overhead between ESs and clients. Notably, when the data distribution across ESs is sufficiently non-IID, the gain from the reduction in total communication rounds dominates this tradeoff, and HHFL achieves lower overall resource consumption on both the ES and client sides for the data transmission between them. For the sake of brevity, the detailed analysis of this tradeoff is relegated to Appendix~F, which derives the theoretical upper bound on the increase factor in resource consumption for per-round ES-client model transmissions under HHFL and uses our experimental results as a concrete case study to demonstrate that the reduction in communication rounds outweighs the increase in per-round resource consumption. Collectively, HHFL is demonstrated to be more resource efficient for every participating entity including clients, ESs, and the CS throughout the entire training process in this sufficiently non-IID case.

	It is noteworthy that the increase in per-round ES–client resource consumption is connectivity-dependent and remains decoupled from the ES-level non-IID degree. Consequently, when the degree of non-IID is insufficient, the marginal gain from knowledge sharing between ESs facilitated by multi-connectivity links diminishes, and the aforementioned advantages of HHFL may fail to offset the higher per-round overhead, leading to greater total resource consumption than HFL. While HHFL is primarily designed to improve time efficiency, resource constraints can sometimes become a critical factor to consider in certain deployment scenarios. 
	To optimize practical resource efficiency in such cases, some active multi-connectivity links can be strategically deactivated to reduce the additional per-round communication overhead for ES–client data transmissions, at the cost of some potential knowledge-sharing gains. Optimizing the number and selection of multi-connectivity links to disable (e.g., by disabling a subset of links for a single client or across multiple clients and ESs) is a highly complex problem. Addressing this optimization problem requires jointly considering data distribution, network topology, time efficiency, and resource constraints, which exceeds the scope of this study. We defer a detailed investigation of resource-aware link activation/deactivation designs for HHFL to future work.
	
}


\begin{figure*}[ht]
	\centering
	\includegraphics[width=0.9\textwidth]{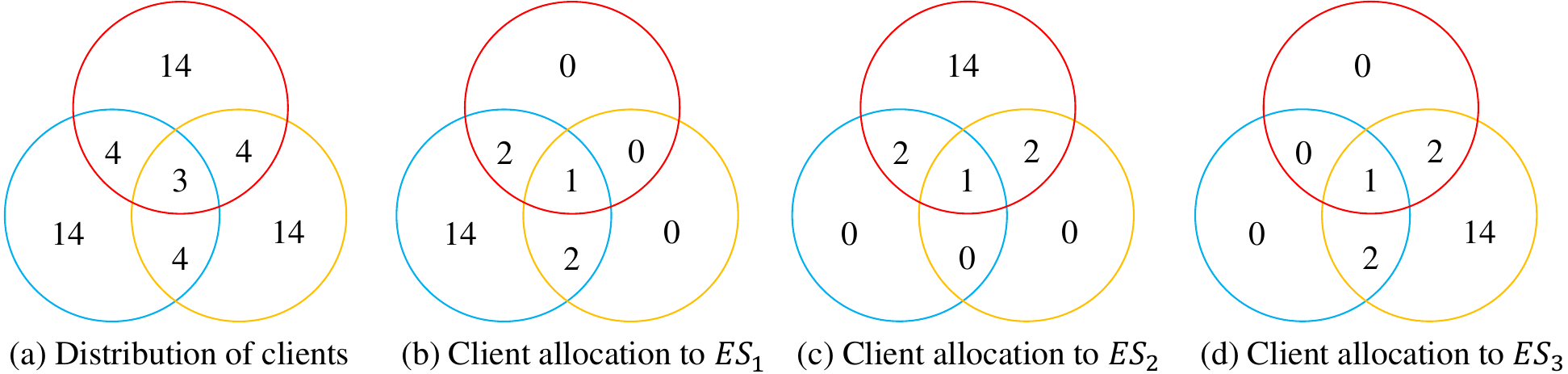}
	\caption{The actual distribution of clients and the number of clients assigned to each ES under the HFL architecture. The coverage area of $ES_1$ is represented by a blue circle, $ES_2$ by a red circle, and $ES_3$ by an orange circle. The values in this figure represent the number of clients.}
	\label{Fig:distribution}
\end{figure*}

\section{Experimental Evaluation}\label{Sec:evaluation}
In this section, we evaluate the effectiveness of HHFL and demonstrate its advantage in convergence efficiency.

\subsection{Basic Experimental Setup}
\noindent\textbf{FL Architecture Configuration:} Without loss of generality, we consider a network scenario consisting of one CS, three ESs labeled as $ES_1$, $ES_2$, and $ES_3$, and $57$ clients. The ESs are symmetrically placed at the vertices of an equilateral triangle, each serving $25$ clients: $14$ connect exclusively to a single ES, $8$ connect to two ESs, and $3$ connect to all three ESs, as illustrated in Fig.~\ref{Fig:distribution}(a). This topology can be directly adopted for the HHFL architecture. For the HFL architecture, clients located in each overlapping region are randomly and evenly assigned to one of the ESs covering that region, ensuring that each client connects to only one ES. Specifically, the distributions of clients allocated to $ES_1$, $ES_2$, and $ES_3$ in the HFL architecture are shown in Fig. \ref{Fig:distribution}(b), Fig. \ref{Fig:distribution}(c), and Fig. \ref{Fig:distribution}(d), respectively.

\noindent\textbf{FL Algorithm Configuration:}
We employ the HHFL algorithm for our proposed architecture. For the HFL architecture, we implement the efficient and widely adopted Hier-FAVG algorithm, as introduced in \cite{liu2020client}. We set $E = 5$ and $G = 5$.

\noindent\textbf{Training Model Configuration:}
To validate that our HHFL method converges effectively regardless of whether the training loss function is convex or nonconvex and does so with higher efficiency than HFL, we selected two models: a logistic regression model and a convolutional neural network (CNN). These models correspond to convex and nonconvex loss functions respectively. Both models are trained on the MNIST dataset using mini-batch SGD with a batch size of $20$. The logistic regression model's initial learning rate is $0.1$, with an exponential decay of $0.992$ per epoch. For the CNN model, to obtain a clearer convergence curve, we adopt a smaller initial learning rate of $0.02$.
Here, we intentionally adopt the simple MNIST dataset and lightweight models. This setup avoids unnecessary complexity introduced by deep models or large-scale datasets, while still providing generalizable insights into architectural performance.

\noindent\textbf{Data Distribution Configuration:} 
We consider three basic data distribution scenarios: i) \textit{client IID, ES IID} scenario, where the data distribution among clients within each ES is IID while the data distribution across ESs is also IID; ii) \textit{client non-IID, ES IID} scenario, where the data distribution among clients within each ES is non-IID while the data distribution across ESs is IID; iii) \textit{client non-IID, ES non-IID} scenario, where both the data distributions among clients within each ES are non-IID while the data distribution across ESs is also non-IID. Since the non-IID setting in this work requires considering IID conditions at the ES level, the conventional Dirichlet-based distribution is not well-suited to meet this requirement. Therefore, we adopt the pathological non-IID setting \cite{mcmahan2017communication} following the setup in \cite{han2021fedmes}. 


\noindent\textbf{Computation and Communication Configuration:} We set the ratio $T_\text{CES}/T_\text{SGD} = 10$ to reflect the communication bottleneck, where $T_\text{SGD}$ denotes the local training time for $E$ local updates, following the experimental setup in \cite{han2021fedmes,reisizadeh2020fedpaq}. In 5G/NextG networks, the connection between base stations and the cloud server station is wired and significantly faster than wireless transmission \cite{qi2024bridge}. Real-world measurements show that the wireless latency between clients and ESs typically ranges from 5 to 20 ms \cite{fezeu2023mmwave}, whereas the wired fiber latency between ESs and the CS is typically between 0.2 and 1 ms \cite{agiwal2016next}. Based on these observations, it is reasonable to set the ratio $T_\text{CES} / T_\text{ECS} = 10$ without loss of generality.

\subsection{Experimental Results and Analysis}

\subsubsection{Impact of intra-ES non-IID data distribution on convergence under \textit{ES IID}  scenario} We compare the convergence performance of the HFL and HHFL architectures under different degrees of non-IID data distribution among clients within each ES, while keeping the data distribution across ESs IID. 
To this end, we design three cases: the \textit{(IID, IID)} case, the \textit{(non-IID-$1$, IID)} case, and the \textit{(non-IID-$2$, IID)} case. In the first case, each client uniformly contains samples from all $10$ classes. In the second case, each client uniformly contains samples from $6$ classes, and in the third case, each client uniformly contains samples from only $2$ classes. For all three cases, the data distribution across ESs remains IID. 
The convergence curves over training steps are shown in  Fig.~\ref{fig:es_iid}. It is clear that the two convergence curves corresponding to HFL and HHFL in the same case almost show no difference. One intuitive explanation for this is that when the data distribution across ESs is IID, no cross-ES knowledge sharing is needed, and clients in overlapping regions would offer no incremental utility.
Another side observation is that the convergence speed in the \textit{(non-IID-$2$, IID)} case is significantly slower than in the other two cases. This is because the local gradient variance across clients becomes much larger in this case, as each client only contains samples from two classes.


\begin{figure*}[ht]
	\centering
	\begin{minipage}[t]{0.48\textwidth}
		\centering
		\begin{subfigure}[b]{0.48\textwidth}
			\includegraphics[width=\textwidth]{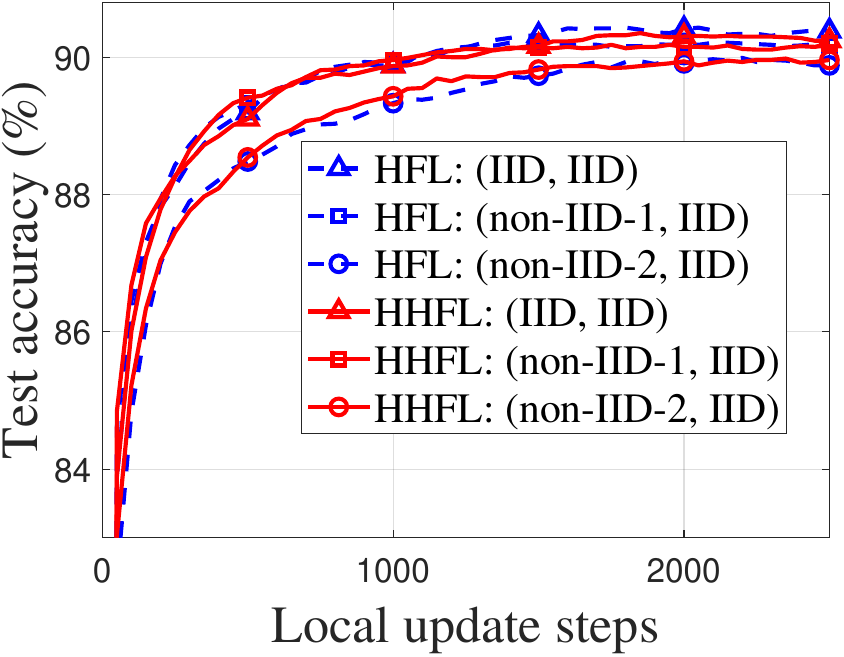}
			\caption*{(a) Convex loss function}
		\end{subfigure}
		\hfill
		\begin{subfigure}[b]{0.48\textwidth}
			\includegraphics[width=\textwidth]{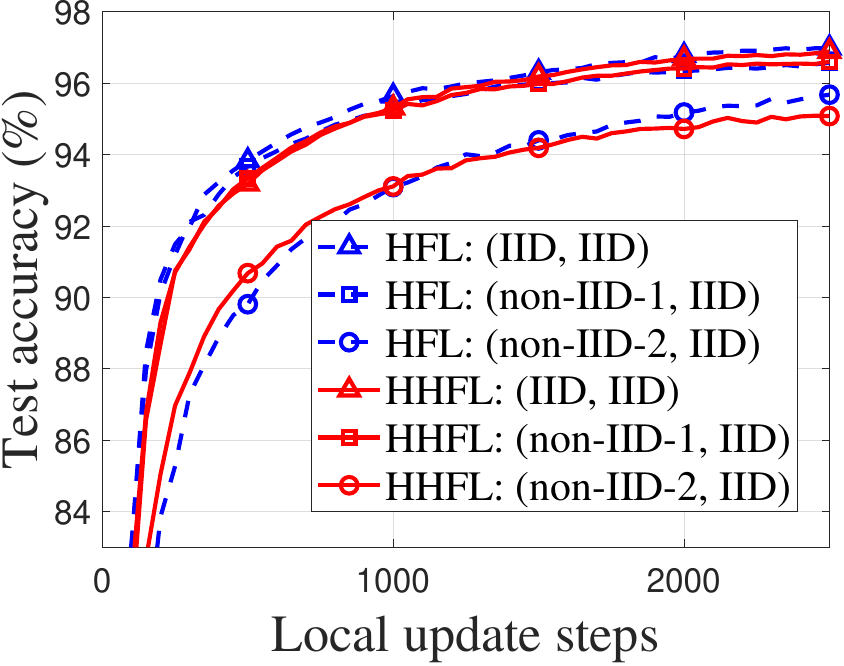}
			\caption*{(b) Non-convex loss function}
		\end{subfigure}
		\caption{Convergence curves under different degrees of non-IID data distribution among clients within each ES ( \textit{ES IID} scenario).}
		\label{fig:es_iid}
	\end{minipage}
	\hfill
	\begin{minipage}[t]{0.48\textwidth}
		\centering
		\begin{subfigure}[b]{0.48\textwidth}
			\includegraphics[width=\textwidth]{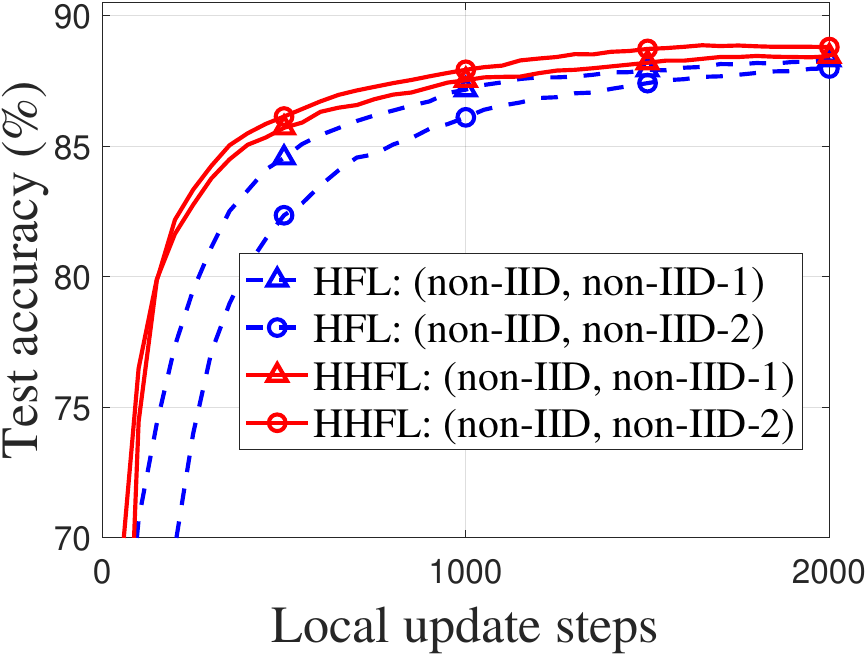}
			\caption*{(a) Convex loss function}
		\end{subfigure}
		\hfill
		\begin{subfigure}[b]{0.48\textwidth}
			\includegraphics[width=\textwidth]{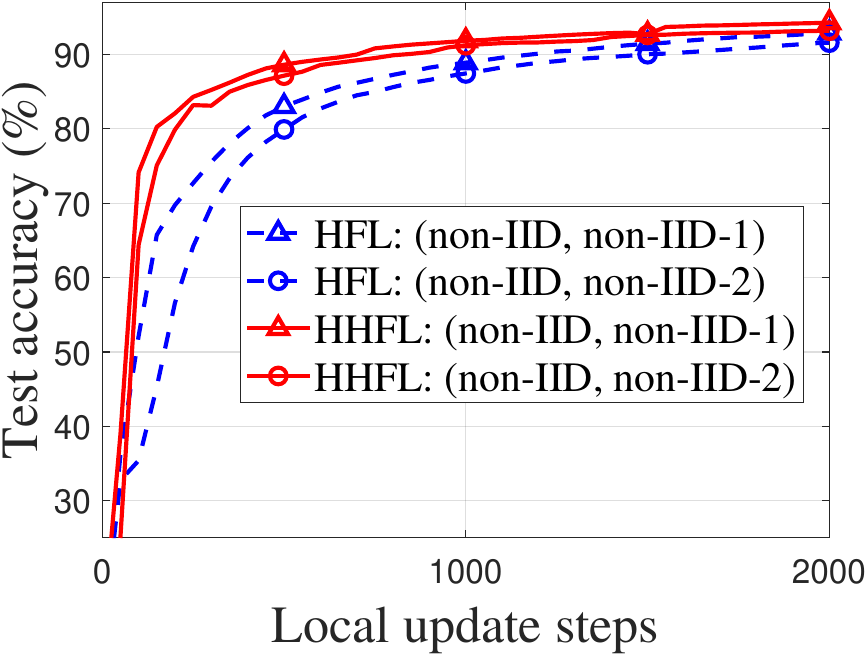}
			\caption*{(b) Non-convex loss function}
		\end{subfigure}
		\caption{Convergence curves under different degrees of non-IID data distribution across ESs (\textit{ES non-IID scenario}).}
		\label{fig:es_noniid}
	\end{minipage}
\end{figure*}

\subsubsection{Impact of ES non-IID data distribution on convergence under \textit{ES non-IID} scenario}
We compare the convergence performance of the HFL and HHFL architectures under different degrees of non-IID data distribution across ESs. Here, two cases are designed: the \textit{(non-IID, IID-$1$) case} and the \textit{(non-IID, IID-$2$) case}. In the first case, each ES is missing $3$ data classes, whereas in the second case, each ES is missing $4$ classes. The data across ESs is more non-IID in the second case than in the first.
The convergence curves over training steps are shown in Fig.~\ref{fig:es_noniid}. It can be observed that HHFL consistently outperforms HFL in terms of convergence performance. Moreover, as the degree of non-IID data distribution across ESs increases, the performance of HHFL remains largely unaffected, whereas the convergence performance of HFL degrades significantly. This results in a larger performance gap between the two architectures. We attribute the advantages of HHFL to the presence of clients located in overlapping regions. These clients enable knowledge sharing between neighboring ESs, which helps mitigate the performance degradation caused by the limited class diversity within individual ESs.

{\hlb \subsubsection{Impact of proportion of clients located in overlapping regions on convergence} We compare the convergence performance of the HFL and HHFL architectures with varying proportion of clients located in overlapping regions. Here, we adjust this proportion by simulating user mobility. Starting from our original deployment, we increase this proportion by relocating some single-coverage clients into overlapping areas; conversely, we decrease it by relocating some overlapping-area clients into single-coverage areas. To ensure that the overlap proportion is the primary change across configurations and to avoid unintended effects caused by ES reassignment (which may alter the ES-level non-IIDness in HFL), we keep the client--ES association in the HFL baseline unchanged when generating different overlap settings. Specifically: (i) when a client is moved from a single-coverage area to an overlapping area, it remains associated with its pre-relocation ES in HFL; and (ii) when a client is moved from an overlapping area to a single-coverage area, it is placed within the single-coverage region of its originally assigned ES in HFL. 
	
	To visualize the impact of changing the overlapping proportion on convergence curves, we first construct an illustrative case based on the \textit{(non-IID, non-IID-$2$)} case. Specifically, we relocate $6$ single-coverage clients into the overlapping areas to get an additional case, called \textit{(non-IID, non-IID-$2'$)} case. The resulting client distribution and the corresponding ES associations are illustrated in Fig.~\ref{Fig:distribution_case6}.
	The convergence curves over training steps are shown in Fig.~\ref{fig:number_overlappingclients}. 
	We observe that HHFL converges faster than HFL under the same case. Moreover, increasing the number of clients in overlapping regions improves the convergence speed of HHFL, although this acceleration is modest due to the small number of relocated clients in our setup, whereas HFL's convergence speed remains largely unaffected. We believe this result comes from the fact that, in the HHFL architecture, a greater number of overlapping clients can enhance indirect knowledge sharing among neighboring ESs, whereas in the HFL architecture, the client assignments at each ES remain unchanged despite the relocation of clients.
	\begin{figure*}[ht]
		\centering
		\includegraphics[width=0.9\textwidth]{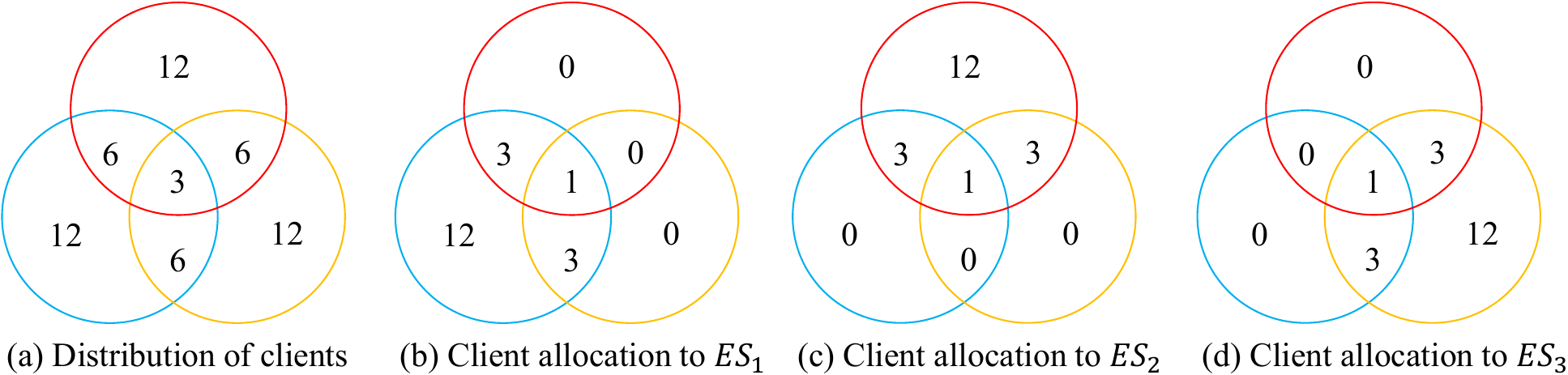} 
		\caption{Client actual distribution and the number of clients assigned to each ES for the \textit{(non-IID, IID-$2'$)} case. The coverage area of $ES_1$ is represented by a blue circle, $ES_2$ by a red circle, and $ES_3$ by an orange circle. The values in this figure represent the number of clients.
		}
		\label{Fig:distribution_case6}
	\end{figure*}
	
	\begin{figure*}[ht]
		\centering
		\begin{minipage}[t]{0.48\textwidth}
			\centering
			\begin{subfigure}[b]{0.48\textwidth}
				\includegraphics[width=\textwidth]{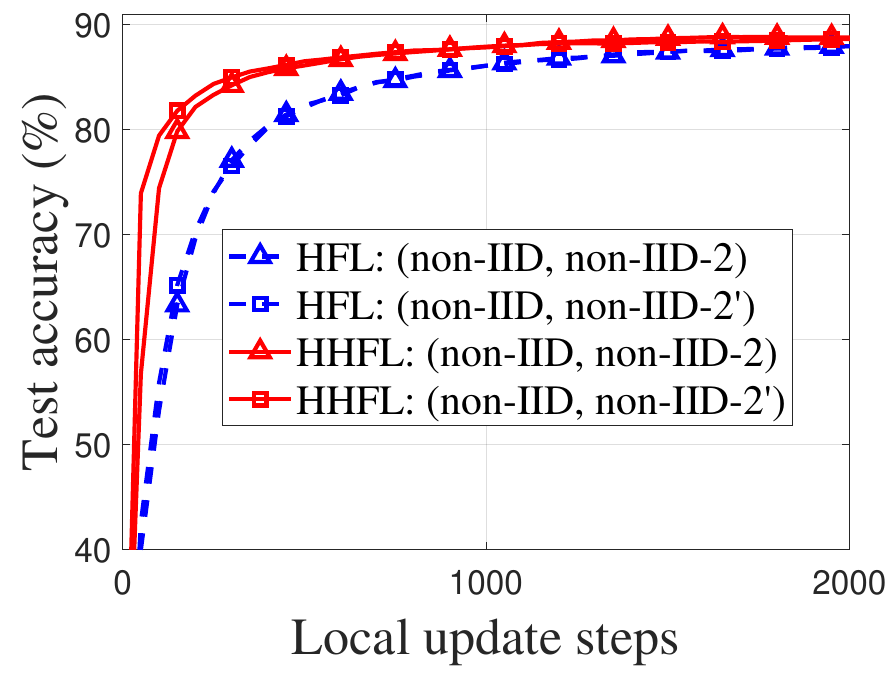}
				\caption*{(a) Convex loss function}
			\end{subfigure}
			\hfill
			\begin{subfigure}[b]{0.48\textwidth}
				\includegraphics[width=\textwidth]{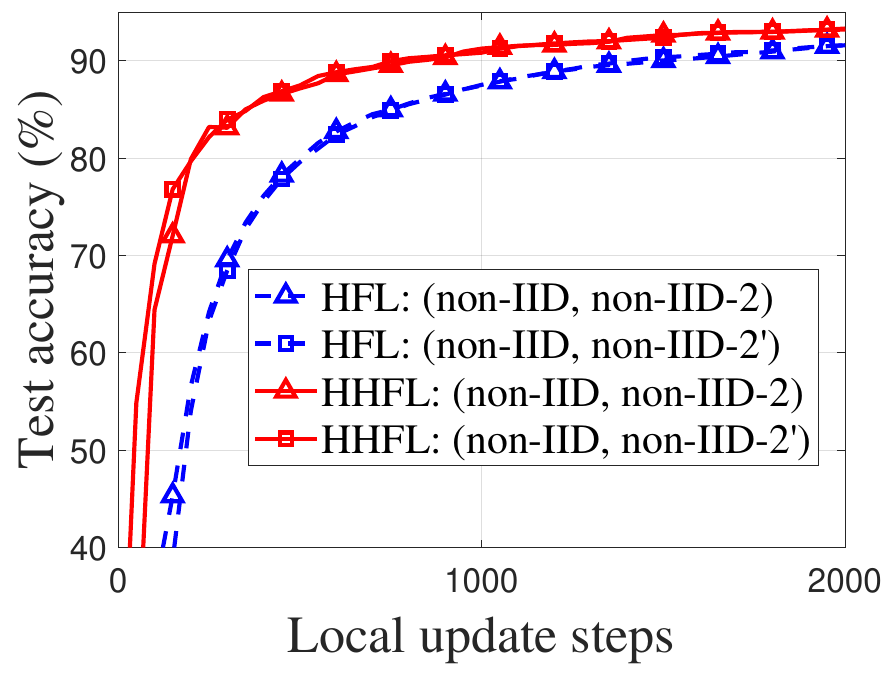}
				\caption*{(b) Non-convex loss function}
			\end{subfigure}
			\caption{Convergence curves under different numbers of clients in overlapping regions (\textit{ES non-IID} scenario).}
			\label{fig:number_overlappingclients}
		\end{minipage}
		\hfill
		\begin{minipage}[t]{0.48\textwidth}
			\centering
			\begin{subfigure}[b]{0.47\textwidth}
				\includegraphics[width=\textwidth]{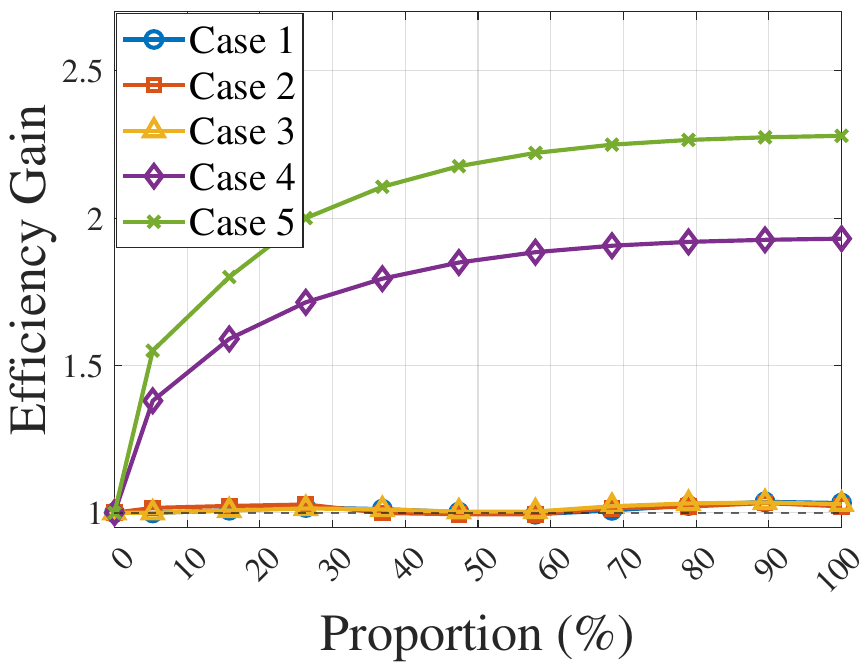}
				\caption*{(a) Convex loss function}
			\end{subfigure}
			\hfill
			\begin{subfigure}[b]{0.48\textwidth}
				\includegraphics[width=\textwidth]{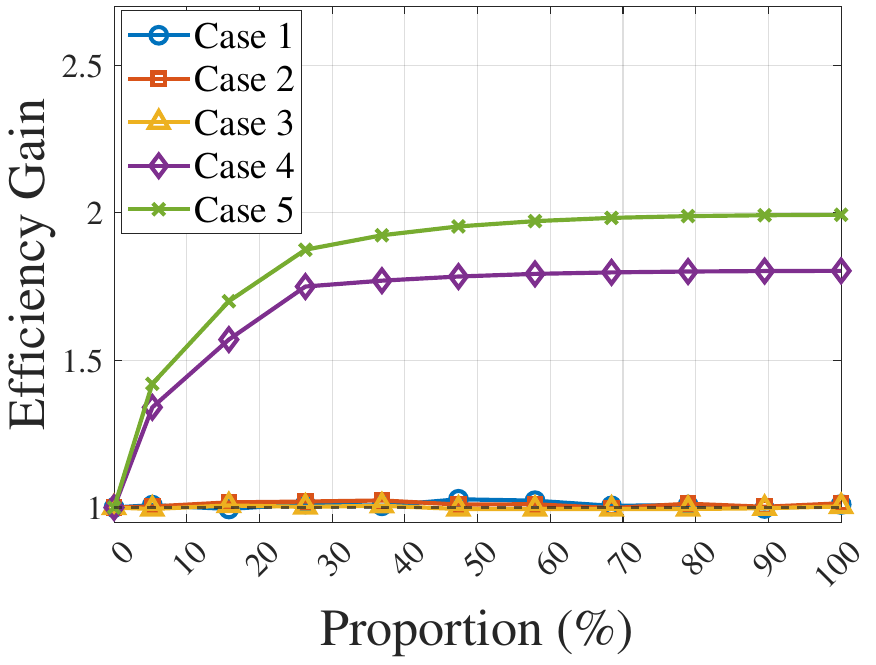}
				\caption*{(b) Non-convex loss function}
			\end{subfigure}
			\caption{{\hlb Efficiency gain of HHFL over HFL versus the overlapping proportion for Cases~1--5:
					\textit{(IID, IID)}, \textit{(non-IID-$1$, IID)}, \textit{(non-IID-$2$, IID)},
					\textit{(non-IID, non-IID-$1$)}, and \textit{(non-IID, non-IID-$2$)}.}
			}
			\label{fig:proportion}
		\end{minipage}
	\end{figure*}
	While the above comparison provides an intuitive understanding, we next conduct a systematic sweep over the overlap proportion. Specifically, we vary the overlap proportion by relocating additional clients into or out of the overlapping areas, and for each proportion value we run both HFL and HHFL under five basic cases. We then quantify HHFL’s advantage using an efficiency-gain metric in terms of the required training steps. To determine the convergence point, we adopt a generalizable approach based on the sliding average of the slope. A threshold of $0.1\%$ is used, such that when the average rate of improvement in accuracy falls below this value at a particular step, the model is regarded as having reached convergence. The resulting Efficiency Gain VS. Proportion curve is shown in Fig.~\ref{fig:proportion}. We observe that when the data distributions across ESs are IID, namely cases 1--3, the efficiency gain remains close to $1$ regardless of the overlap proportion, indicating that HHFL’s multi-connectivity offers almost no speedup over HFL. This is because, under ES-level IID data, HFL does not need the overlapping-area clients to facilitate inter-ES knowledge sharing. In contrast, when the data distributions across ESs are non-IID, namely cases 4--5, HHFL’s efficiency gain over HFL increases with the overlap proportion. The improvement is pronounced when the overlap proportion is small, and then gradually saturates as the proportion becomes larger. We believe this behavior arises because a small number of overlapping clients can already provide much-needed bridging for inter-ES knowledge sharing; once this bridging effect becomes sufficiently strong, additional overlapping clients yield diminishing marginal gains, leading to a saturation trend. 
}
\begin{figure*}[ht]
	\centering
	\begin{minipage}[t]{0.48\textwidth}
		\centering
		\begin{subfigure}[b]{0.48\textwidth}
			\includegraphics[width=\textwidth]{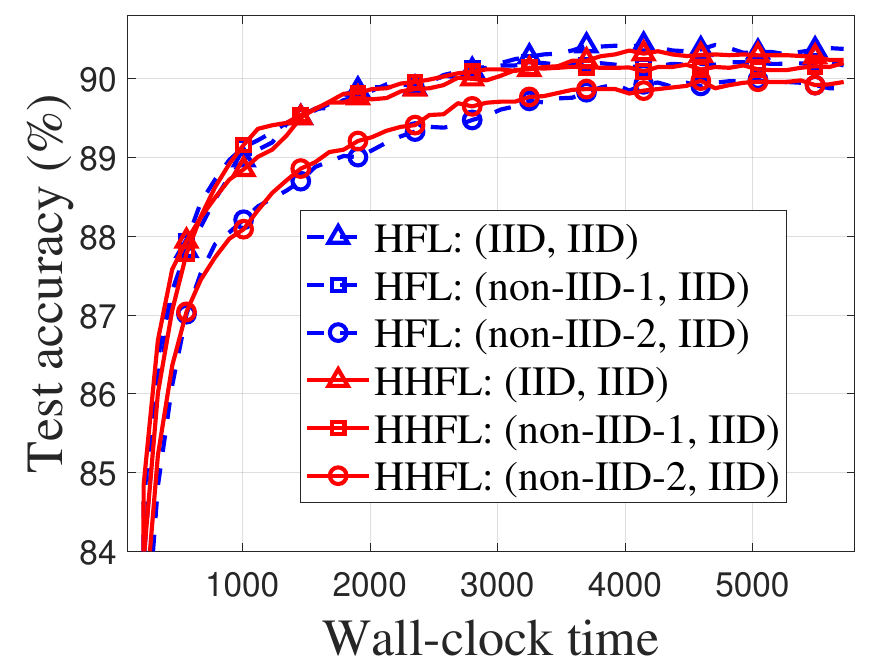}
			\caption*{(a) \textit{ES IID} scenario, convex loss function}
		\end{subfigure}
		\hfill
		\begin{subfigure}[b]{0.48\textwidth}
			\includegraphics[width=\textwidth]{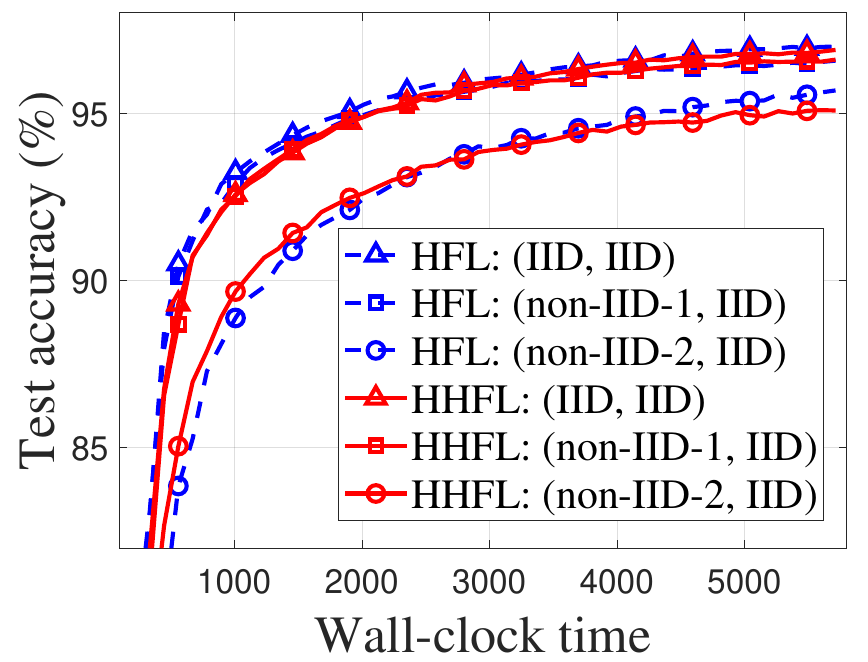}
			\caption*{(b) \textit{ES IID} scenario, non-convex loss function}
		\end{subfigure}
		\hfill
		\begin{subfigure}[b]{0.48\textwidth}
			\includegraphics[width=\textwidth]{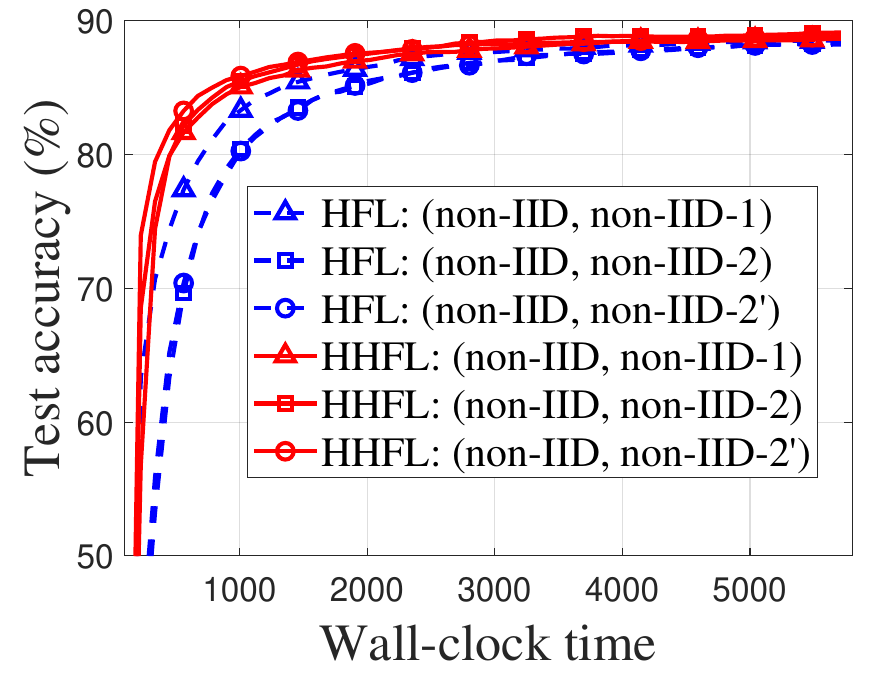}
			\caption*{(c) \textit{ES non-IID} scenario, convex loss function }
		\end{subfigure}
		\hfill
		\begin{subfigure}[b]{0.48\textwidth}
			\includegraphics[width=\textwidth]{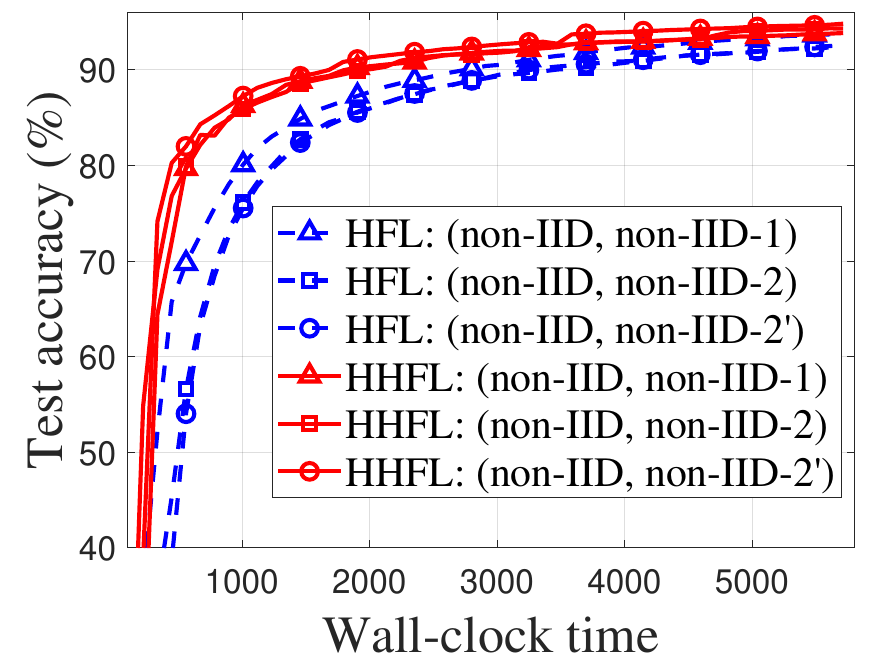}
			\caption*{ (d) {\textit{ES non-IID} scenario, non-convex loss function}\strut}
		\end{subfigure}
		\caption{{\hlb Test accuracy versus wall-clock time across six cases (time is normalized such that one unit corresponds to the time for a client to complete $E$ local updates).}}
		\label{fig:time}
	\end{minipage}
	\hfill
	\begin{minipage}[t]{0.48\textwidth}
		\centering
		\begin{subfigure}[b]{0.48\textwidth}
			\includegraphics[width=\textwidth]{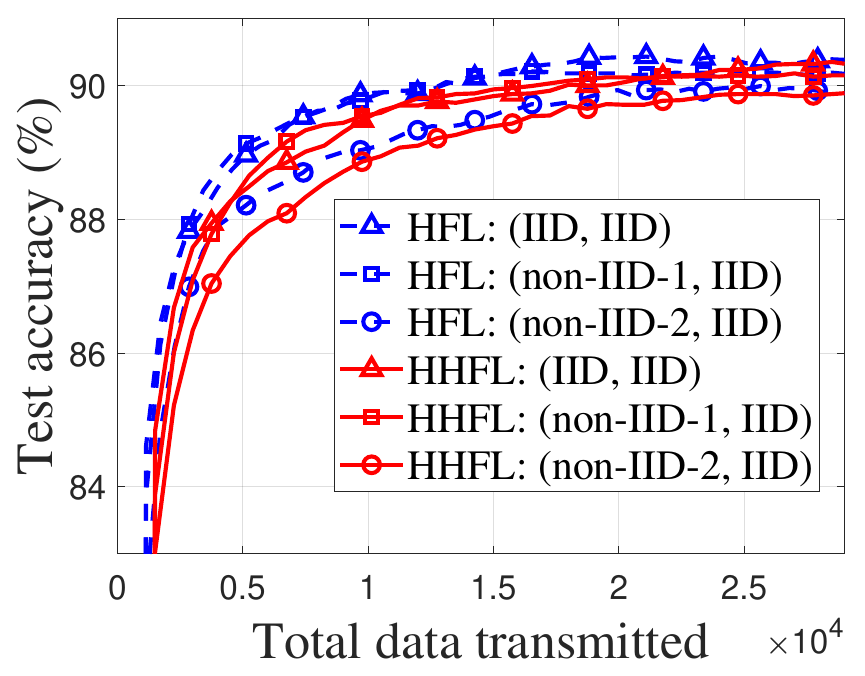}
			\caption*{(a) \textit{ES IID} scenario, convex loss function}
		\end{subfigure}
		\hfill
		\begin{subfigure}[b]{0.48\textwidth}
			\includegraphics[width=\textwidth]{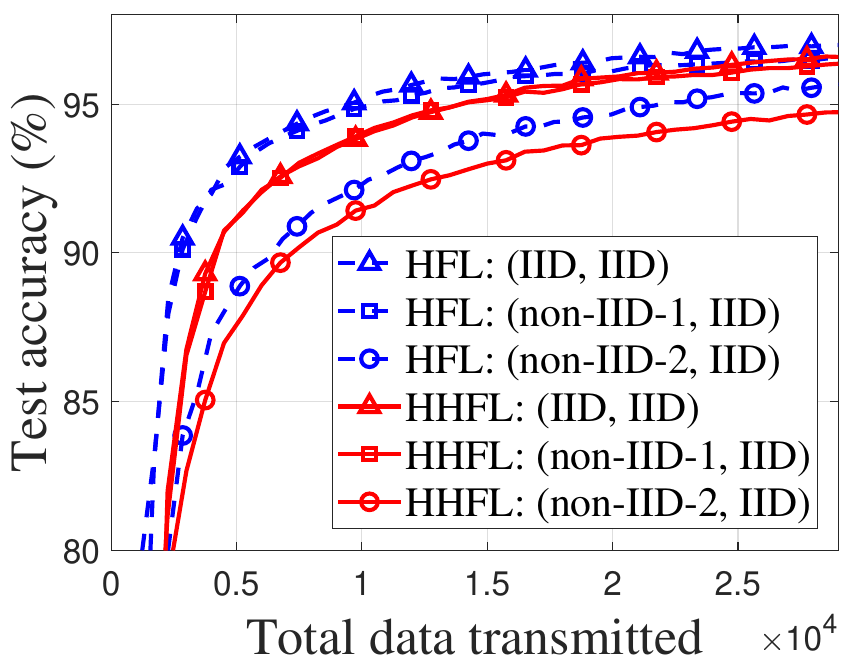}
			\caption*{(b) \textit{ES IID} scenario, non-convex loss function}
		\end{subfigure}
		\hfill
		\begin{subfigure}[b]{0.48\textwidth}
			\includegraphics[width=\textwidth]{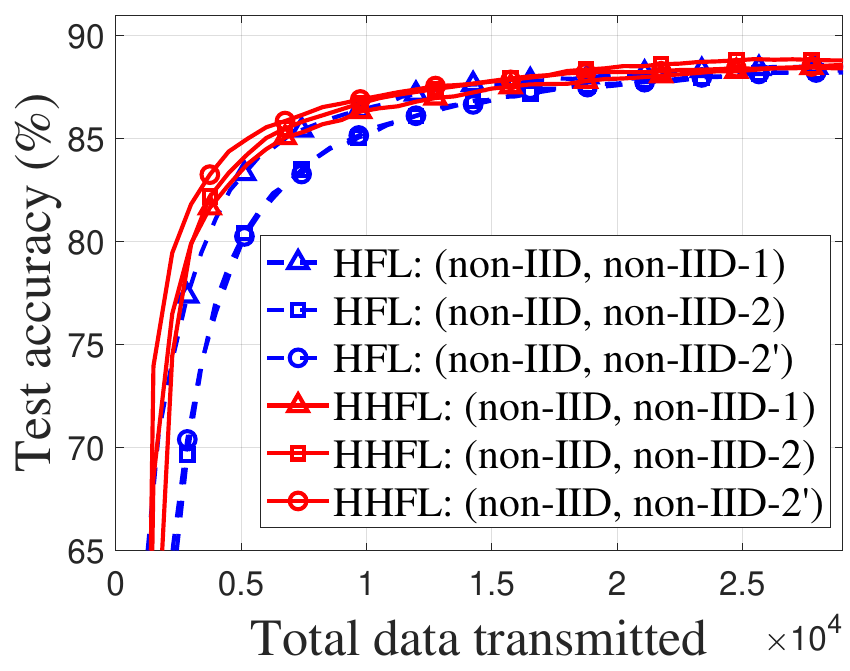}
			\caption*{(c) \textit{ES non-IID} scenario, convex loss function }
		\end{subfigure}
		\hfill
		\begin{subfigure}[b]{0.48\textwidth}
			\includegraphics[width=\textwidth]{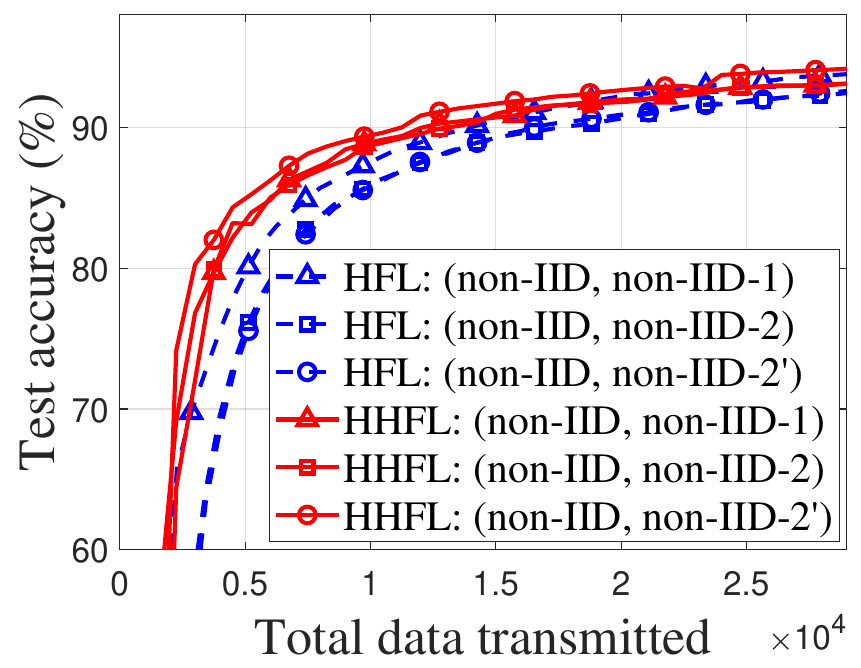}
			\caption*{ (d) {\textit{ES non-IID} scenario, non-convex loss function}\strut}
		\end{subfigure}
		\caption{{\hlb Test accuracy versus total data transmitted across six cases (data volume is normalized such that one unit corresponds to the number of bits required to transmit one model).}}
		\label{fig:data}
	\end{minipage}
\end{figure*}
{\hlb \subsubsection{Convergence curves over wall-clock time and total data transmitted} We further provide a holistic efficiency evaluation by plotting the convergence performance (i.e., test accuracy) of HFL and HHFL against (i) wall-clock time and (ii) the total data transmitted, under different data-distribution cases. 
	
	By setting the time required for a client to complete $E$ local updates as one unit time, the experimental results are presented in Fig. \ref{fig:time}. As observed in Fig. \ref{fig:time} (a) and Fig. \ref{fig:time} (b), when the data distribution across ESs is IID, HHFL does not exhibit a noticeable advantage in terms of wall-clock time. However, when the data distribution across ESs is non-IID, HHFL outperforms HFL, as shown in Fig. \ref{fig:time} (c) and Fig. \ref{fig:time} (d). This is because HHFL requires significantly fewer total training steps, which directly cuts down the time clients spend on local training. It also reduces the number of communication rounds between all participants (i.e., Client-to-ES and ES-to-CS), thereby cutting down the time each party spends on communication. Since these phases are the main components of the total time, HHFL ends up being much more efficient overall.
	
	
	By setting the number of bits required to transmit a single model as one unit of data volume, the experimental results are presented in Fig. \ref{fig:data}. Note that the transmitted data volume accounts specifically for the communication between the BS/ES and the clients. As observed in Fig. \ref{fig:data} (a) and Fig. \ref{fig:data} (b), when the data distribution across ESs is IID, HHFL converges slower than HFL. This is because HHFL establishes more communication links, which involves more data transmission without providing the compensating benefit of effective knowledge sharing in an IID setting. Conversely, as shown in Fig. \ref{fig:data} (c) and Fig. \ref{fig:data} (d), under the cases where the data distribution across ESs is non-IID, HHFL reaches convergence with significantly less total data transmission. This is because the benefits of knowledge sharing facilitated by the additional links effectively outweigh the additional data-transmission cost introduced by the extra links.
	
}

\begin{figure*}[ht]
	\centering

	\begin{minipage}[t]{0.48\textwidth}
		\centering
		\begin{subfigure}[b]{0.48\textwidth}
			\includegraphics[width=\textwidth]{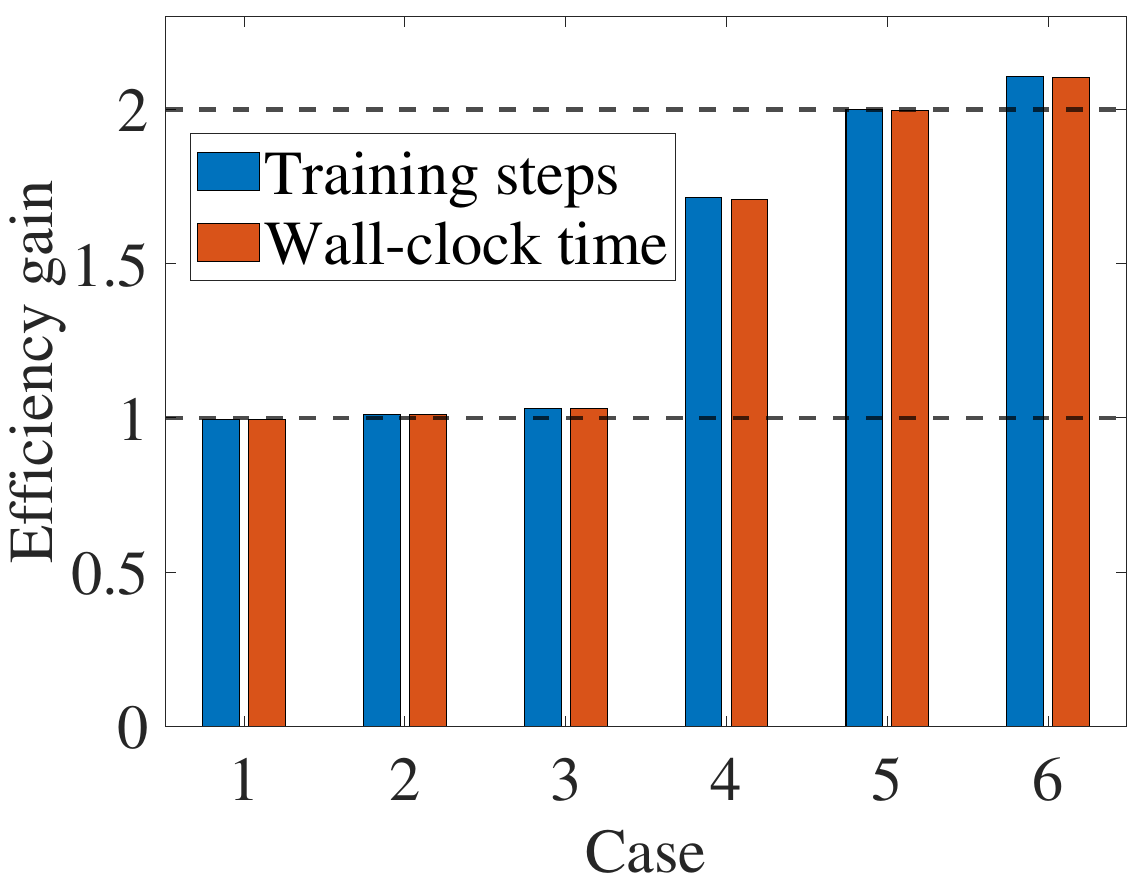}
			\caption*{(a) Convex loss function} 
		\end{subfigure}
		\hfill
		\begin{subfigure}[b]{0.48\textwidth}
			\includegraphics[width=\textwidth]{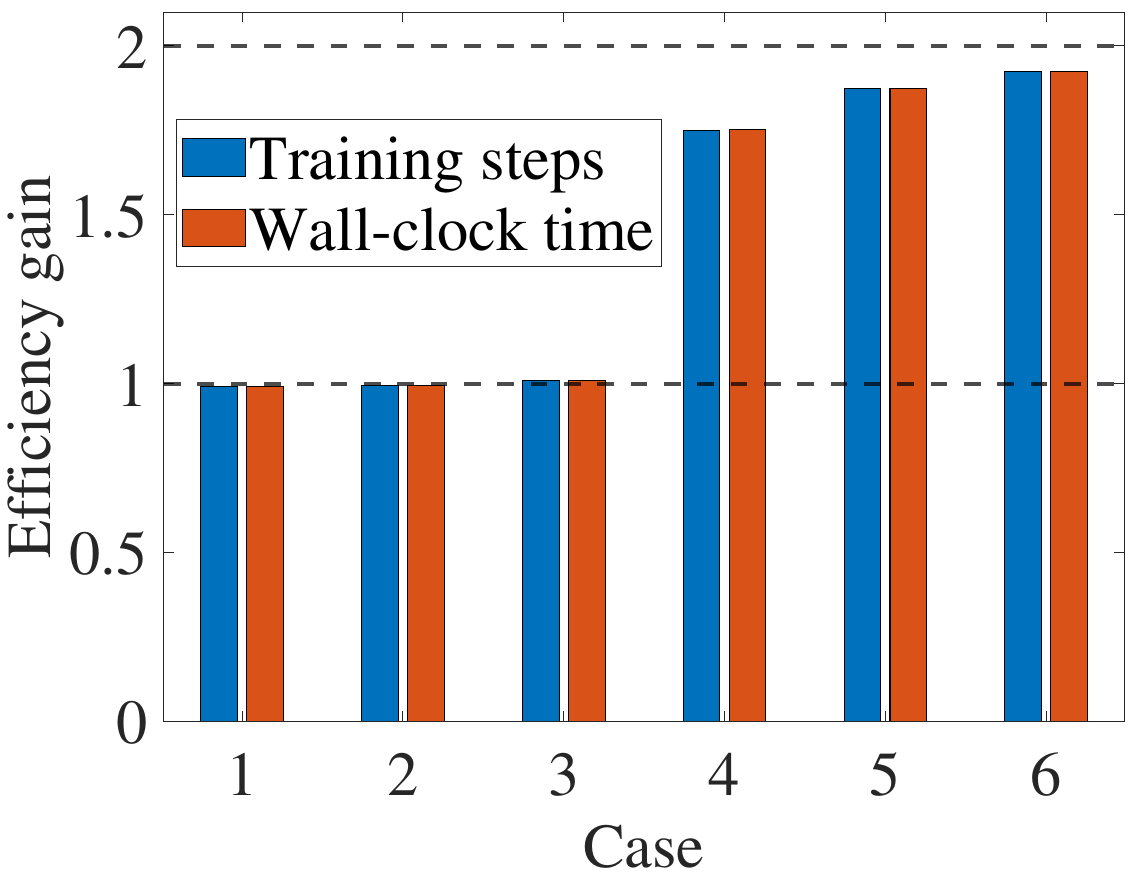}
			\caption*{(b) Non-convex loss function}
		\end{subfigure}
		\caption{Efficiency gain of HHFL over HFL across Cases 1-6:
			\textit{(IID, IID)}, \textit{(non-IID-$1$, IID)}, \textit{(non-IID-$2$, IID)},
			\textit{(non-IID, non-IID-$1$)}, \textit{(non-IID, non-IID-$2$)}, and
			\textit{(non-IID, non-IID-$2'$)}.}
		\label{fig:efficiency}
	\end{minipage}
	\hfill
	\begin{minipage}[t]{0.48\textwidth}
		\centering
		\begin{subfigure}[b]{0.48\textwidth}
			\includegraphics[width=\textwidth]{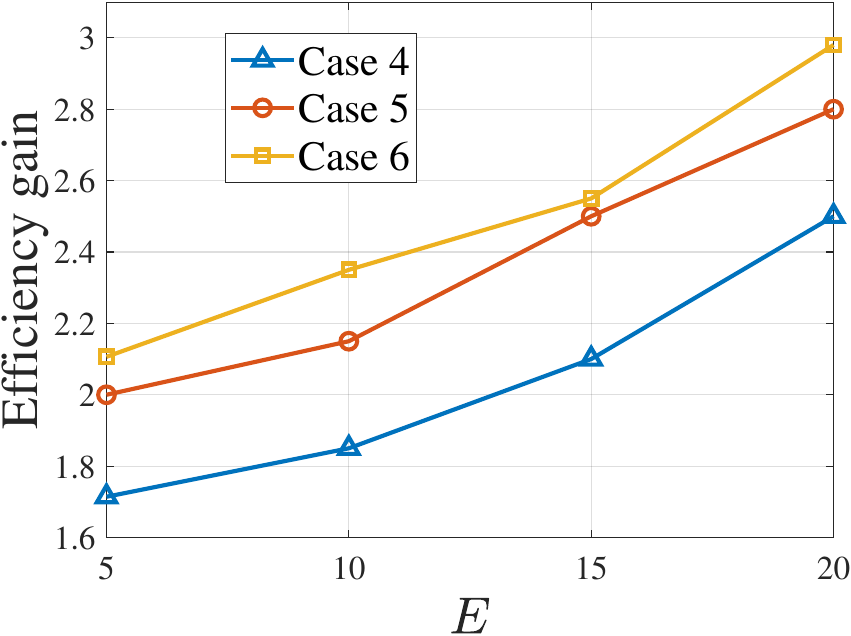}
			\caption*{(a) Convex loss function}
		\end{subfigure}
		\hfill
		\begin{subfigure}[b]{0.48\textwidth}
			\includegraphics[width=\textwidth]{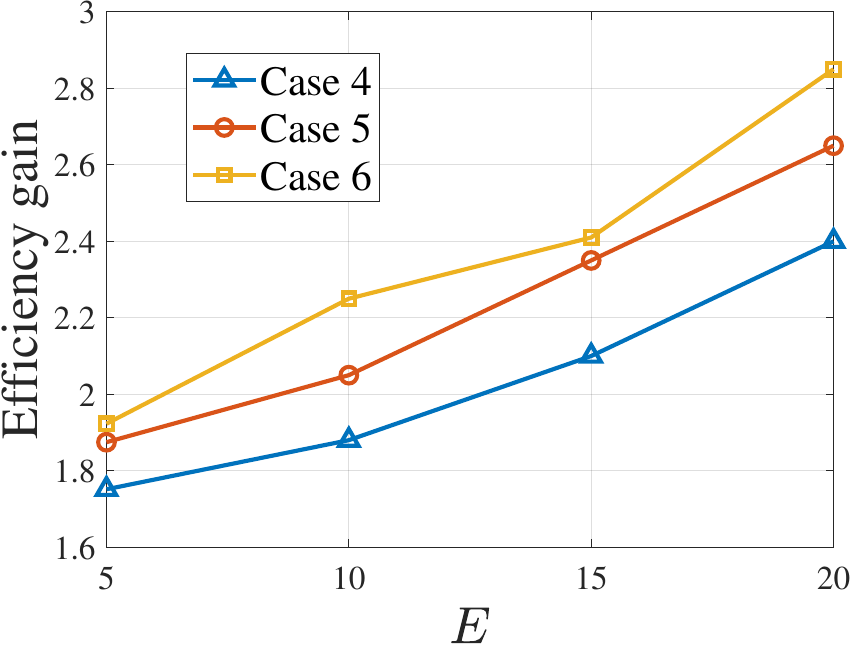}
			\caption*{(b) Non-convex loss function}
		\end{subfigure}
		\caption{Efficiency gain of HHFL over HFL with varying $E$ across Cases 4-6: \textit{(non-IID, non-IID-$1$)}, \textit{(non-IID, non-IID-$2$)}, and \textit{(non-IID, non-IID-$2'$)}.}
		\label{efficiency with varing E}
	\end{minipage}
	
\end{figure*}


\begin{figure*}[ht]
	\begin{minipage}[t]{0.49\textwidth}
		\centering
		\begin{subfigure}[b]{0.49\textwidth}
			\includegraphics[width=\textwidth]{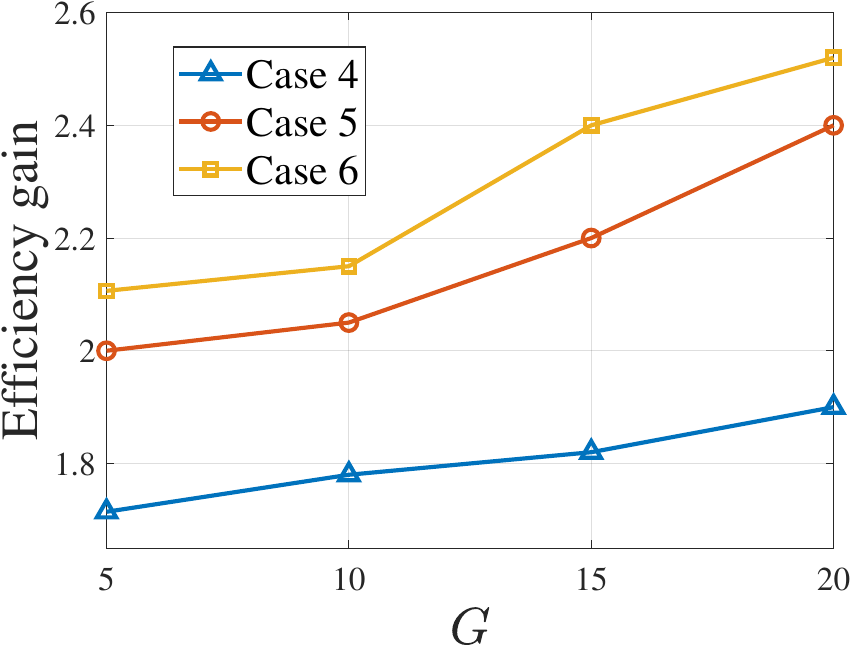}
			\caption*{(a) Convex loss function}
		\end{subfigure}
		\hfill
		\begin{subfigure}[b]{0.49\textwidth}
			\includegraphics[width=\textwidth]{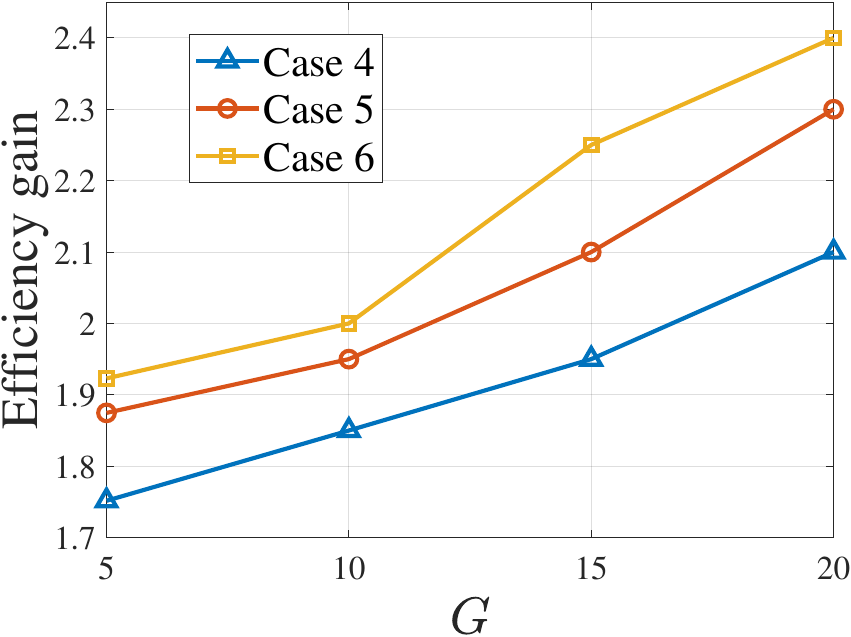}
			\caption*{(b) Non-convex loss function}
		\end{subfigure}
		\caption{Efficiency gain of HHFL over HFL with varying $G$ across Cases 4-6: \textit{(non-IID, non-IID-$1$)}, \textit{(non-IID, non-IID-$2$)}, and \textit{(non-IID, non-IID-$2'$)}.}
		\label{efficiency with varing G}
	\end{minipage}
	\hfill
	\begin{minipage}[t]{0.49\textwidth}
		\centering
		\begin{subfigure}[b]{0.49\textwidth}
			\includegraphics[width=\textwidth]{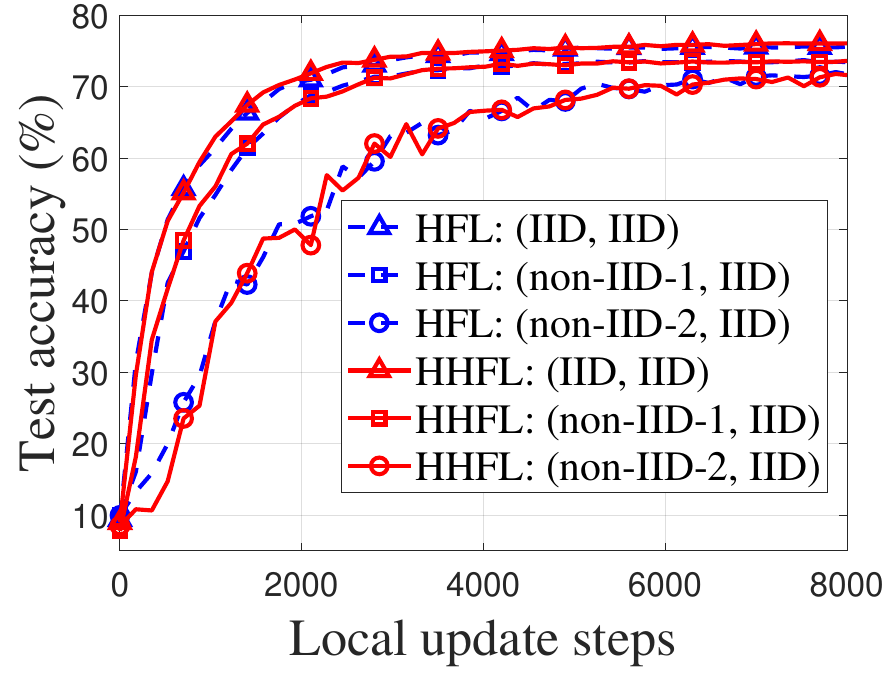}
			\caption*{(a) \textit{ES IID} scenario}
		\end{subfigure}
		\hfill
		\begin{subfigure}[b]{0.49\textwidth}
			\includegraphics[width=\textwidth]{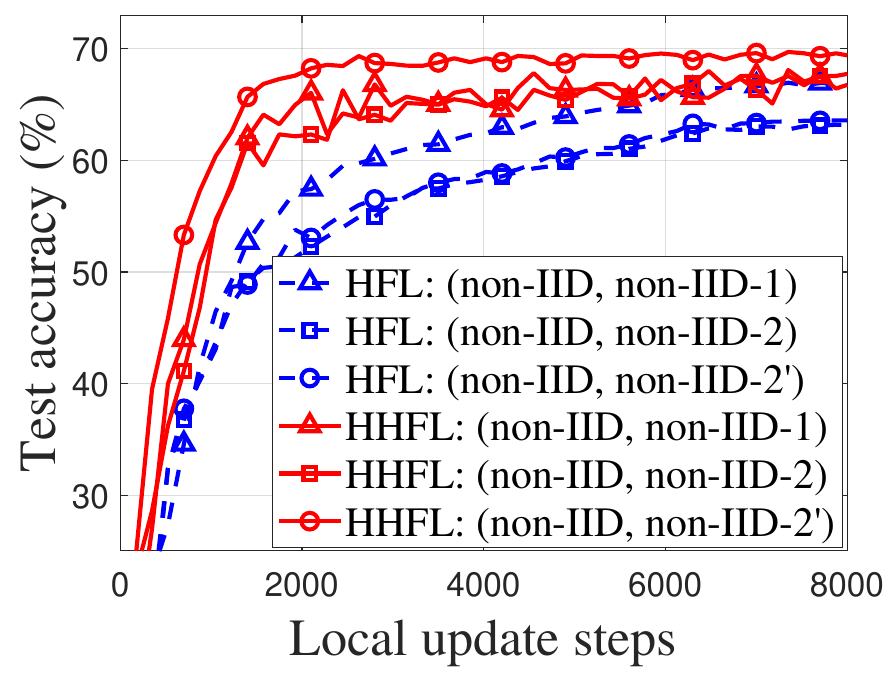}
			\caption*{(b) \textit{ES non-IID} scenario}
		\end{subfigure}
		\caption{{\hlb Convergence curves across six experimental cases on the CIFAR-10 dataset.}}
		\label{fig:cifar_convergence}
	\end{minipage}
\end{figure*}

\subsubsection{Efficiency gain of HHFL over HFL in terms of training steps and wall-clock time across all cases} Based on the obtained convergence curves, we calculate the number of training steps and the wall-clock time required for each architecture to reach convergence in each case using the same convergence criterion. We then compute the efficiency gain of HHFL over HFL across all cases, in terms of both training steps and wall-clock time. 
The results of the efficiency gain comparison are shown in Fig.~\ref{fig:efficiency}. It is evident that both the training steps and the overall time bars are almost equal in height in each case. This is because the overall time scales nearly linearly with training steps, as shown in (\ref{Eq: overall}), then the time ratio between HFL and HHFL closely matches the ratio of their training steps.
Besides that, the efficiency ratios for the first three cases under the \textit{ES IID} scenario are all close to $1$, which indicates that there is little to no efficiency improvement for the use of HHFL. Moreover, for the last three cases, the efficiency ratios are significantly greater than $1$, even reaching $2$ in case $6$, indicating that the use of HHFL results in substantial efficiency improvements.

\subsubsection{Impact of parameters $E$ and $G$ on the efficiency gain of HHFL over HFL} We evaluate how the efficiency gain of HHFL over HFL changes with different values of parameters $E$ and $G$ across three cases under \textit{ES non-IID} scenario: \textit{(non-IID, non-IID-$1$)}, \textit{(non-IID, non-IID-$2$)}, and \textit{(non-IID, non-IID-$2'$)}. The default values are set to $E = 5$ and $G = 5$. When varying one of the parameters, the other is kept fixed at its default value. Since the efficiency gains in terms of training steps and the overall time are nearly identical across all cases regardless of the model used, we use the efficiency gain with respect to training steps as a representative metric.
The changes in efficiency gain with respect to $E$ and $G$ are shown in Fig.~\ref{efficiency with varing E} and Fig.~\ref{efficiency with varing G}, respectively.
It can be observed that the efficiency gain increases in all three cases as either $E$ or $G$ becomes larger. {\hlb The underlying reason is that increasing $E$ or $G$ weakens the frequency of global alignment: a larger $E$ means more local steps are taken between edge aggregations, and a larger $G$ means cloud aggregation is invoked less often. Under non-IID data, both effects aggravate inter-ES model drift and make the ES models diverge more easily in conventional HFL, where ESs mainly evolve in isolation between cloud aggregations.  
	In contrast, HHFL introduces an additional cross-ES knowledge bridge via overlapping clients, which repeatedly receive models from multiple ESs and feed back updates. This inter-ES coupling partially compensates for the reduced global alignment, mitigates model drift, and helps maintain consistency across ESs. Therefore, when $E$ or $G$ increases, the relative advantage of HHFL over HFL becomes more pronounced. }
\begin{figure*}[ht]
	\begin{minipage}[t]{0.49\textwidth}
		\centering
		\begin{subfigure}[b]{0.49\textwidth}
			\includegraphics[width=\textwidth]{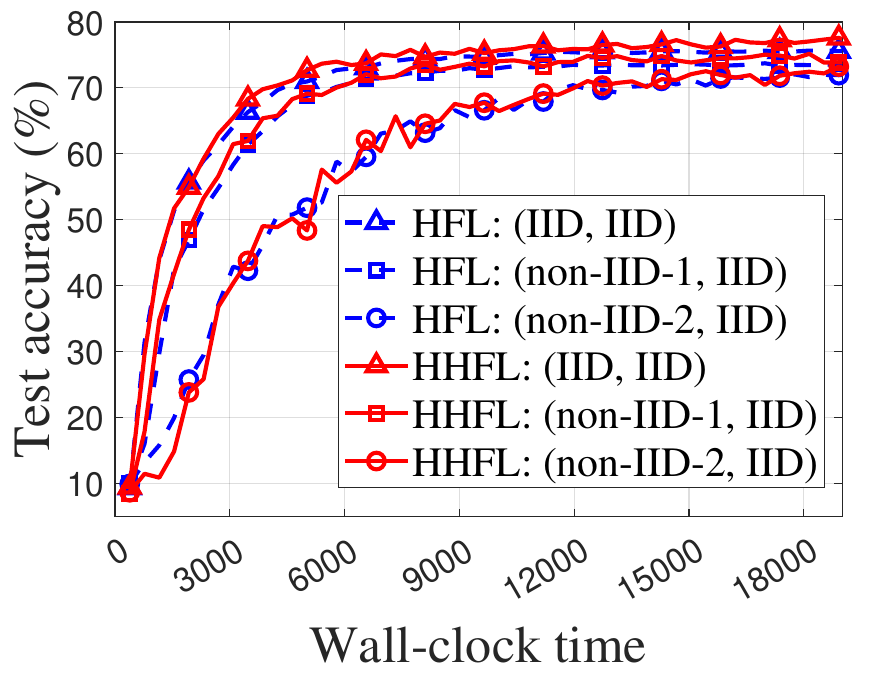}
			\caption*{(a) \textit{ES IID} scenario}
		\end{subfigure}
		\hfill
		\begin{subfigure}[b]{0.49\textwidth}
			\includegraphics[width=\textwidth]{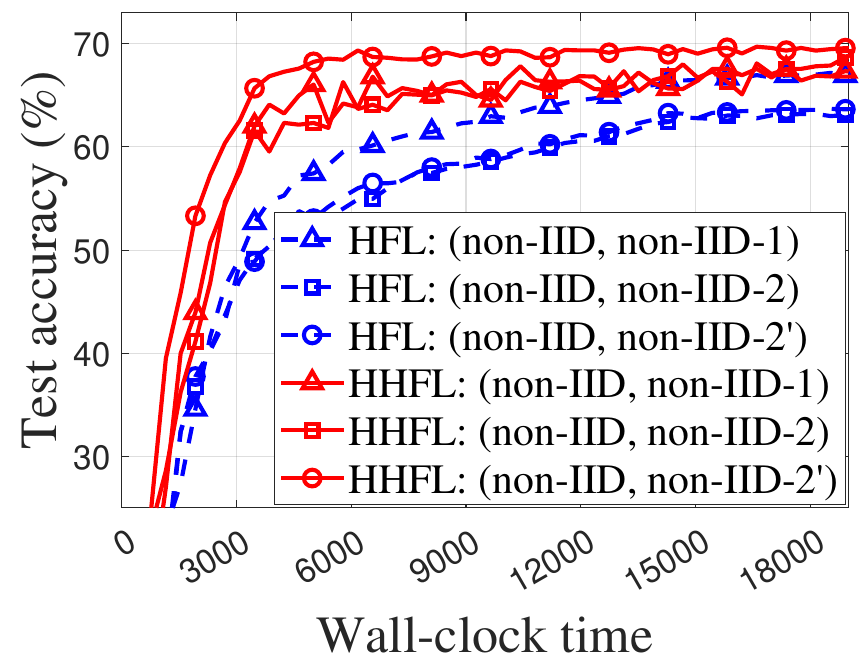}
			\caption*{(b) \textit{ES non-IID} scenario}
		\end{subfigure}
		\caption{{\hlb Test accuracy versus wall-clock time across six cases on the CIFAR-10 dataset (time is normalized such that one unit corresponds to the time for a client to complete $E$ local updates).}}
		\label{fig:Cifar_time}
	\end{minipage}
	\hfill
	\begin{minipage}[t]{0.49\textwidth}
		\centering
		\begin{subfigure}[b]{0.49\textwidth}
			\includegraphics[width=\textwidth]{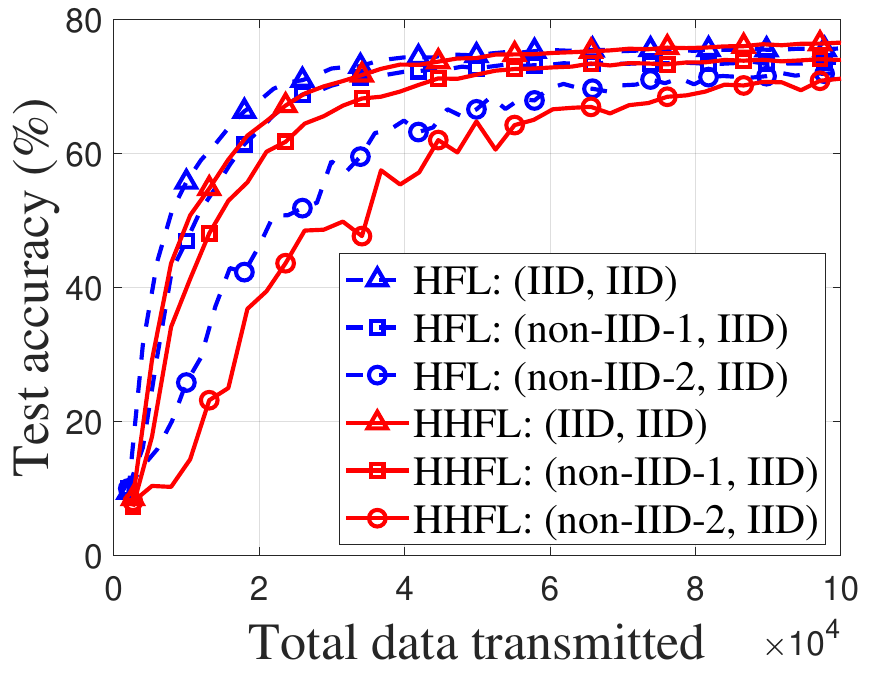}
			\caption*{(a) \textit{ES IID} scenario}
		\end{subfigure}
		\hfill
		\begin{subfigure}[b]{0.49\textwidth}
			\includegraphics[width=\textwidth]{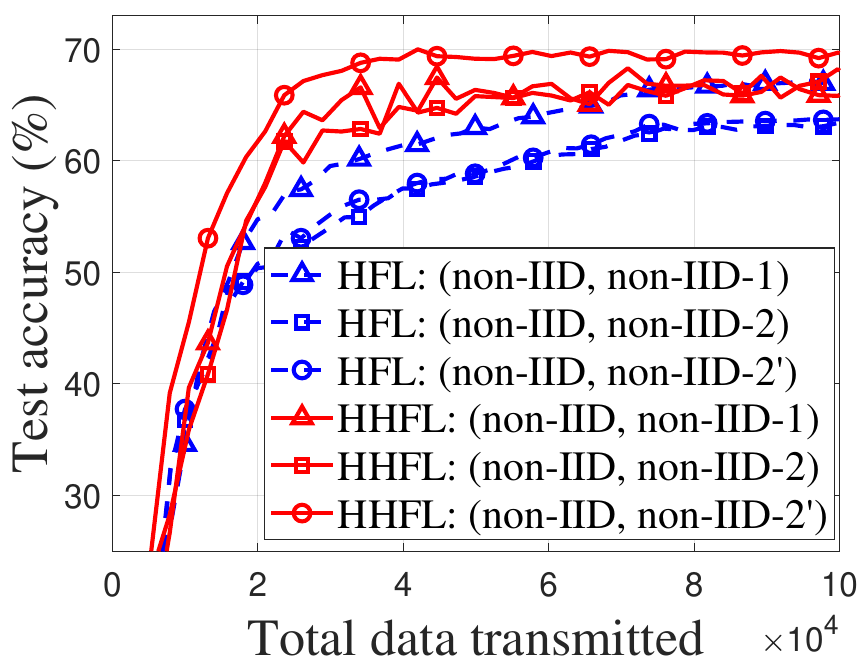}
			\caption*{(b) \textit{ES non-IID} scenario}
		\end{subfigure}
		\caption{{\hlb Test accuracy versus total data transmitted across six cases on the CIFAR-10 dataset (data volume is normalized such that one unit corresponds to the number of bits required to transmit one model).}}
		\label{fig:Cifar_data}
	\end{minipage}
\end{figure*}
{\hlb \subsection{Additional Experiments on CIFAR-10 dataset}
	\subsubsection{Setup}
	To validate the robustness of HHFL under complex data heterogeneity, we extended our evaluation to the more challenging CIFAR-10 dataset. The model architecture follows a VGG-11 style backbone, which comprises a feature extractor with three convolutional blocks (64, 128, and 256 channels) interspersed with max-pooling layers. To stabilize training under non-IID data and avoid overfitting, we applied group normalization and dropout. The feature extractor is followed by a streamlined two-layer fully connected classifier ($4096 \to 1024 \to 10$). Local training is performed using SGD with a momentum of 0.9 and a weight decay of $5 \times 10^{-4}$. The initial learning rate is $0.07$, with an exponential decay of $0.9$ per epoch. For other parameters and configurations, unless otherwise specified, we follow the same experimental settings as in our basic experiments on MNIST.
	\subsubsection{Results Analysis}
	The results are shown in Figs.~\ref{fig:cifar_convergence}--\ref{fig:Cifar_data}. We observe that the experimental trends on CIFAR-10 are highly consistent with those observed on the MNIST dataset. Specifically, as illustrated in Fig. \ref{fig:cifar_convergence} (a) and Fig. \ref{fig:Cifar_time} (a), when the data distribution across ESs is IID, HHFL exhibits negligible advantages over HFL in terms of both total training steps and wall-clock time. However, under non-IID data distributions across ESs, HHFL consistently exhibits significantly better convergence efficiency than HFL, as shown in Fig. \ref{fig:cifar_convergence} (b) and Fig. \ref{fig:Cifar_time} (b). The fundamental reason is the effect of multi-connectivity-enabled knowledge sharing, which makes the update directions of different ES models more aligned with the global optimization direction, thereby reducing the total number of updates required for convergence. Based on our convergence criterion, HHFL consistently achieves efficiency gains greater than $1.5\times$ over HFL across all these non-IID scenarios, in terms of both training steps and wall-clock time. Fig.~\ref{fig:Cifar_data} (a) shows that when the data distributions across ESs are IID, HHFL requires more total data transmission than HFL to reach the same test accuracy, as shown in Fig.~\ref{fig:Cifar_data} (b). This is because, in this case, the additional data transmitted over the multi-connectivity links is largely redundant in terms of knowledge content.  
	In contrast, when the data distributions across ESs are non-IID, HHFL transmits less total data than HFL to achieve convergence. This is because the data exchanged through multi-connectivity links is no longer redundant; instead, it conveys useful knowledge from different ESs. Moreover, the resulting benefit of accelerating training convergence outweighs the additional data-transmission cost introduced by the extra links.
	
}

\section{Related Work}\label{Sec:relatedwork}
In this section, we discuss the research related to this work. 

\noindent \textbf{HFL Designs:}
HFL has been extensively studied to address various real-world challenges faced by FL in large-scale wireless networks \cite{liu2020client,abad2020hierarchical,luo2020hfel,ma2024personalized,wu2023hiflash,alqattan2024security,mitra2023timely,wu2021hierarchical,gao2019hhhfl,feng2021min,wainakh2020enhancing,zhou2023toward,feng2022mobility,zhou2023hierarchical,mhaisen2021optimal,guo2023privacy,qu2022context,wang2022accelerating,li2019smartpc,de2023hed,liu2022time,fang2024hierarchical,azimi2024quantized,pervej2024personalized,xu2021adaptive,briggs2020federated,lim2021dynamic,su2025joint,liu2025compressed,jiang2024hdhrfl,lim2021decentralized,wang2021resource}. 
For example, several studies have explored cluster-based hierarchical aggregation, where clients are first grouped into clusters based on geographical proximity, data similarity, or mobility patterns to improve communication efficiency and training accuracy \cite{briggs2020federated,feng2022mobility,zhou2023hierarchical,wang2022accelerating}. Other works have investigated personalized hierarchical FL frameworks, where clients and ESs collaboratively adapt models to individual data distributions to tackle non-IID challenges \cite{ma2024personalized,wu2021hierarchical,pervej2024personalized}. In addition, some studies focus on optimizing resource allocation and edge association to dynamically improve training efficiency under varying network conditions \cite{luo2020hfel,wang2021resource,lim2021dynamic}. 
Although these studies have introduced valuable improvements across various aspects, they fundamentally preserve the traditional HFL architecture, where each client connects to a single ES, and the ESs periodically upload aggregated models to the CS. This communication structure inherently limits cross-ES information sharing and constrains the potential benefits of modern multi-connectivity wireless infrastructures. In contrast, our proposed HHFL architecture fully leverages the CoMP capabilities of modern wireless networks, allowing clients located in overlapping regions to simultaneously communicate with multiple ESs. By enabling such multi-ES connections, HHFL facilitates broader model sharing across different ESs, effectively mitigates statistical heterogeneity, and significantly accelerates convergence under practical non-IID distributions. 

\noindent \textbf{Multi-ES architecture:}
This FL architecture is technically regarded as a client-ES structure with multiple parallel aggregation ESs. Numerous research efforts have been devoted to this type of architecture, generally aiming to explore how inter-ES collaboration can improve overall learning performance, enhance convergence under non-IID data, and optimize communication and resource allocation efficiency in distributed systems \cite{kawnine2024evaluating,ganguly2023multi,han2021fedmes,qu2022convergence,nguyen2022latency,guo2021inter,diamanti2025resource,nguyen2025mario,thapa2021advancements,yang2022federated}. Among these, several studies have investigated the overlapping-region scenario, where a client may connect to multiple ESs simultaneously. This setting closely reflects the characteristics of 5G and beyond communication systems and is particularly relevant to our work. For example, \cite{han2021fedmes} proposed a multi-ES architecture and introduced the FedMes algorithm, in which clients located in overlapping regions aggregate the models from their connected ESs to initialize local training. Based on that, \cite{qu2022convergence} designed the MS-FedAvg algorithm, which further improved model aggregation and achieved higher efficiency.
These existing works allow only indirect and localized data sharing between ESs through clients in overlapping regions. However, our proposed HHFL supports direct and global cross-ES data sharing via the CS during training. Under practical 5G and beyond network environments, the wired connections between base stations and the cloud server are stable and high-speed, and it is common to see the data distributions across ESs are non-IID. In such situations, our HHFL can achieve higher efficiency by mitigating statistical heterogeneity while incurring only a small communication cost between the CS and ESs.

\section{Conclusion}\label{Sec:conclusion}
This paper presents a novel FL architecture named HHFL, which can be regarded as a hybrid extension of HFL. In this architecture, each client is capable of communicating with one or more ESs simultaneously, making HHFL inherently aligned with current and future wireless infrastructures that support CoMP. Through both theoretical analysis and experimental evaluation, we demonstrate the effectiveness and efficiency of HHFL, underscoring its strong potential as a scalable and efficient solution for practical FL model training in today’s and future wireless networks.
\bibliography{References}

@article{han2021fedmes,
  title={FedMes: Speeding up federated learning with multiple edge servers},
  author={Han, Dong-Jun and Choi, Minseok and Park, Jungwuk and Moon, Jaekyun},
  journal={IEEE J. Sel. Areas Commun. },
  year={2021},
  publisher={IEEE}
}

@inproceedings{liu2020client,
  title={Client-edge-cloud hierarchical federated learning},
  author={Liu, Lumin and Zhang, Jun and Song, SH and Letaief, Khaled B},
  booktitle={IEEE ICC},
  year={2020}
}

@article{li2019convergence,
  title={On the convergence of fedavg on non-iid data},
  author={Li, Xiang and Huang, Kaixuan and Yang, Wenhao and Wang, Shusen and Zhang, Zhihua},
  journal={arXiv preprint},
  year={2019}
}

@article{qu2022convergence,
  title={On the convergence of multi-server federated learning with overlapping area},
  author={Qu, Zhe and Li, Xingyu and Xu, Jie and Tang, Bo and Lu, Zhuo and Liu, Yao},
  journal={IEEE Trans. Mobile Comput. },
  year={2022},
  publisher={IEEE}
}

@inproceedings{fang2024submodel,
  title={Submodel partitioning in hierarchical federated learning: Algorithm design and convergence analysis},
  author={Fang, Wenzhi and Han, Dong-Jun and Brinton, Christopher G},
  booktitle={IEEE ICC},
  year={2024}
}

@inproceedings{mcmahan2017communication,
  title={Communication-efficient learning of deep networks from decentralized data},
  author={McMahan, Brendan and Moore, Eider and Ramage, Daniel and Hampson, Seth and y Arcas, Blaise Aguera},
  booktitle={PMLR AISTATS},
  year={2017}
}

@inproceedings{bonawitz2017practical,
  title={Practical secure aggregation for privacy-preserving machine learning},
  author={Bonawitz, Keith and Ivanov, Vladimir and Kreuter, Ben and Marcedone, Antonio and McMahan, H Brendan and Patel, Sarvar and Ramage, Daniel and Segal, Aaron and Seth, Karn},
  booktitle={ACM CCS},
  year={2017}
}

@inproceedings{li2020federated,
  title={Federated optimization in heterogeneous networks},
  author={Li, Tian and Sahu, Anit Kumar and Zaheer, Manzil and Sanjabi, Maziar and Talwalkar, Ameet and Smith, Virginia},
  booktitle={MLSys},
  year={2020}
}

@article{smith2017federated,
  title={Federated multi-task learning},
  author={Smith, Virginia and Chiang, Chao-Kai and Sanjabi, Maziar and Talwalkar, Ameet S},
  journal={Adv. Neural Inf. Process. Syst.},
  year={2017}
}

@article{fang2022communication,
  title={Communication-efficient stochastic zeroth-order optimization for federated learning},
  author={Fang, Wenzhi and Yu, Ziyi and Jiang, Yuning and Shi, Yuanming and Jones, Colin N and Zhou, Yong},
  journal={IEEE Trans. Signal Process. },
  year={2022},
  publisher={IEEE}
}

@article{kairouz2021advances,
  title={Advances and open problems in federated learning},
  author={Kairouz, Peter and McMahan, H Brendan and Avent, Brendan and Bellet, Aur{\'e}lien and Bennis, Mehdi and Bhagoji, Arjun Nitin and Bonawitz, Kallista and Charles, Zachary and Cormode, Graham and Cummings, Rachel and others},
  journal={Found. Trends Mach. Learn.},
  year={2021}
}

@article{bonawitz2019towards,
  title={Towards federated learning at scale: Syste m design},
  author={Bonawitz, Keith},
  journal={arXiv preprint},
  year={2019}
}

@article{zhao2018federated,
  title={Federated learning with non-iid data},
  author={Zhao, Yue and Li, Meng and Lai, Liangzhen and Suda, Naveen and Civin, Damon and Chandra, Vikas},
  journal={arXiv preprint},
  year={2018}
}

@article{wang2019adaptive,
  title={Adaptive federated learning in resource constrained edge computing systems},
  author={Wang, Shiqiang and Tuor, Tiffany and Salonidis, Theodoros and Leung, Kin K and Makaya, Christian and He, Ting and Chan, Kevin},
  journal={IEEE J. Sel. Areas Commun.},
  year={2019},
  publisher={IEEE}
}

@article{lyu2020threats,
  title={Threats to federated learning: A survey},
  author={Lyu, Lingjuan and Yu, Han and Yang, Qiang},
  journal={arXiv preprint},
  year={2020}
}

@article{nguyen2020fast,
  title={Fast-convergent federated learning},
  author={Nguyen, Hung T and Sehwag, Vikash and Hosseinalipour, Seyyedali and Brinton, Christopher G and Chiang, Mung and Poor, H Vincent},
  journal={IEEE J. Sel. Areas Commun.},
  year={2020},
  publisher={IEEE}
}

@article{amiri2021convergence,
  title={Convergence of update aware device scheduling for federated learning at the wireless edge},
  author={Amiri, Mohammad Mohammadi and G{\"u}nd{\"u}z, Deniz and Kulkarni, Sanjeev R and Poor, H Vincent},
  journal={IEEE Trans. Wireless Commun. },
  year={2021},
  publisher={IEEE}
}

@article{zhang2012communication,
  title={Communication-efficient algorithms for statistical optimization},
  author={Zhang, Yuchen and Wainwright, Martin J and Duchi, John C},
  journal={	Adv. Neural Inf. Process. Syst.},
  year={2012}
}

@article{stich2018local,
  title={Local SGD converges fast and communicates little},
  author={Stich, Sebastian U},
  journal={arXiv preprint},
  year={2018}
}

@inproceedings{reisizadeh2020fedpaq,
  title={Fedpaq: A communication-efficient federated learning method with periodic averaging and quantization},
  author={Reisizadeh, Amirhossein and Mokhtari, Aryan and Hassani, Hamed and Jadbabaie, Ali and Pedarsani, Ramtin},
  booktitle={PMLR AISTATS},
  year={2020}
}

@article{brik2020federated,
  title={Federated learning for UAVs-enabled wireless networks: Use cases, challenges, and open problems},
  author={Brik, Bouziane and Ksentini, Adlen and Bouaziz, Maha},
  journal={IEEE Access},
  year={2020},
  publisher={IEEE}
}

@article{nguyen2021federated,
  title={Federated learning for Internet of Things: A comprehensive survey},
  author={Nguyen, Dinh C. and Ding, Ming and Pathirana, Pubudu N. and Seneviratne, Aruna and Li, Jun and Vincent Poor, H.},
  journal={IEEE Commun. Surv. Tutor.},
  year={2021},
  publisher={IEEE}
}

@article{yin2020fedloc,
  title={FedLoc: Federated learning framework for data-driven cooperative localization and location data processing},
  author={Yin, Feng and Lin, Zhidi and Kong, Qinglei and Xu, Yue and Li, Deshi and Theodoridis, Sergios and Cui, Shuguang Robert},
  journal={IEEE Open J. Signal Process.},
  year={2020},
  publisher={IEEE}
}

@article{mehmet1988traffic,
  title={Traffic analysis of a local area network with a star topology},
  author={Mehmet-Ali, Mustafa K and Hayes, Jeremiah F and Elhakeem, Ahmed K},
  journal={IEEE Trans. Commun.},
  year={1988},
  publisher={IEEE}
}

@article{petrek2001large,
  title={A large hierarchical network star—star topology design algorithm},
  author={Petrek, Jozef and Sledt, Volker},
  journal={Eur. Trans. Telecommun.},
  year={2001},
  publisher={Wiley Online Library}
}

@article{goratti2015nacrp,
  title={NACRP: A connectivity protocol for star topology wireless sensor networks},
  author={Goratti, Leonardo and Baykas, Tuncer and Rasheed, Tinku and Kato, Shuzo},
  journal={IEEE Wirel. Commun. Lett.},
  year={2015},
  publisher={IEEE}
}

@article{barranco2006active,
  title={An active star topology for improving fault confinement in CAN networks},
  author={Barranco, Manuel and Proenza, Juli{\'a}n and Rodr{\'\i}guez-Navas, Guillermo and Almeida, Lu{\'\i}s},
  journal={IEEE Trans Ind. Informat. },
  year={2006},
  publisher={IEEE}
}

@article{day2002comparative,
  title={A comparative study of topological properties of hypercubes and star graphs},
  author={Day, Khaled and Tripathi, Anand},
  journal={IEEE Trans. Parallel Distrib. Syst.},
  year={2002},
  publisher={IEEE}
}

@article{xu2021adaptive,
  title={Adaptive hierarchical federated learning over wireless networks},
  author={Xu, Bo and Xia, Wenchao and Wen, Wanli and Liu, Pei and Zhao, Haitao and Zhu, Hongbo},
  journal={IEEE Trans. Veh. Technol.},
  year={2021},
  publisher={IEEE}
}

@article{ravasz2003hierarchical,
  title={Hierarchical organization in complex networks},
  author={Ravasz, Erzs{\'e}bet and Barab{\'a}si, Albert-L{\'a}szl{\'o}},
  journal={APS Physical review E},
  year={2003},
  publisher={APS}
}

@inproceedings{lessmann2007parameterized,
  title={Parameterized hierarchical layer topology construction for wireless networks},
  author={Lessmann, Johannes and Krishnamurthy, Arvind},
  booktitle={IEEE ICSNC},
  year={2007}
}

@article{jiao2023performance,
  title={Performance analysis for downlink transmission in multiconnectivity cellular V2X networks},
  author={Jiao, Luofang and Zhao, Jiwei and Xu, Yunting and Zhang, Tianqi and Zhou, Haibo and Zhao, Dongmei},
  journal={	IEEE Internet Things J.},
  year={2023},
  publisher={IEEE}
}

@article{ye2023multi,
  title={Multi-connection to the sky: Energy-efficient Beamforming for multi-satellite uplink transmission with lens antenna array},
  author={Ye, Neng and Cao, Xinyuan and Ding, Xuhui and Li, Jiaxuan and Zhao, Deguang and Ouyang, Qiaolin},
  journal={IEEE Trans. Green Commun. Netw.},
  year={2023},
  publisher={IEEE}
}

@article{xiao2024space,
  title={Space-air-ground integrated wireless networks for 6G: Basics, key technologies and future trends},
  author={Xiao, Yue and Ye, Ziqiang and Wu, Mingming and Li, Haoyun and Xiao, Ming and Alouini, Mohamed-Slim and Al-Hourani, Akram and Cioni, Stefano},
  journal={IEEE J. Sel. Areas Commun. },
  year={2024},
  publisher={IEEE}
}

@inproceedings{xu2020research,
  title={Research on 6G mobile communication system},
  author={Xu, Guangbin},
  booktitle={Journal of Physics: Conference Series},
  year={2020},
  organization={IOP Publishing}
}

@article{chowdhury20206g,
  title={6G wireless communication systems: Applications, requirements, technologies, challenges, and research directions},
  author={Chowdhury, Mostafa Zaman and Shahjalal, Md and Ahmed, Shakil and Jang, Yeong Min},
  journal={IEEE Open J. Commun. Soc.},
  year={2020},
  publisher={IEEE}
}

@article{akyildiz20206g,
  title={6G and beyond: The future of wireless communications systems},
  author={Akyildiz, Ian F and Kak, Ahan and Nie, Shuai},
  journal={IEEE access},
  year={2020},
  publisher={IEEE}
}

@article{yang20196g,
  title={6G wireless communications: Vision and potential techniques},
  author={Yang, Ping and Xiao, Yue and Xiao, Ming and Li, Shaoqian},
  journal={IEEE network},
  year={2019},
  publisher={IEEE}
}

@inproceedings{briggs2020federated,
  title={Federated learning with hierarchical clustering of local updates to improve training on non-IID data},
  author={Briggs, Christopher and Fan, Zhong and Andras, Peter},
  booktitle={IEEE IJCNN},
  year={2020}
}

@article{bhushan2014network,
  title={Network densification: the dominant theme for wireless evolution into 5G},
  author={Bhushan, Naga and Li, Junyi and Malladi, Durga and Gilmore, Rob and Brenner, Dean and Damnjanovic, Aleksandar and Sukhavasi, Ravi Teja and Patel, Chirag and Geirhofer, Stefan},
  journal={IEEE Commun. Mag. },
  year={2014},
  publisher={IEEE}
}

@article{ullah2023survey,
  title={A survey on handover and mobility management in 5G HetNets: current state, challenges, and future directions},
  author={Ullah, Yasir and Roslee, Mardeni Bin and Mitani, Sufian Mousa and Khan, Sajjad Ahmad and Jusoh, Mohamad Huzaimy},
  journal={Sensors},
  year={2023},
  publisher={MDPI}
}

@article{borralho2021survey,
  title={A survey on coverage enhancement in cellular networks: Challenges and solutions for future deployments},
  author={Borralho, Ruben and Mohamed, Abdelrahim and Quddus, Atta Ul and Vieira, Pedro and Tafazolli, Rahim},
  journal={IEEE Commun. Surv. Tutor.},
  year={2021},
  publisher={IEEE}
}

@article{rehman2023survey,
  title={A survey of handover management in mobile HetNets: Current challenges and future directions},
  author={Rehman, Aziz Ur and Roslee, Mardeni Bin and Jun Jiat, Tiang},
  journal={Applied Sciences},
  year={2023},
  publisher={MDPI}
}

@article{luo2020hfel,
  title={HFEL: Joint edge association and resource allocation for cost-efficient hierarchical federated edge learning},
  author={Luo, Siqi and Chen, Xu and Wu, Qiong and Zhou, Zhi and Yu, Shuai},
  journal={IEEE Trans. Wireless Commun. },
  year={2020},
  publisher={IEEE}
}

@article{ma2024personalized,
  title={Personalized client-edge-cloud hierarchical federated learning in mobile edge computing},
  author={Ma, Chunmei and Li, Xiangqian and Huang, Baogui and Li, Guangshun and Li, Fengyin},
  journal={	Springer J. Cloud Comput.},
  year={2024},
  publisher={Springer}
}

@article{wu2023hiflash,
  title={HiFlash: Communication-efficient hierarchical federated learning with adaptive staleness control and heterogeneity-aware client-edge association},
  author={Wu, Qiong and Chen, Xu and Ouyang, Tao and Zhou, Zhi and Zhang, Xiaoxi and Yang, Shusen and Zhang, Junshan},
  journal={IEEE Trans. Parallel Distrib. Syst.},
  year={2023},
  publisher={IEEE}
}

@inproceedings{alqattan2024security,
  title={Security Assessment of Hierarchical Federated Deep Learning},
  author={Alqattan, Duaa S and Sun, Rui and Liang, Huizhi and Nicosia, Guiseppe and Snasel, Vaclav and Ranjan, Rajiv and Ojha, Varun},
  booktitle={Springer ICANN},
  year={2024}
}

@inproceedings{mitra2023timely,
  title={Timely asynchronous hierarchical federated learning: Age of convergence},
  author={Mitra, Purbesh and Ulukus, Sennur},
  booktitle={IEEE WiOpt},
  year={2023}
}

@inproceedings{shamai2001enhancing,
  title={Enhancing the cellular downlink capacity via co-processing at the transmitting end},
  author={Shamai, Shlomo and Zaidel, Benjamin M},
  booktitle={IEEE VTC Spring},
  year={2001}
}

@article{qamar2017comprehensive,
  title={A comprehensive review on coordinated multi-point operation for LTE-A},
  author={Qamar, Faizan and Dimyati, Kaharudin Bin and Hindia, Mhd Nour and Noordin, Kamarul Ariffin Bin and Al-Samman, Ahmed M},
  journal={Computer Networks},
  year={2017},
  publisher={Elsevier}
}

@article{solaija2021generalized,
  title={Generalized coordinated multipoint framework for 5G and beyond},
  author={Solaija, Muhammad Sohaib J and Salman, Hanadi and Kihero, Abuu B and Sa{\u{g}}lam, Mehmet {\.I}zzet and Arslan, H{\"u}seyin},
  journal={IEEE Access},
  year={2021},
  publisher={IEEE}
}

@inproceedings{irram2020coordinated,
  title={Coordinated multi-point transmission in 5G and beyond heterogeneous networks},
  author={Irram, Fauzia and Ali, Mudassar and Maqbool, Zubdah and Qamar, Farhan and Rodrigues, Joel JPC},
  booktitle={IEEE INMIC},
  year={2020}
}

@article{shen2022comp,
  title={CoMP enhanced subcarrier and power allocation for multi-numerology based 5G-NR networks},
  author={Shen, Li-Hsiang and Su, Chia-Yu and Feng, Kai-Ten},
  journal={IEEE Trans. Veh. Technol. },
  year={2022},
  publisher={IEEE}
}

@article{jiang2021road,
  title={The road towards 6G: A comprehensive survey},
  author={Jiang, Wei and Han, Bin and Habibi, Mohammad Asif and Schotten, Hans Dieter},
  journal={IEEE Open J. Commun. Soc.},
  year={2021},
  publisher={IEEE}
}

@inproceedings{wu2021hierarchical,
  title={Hierarchical personalized federated learning for user modeling},
  author={Wu, Jinze and Liu, Qi and Huang, Zhenya and Ning, Yuting and Wang, Hao and Chen, Enhong and Yi, Jinfeng and Zhou, Bowen},
  booktitle={ACM Web Conf.},
  year={2021}
}

@inproceedings{liu2022time,
  title={Time minimization in hierarchical federated learning},
  author={Liu, Chang and Chua, Terence Jie and Zhao, Jun},
  booktitle={ IEEE/ACM SEC},
  year={2022}
}

@article{fang2024hierarchical,
  title={Hierarchical federated learning with multi-timescale gradient correction},
  author={Fang, Wenzhi and Han, Dong-Jun and Chen, Evan and Wang, Shiqiang and Brinton, Christopher},
  journal={Adv. Neural Inf. Process. Syst.},
  year={2024}
}

@article{azimi2024quantized,
  title={Quantized hierarchical federated learning: A robust approach to statistical heterogeneity},
  author={Azimi-Abarghouyi, Seyed Mohammad and Fodor, Viktoria},
  journal={arXiv preprint},
  year={2024}
}

@article{pervej2024personalized,
  title={Personalized Hierarchical Split Federated Learning in Wireless Networks},
  author={Pervej, Md-Ferdous and Molisch, Andreas F},
  journal={arXiv preprint},
  year={2024}
}

@inproceedings{abad2020hierarchical,
  title={Hierarchical federated learning across heterogeneous cellular networks},
  author={Abad, Mehdi Salehi Heydar and Ozfatura, Emre and Gunduz, Deniz and Ercetin, Ozgur},
  booktitle={IEEE ICASSP},
  year={2020}
}

@article{lim2021dynamic,
  title={Dynamic edge association and resource allocation in self-organizing hierarchical federated learning networks},
  author={Lim, Wei Yang Bryan and Ng, Jer Shyuan and Xiong, Zehui and Niyato, Dusit and Miao, Chunyan and Kim, Dong In},
  journal={IEEE J. Sel. Areas Commun. },
  year={2021},
  publisher={IEEE}
}

@article{su2025joint,
  title={Joint Adaptive Aggregation and Resource Allocation for Hierarchical Federated Learning Systems Based on Edge-Cloud Collaboration},
  author={Su, Yi and Fan, Wenhao and Meng, Qingcheng and Chen, Penghui and Liu, Yuan'an},
  journal={IEEE Trans. Cloud Comput.},
  year={2025},
  publisher={IEEE}
}

@article{liu2025compressed,
  title={Compressed Hierarchical Federated Learning for Edge-Level Imbalanced Wireless Networks},
  author={Liu, Yuan and Qu, Zhe and Wang, Jianxin},
  journal={IEEE Trans. Comput. Soc. Syst.},
  year={2025},
  publisher={IEEE}
}

@article{jiang2024hdhrfl,
  title={Hdhrfl: A hierarchical robust federated learning framework for dual-heterogeneous and noisy clients},
  author={Jiang, Yalan and Wang, Dan and Song, Bin and Luo, Shengyang},
  journal={Future Generation Computer Systems},
  year={2024},
  publisher={Elsevier}
}

@article{lim2021decentralized,
  title={Decentralized edge intelligence: A dynamic resource allocation framework for hierarchical federated learning},
  author={Lim, Wei Yang Bryan and Ng, Jer Shyuan and Xiong, Zehui and Jin, Jiangming and Zhang, Yang and Niyato, Dusit and Leung, Cyril and Miao, Chunyan},
  journal={IEEE Trans. Parallel Distrib. Syst.},
  year={2021},
  publisher={IEEE}
}

@inproceedings{wang2021resource,
  title={Resource-efficient federated learning with hierarchical aggregation in edge computing},
  author={Wang, Zhiyuan and Xu, Hongli and Liu, Jianchun and Huang, He and Qiao, Chunming and Zhao, Yangming},
  booktitle={IEEE INFOCOM },
  year={2021}}

@article{gao2019hhhfl,
  title={Hhhfl: Hierarchical heterogeneous horizontal federated learning for electroencephalography},
  author={Gao, Dashan and Ju, Ce and Wei, Xiguang and Liu, Yang and Chen, Tianjian and Yang, Qiang},
  journal={arXiv preprint },
  year={2019}
}

@article{feng2021min,
  title={Min-max cost optimization for efficient hierarchical federated learning in wireless edge networks},
  author={Feng, Jie and Liu, Lei and Pei, Qingqi and Li, Keqin},
  journal={IEEE Trans. Parallel Distrib. Syst.},
  year={2021},
  publisher={IEEE}
}

@inproceedings{wainakh2020enhancing,
  title={Enhancing privacy via hierarchical federated learning},
  author={Wainakh, Aidmar and Guinea, Alejandro Sanchez and Grube, Tim and M{\"u}hlh{\"a}user, Max},
  booktitle={ IEEE EuroS\&PW},
  year={2020}
}

@article{zhou2023toward,
  title={Toward robust hierarchical federated learning in internet of vehicles},
  author={Zhou, Hongliang and Zheng, Yifeng and Huang, Hejiao and Shu, Jiangang and Jia, Xiaohua},
  journal={IEEE Trans. Intell. Transp. Syst.},
  year={2023},
  publisher={IEEE}
}

@article{feng2022mobility,
  title={Mobility-aware cluster federated learning in hierarchical wireless networks},
  author={Feng, Chenyuan and Yang, Howard H and Hu, Deshun and Zhao, Zhiwei and Quek, Tony QS and Min, Geyong},
  journal={IEEE Trans. Wireless Commun. },
  year={2022},
  publisher={IEEE}
}

@article{zhou2023hierarchical,
  title={Hierarchical federated learning with social context clustering-based participant selection for internet of medical things applications},
    author={Zhou, Xiaokang and Ye, Xiaozhou and Wang, Kevin I-Kai and Liang, Wei and Nair, Nirmal Kumar C. and Shimizu, Shohei and Yan, Zheng and Jin, Qun},
  journal={IEEE Trans. Comput. Soc. Syst.},
  year={2023},
  publisher={IEEE}
}

@article{mhaisen2021optimal,
  title={Optimal user-edge assignment in hierarchical federated learning based on statistical properties and network topology constraints},
  author={Mhaisen, Naram and Abdellatif, Alaa Awad and Mohamed, Amr and Erbad, Aiman and Guizani, Mohsen},
  journal={IEEE Trans. Netw. Sci. Eng.},
  year={2021},
  publisher={IEEE}
}

@article{guo2023privacy,
  title={Privacy vs. efficiency: Achieving both through adaptive hierarchical federated learning},
  author={Guo, Yeting and Liu, Fang and Zhou, Tongqing and Cai, Zhiping and Xiao, Nong},
  journal={IEEE Trans. Parallel Distrib. Syst.},
  year={2023},
  publisher={IEEE}
}

@article{qu2022context,
  title={Context-aware online client selection for hierarchical federated learning},
  author={Qu, Zhe and Duan, Rui and Chen, Lixing and Xu, Jie and Lu, Zhuo and Liu, Yao},
  journal={IEEE Trans. Parallel Distrib. Syst.},
  year={2022},
  publisher={IEEE}
}

@article{wang2022accelerating,
  title={Accelerating federated learning with cluster construction and hierarchical aggregation},
  author={Wang, Zhiyuan and Xu, Hongli and Liu, Jianchun and Xu, Yang and Huang, He and Zhao, Yangming},
  journal={IEEE Trans. Mobile Comput. },
  year={2022},
  publisher={IEEE}
}

@inproceedings{li2019smartpc,
  title={SmartPC: Hierarchical pace control in real-time federated learning system},
  author={Li, Li and Xiong, Haoyi and Guo, Zhishan and Wang, Jun and Xu, Cheng-Zhong},
  booktitle={ IEEE RTSS},
  year={2019}
}

@article{de2023hed,
  title={HED-FL: A hierarchical, energy efficient, and dynamic approach for edge Federated Learning},
  author={De Rango, Floriano and Guerrieri, Antonio and Raimondo, Pierfrancesco and Spezzano, Giandomenico},
  journal={Pervasive and Mobile Computing},
  year={2023},
  publisher={Elsevier}
}

@article{qi2024bridge,
  title={Bridge the present and future: A cross-layer matching game in dynamic cloud-aided mobile edge networks},
  author={Qi, Houyi and Liwang, Minghui and Wang, Xianbin and Li, Li and Gong, Wei and Jin, Jian and Jiao, Zhenzhen},
  journal={IEEE Trans. Mobile Comput. },
  year={2024},
  publisher={IEEE}
}

@inproceedings{fezeu2023mmwave,
  title={An In-Depth Measurement Analysis of 5G mmWave PHY Latency and Its Impact on End-to-End Delay},
  author={Fezeu, R.A.K. and Brunstrom, A. and Flores, M. and Fiore, M.},
  booktitle={Springer PAM},
  year={2023}
}

@article{agiwal2016next,
  title={Next Generation 5G Wireless Networks: A Comprehensive Survey}, 
  author={Agiwal, Mamta and Roy, Abhishek and Saxena, Navrati},
  journal={IEEE Commun. Surv. Tutor}, 
  year={2016},
  publisher={IEEE}
}

@inproceedings{kawnine2024evaluating,
  title={Evaluating multi-global server architecture for federated learning},
  author={Kawnine, Asfia and Cao, Hung and Mih, Atah Nuh and Wachowicz, Monica},
  booktitle={ IEEE ICCE},
  year={2024}
}

@article{ganguly2023multi,
  title={Multi-edge server-assisted dynamic federated learning with an optimized floating aggregation point},
  author={Ganguly, Bhargav and Hosseinalipour, Seyyedali and Kim, Kwang Taik and Brinton, Christopher G and Aggarwal, Vaneet and Love, David J and Chiang, Mung},
  journal={IEEE/ACM Trans. Netw. },
  year={2023},
  publisher={IEEE}
}

@article{nguyen2022latency,
  title={Latency optimization for blockchain-empowered federated learning in multi-server edge computing},
  author={Nguyen, Dinh C and Hosseinalipour, Seyyedali and Love, David J and Pathirana, Pubudu N and Brinton, Christopher G},
  journal={IEEE J. Sel. Areas Commun. },
  year={2022},
  publisher={IEEE}
}

@article{guo2021inter,
  title={Inter-server collaborative federated learning for ultra-dense edge computing},
  author={Guo, Hongzhi and Huang, Weifeng and Liu, Jiajia and Wang, Yutao},
  journal={IEEE Trans. Wireless Commun. },
  year={2021},
  publisher={IEEE}
}

@article{diamanti2025resource,
  title={Resource allocation and pricing for multi-server multi-model federated learning based on market equilibrium},
  author={Diamanti, Maria and Rahman, Aisha B and Charatsaris, Panagiotis and Tsiropoulou, Eirini Eleni and Papavassiliou, Symeon},
  journal={Future Gener. Comput. Syst.},
  year={2025},
  publisher={Elsevier}
}

@inproceedings{nguyen2025mario,
  title={Mario: multi-round multiple-aggregator secure aggregation with robustness against malicious actors},
  author={Nguyen, Truong Son and Lepoint, Tancr{\`e}de and Trieu, Ni},
  booktitle={ IEEE EuroS\&P},
  year={2025}
}

@incollection{thapa2021advancements,
  title={Advancements of federated learning towards privacy preservation: from federated learning to split learning},
  author={Thapa, Chandra and Chamikara, Mahawaga Arachchige Pathum and Camtepe, Seyit A},
  booktitle={Springer Federated Learning Systems: Towards Next-Generation AI},
  year={2021}
}

@article{yang2022federated,
  title={Federated learning for 6G: Applications, challenges, and opportunities},
  author={Yang, Zhaohui and Chen, Mingzhe and Wong, Kai-Kit and Poor, H Vincent and Cui, Shuguang},
  journal={Elsevier Engineering},
  year={2022},
  publisher={Elsevier}
}

@article{standard2017final,
  title={Final draft ETSI ES 203 228 V1. 2.0 (2017-02)},
  author={STANDARD, ETSI},
  year={2017}
}

@article{liu2020analyzing,
  title={Analyzing grant-free access for URLLC service},
  author={Liu, Yan and Deng, Yansha and Elkashlan, Maged and Nallanathan, Arumugam and Karagiannidis, George K},
  journal={IEEE J. Sel. Areas Commun.},
  year={2020},
  publisher={IEEE}
}
\begin{IEEEbiography}[{\includegraphics[width=1in,height=1.5in,clip,keepaspectratio]{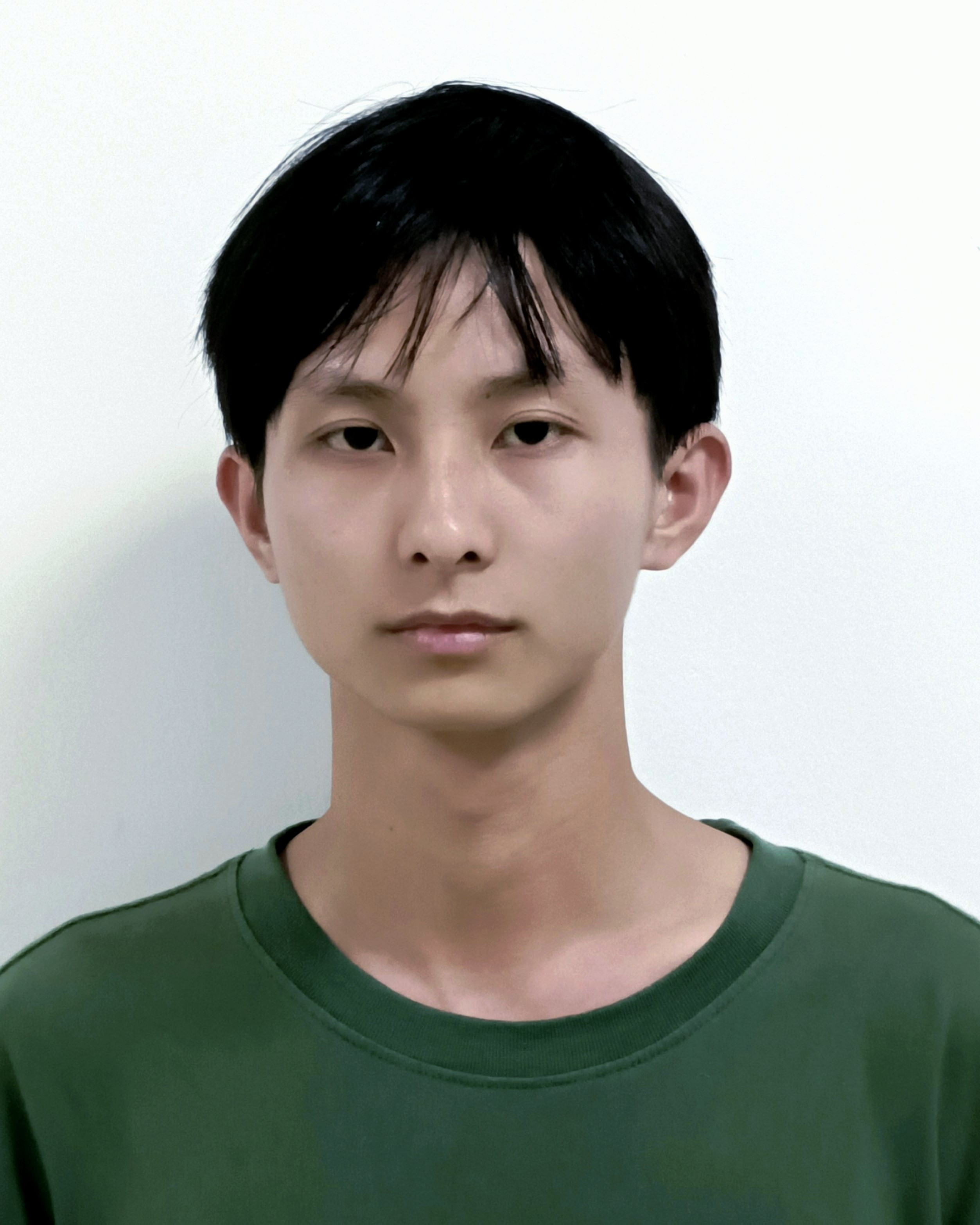}}]{Haiyun Liu} (Student Member, IEEE) received the B.S. degree in electrical engineering and automation from Shanghai University, Shanghai, China, in 2019, and the M.S. degree in signal and information processing from Sichuan University, Chengdu, China, in 2023. He is currently a Ph.D. student in the Department of Computer Science and Engineering, University of South Florida, Tampa, FL, USA. His research interests include security and privacy in wireless communications, and federated learning for networks.
\end{IEEEbiography}
\begin{IEEEbiography}[{\includegraphics[width=1in,height=1.8in,clip,keepaspectratio]{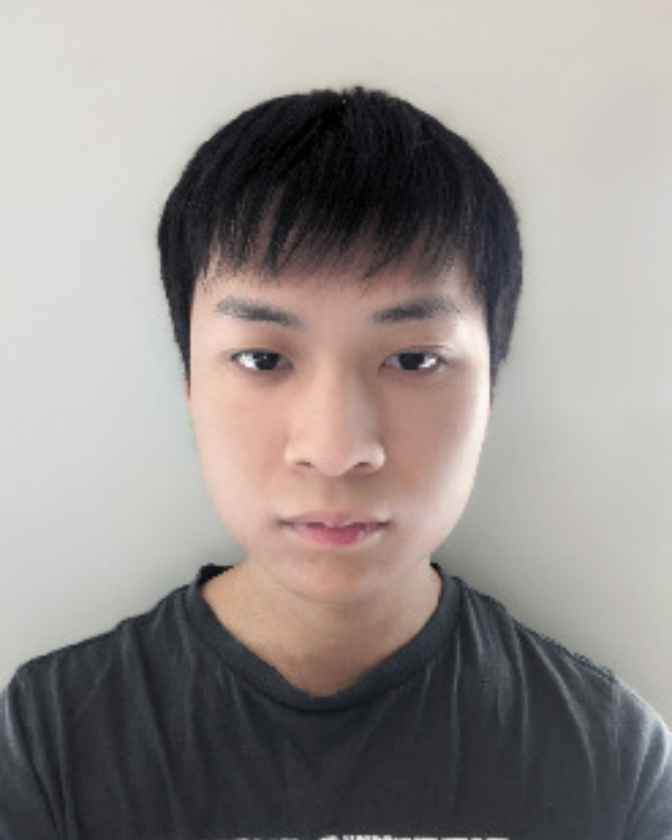}}]{Jiahao Xue} (Student Member, IEEE) received the B.S. degree from Huazhong University of Science and Technology, Wuhan, China, in 2018, and the M.S. degree from the University of California Santa Cruz, Santa Cruz, CA, USA, in 2021. He is currently a Ph.D. student of Systems and Security in the Department of Electrical Engineering, University of South Florida, Tampa, FL, USA. His primary research interests include network and mobile system security, and machine learning for networks.
\end{IEEEbiography}
\begin{IEEEbiography}[{\includegraphics[width=1in,height=1.8in,clip,keepaspectratio]{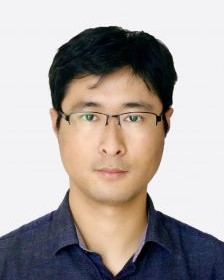}}]{Jie Xu} (Senior Member, IEEE) received the B.S. and M.S. degrees in electronic engineering from Tsinghua University, Beijing, China, in 2008 and 2010, respectively, and the Ph.D. degree in electrical engineering from UCLA, in 2015. He is currently an associate professor with the Department of Electrical and Computer Engineering, University of Florida. His research interests include mobile edge computing/intelligence, machine learning for networks, and network security. He received the NSF CAREER Award, in 2021.
\end{IEEEbiography}
\begin{IEEEbiography}[{\includegraphics[width=1in,height=1.8in,clip,keepaspectratio]{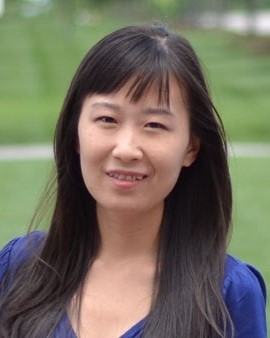}}]{Yao Liu} (Senior Member, IEEE) received the Ph.D. degree in computer science from North Carolina State University, in 2012. She is an professor with the Department of Computer Science and Engineering, University of South Florida. Her research interests include in the security applications for cyberphysical systems, Internet of Things, and machine learning. She was an NSF CAREER Award recipient in 2016. She also received the ACM CCS Test-of-Time Award by ACM SIGSAC in 2019.
\end{IEEEbiography}
\begin{IEEEbiography}[{\includegraphics[width=1in,height=1.8in,clip,keepaspectratio]{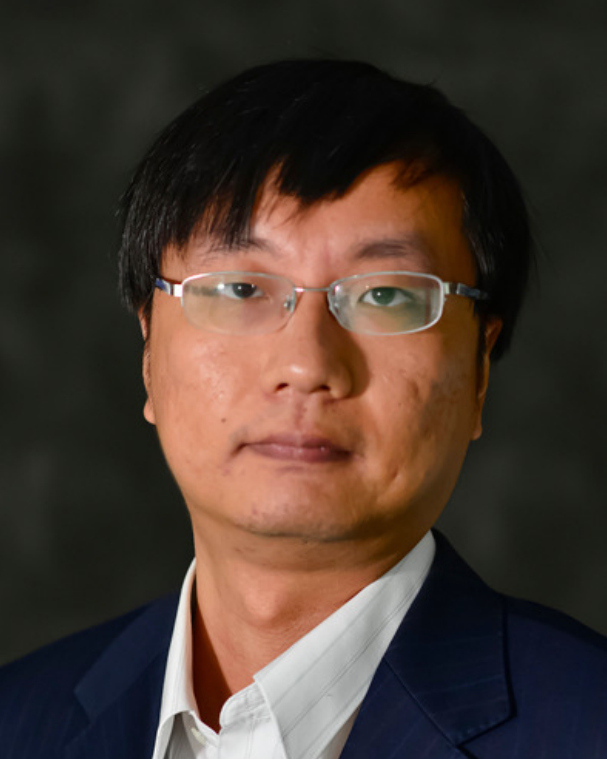}}]{Zhuo Lu} (Senior Member, IEEE) received the Ph.D. degree from North Carolina State University, in 2013. He is an associate professor with the Department of Electrical Engineering, University of South Florida. His research has been mainly focused on both theoretical and system perspectives on communication, network, and security. He received the NSF CISE CRII award in 2016, the Best Paper Award from IEEE GlobalSIP in 2019, and the NSF CAREER award in 2021.
\end{IEEEbiography}

\clearpage
\appendix
\section{Proofs}
\setcounter{equation}{0}
\renewcommand{\theequation}{\thesection.\arabic{equation}}
Here, we provide essential supplementary proofs and analyses for the main manuscript.

\subsection{Proof of Lemma $2$}
The proof of this lemma relies on (18) in the main paper and is carried out by examining the relationship between $w_i^t$ and $v_i^t$ in three cases of $t$.

\noindent \textit{Case 1}: $E\nmid t$.
Given that $w_i^t = v_i^t$ for all $i \in \{1, \ldots, K\}$, we can conclude that the aggregated models must also be equal, i.e., $\overline{w^t} = \overline{v^t}$.

\noindent \textit{Case 2}: $EG\mid t$.
Given that $w_i^{t} = \overline{v^t}$ for all $i \in \{1, \ldots, K\}$, we have
$	\overline{w^{t}} = \sum_{i=1}^K p_i w_i^{t} = \sum_{i=1}^K p_i \overline{v^{t}} = \overline{v^{t}}.
$

\noindent \textit{Case 3}: $E\mid t$ and $EG\nmid t$.
Given $\overline{v^{t}} = \sum_{i=1}^K p_i v_i^{t}$, the coefficient of $v_i^{t}$ in the expression of $\overline{v^{t}}$ is $p_i$. If $v_i^{t}$ has the same coefficient $p_i$ in $\overline{w^{t}}$ for all $i \in \{1, \ldots, K\}$, then it follows that $\overline{w^t} = \overline{v^t}$. We will use this idea as the foundation to complete the proof for this case.
Since $S_i$ denotes the set of indices of the ESs that client $i$ can connect to, for any $n \in S_i$, the model of $ES_n$ can be written as

\begin{equation}
	\begin{split}
		w_{(n)}^{t} = \left({\sum_{j\in {C_n \setminus i}} \frac{1}{|S_j|}p_j v_j^t}+\frac{1}{|S_i|}p_i v_i^j\right )\frac{1}{\lozenge_n},
	\end{split}
\end{equation}
where $\lozenge_n=\sum_{i\in C_n}\frac{p_i}{|S_i|}$.
It is clear that $w_{(n)}^t$ contains the term $v_i^t$ with a coefficient of $\frac{1}{|S_i|} p_i\frac{1}{\lozenge_n}$. For any $k \in C_n$, the model of client $k$ obtained through client aggregation, denoted as $w_k^t$, will be associated with $w_{(n)}^t$, ultimately causing $w_k^t$ to include the term $v_i^t$. In other words, the term related to $v_i^t$ in $w_{(n)}^t$ will be indirectly transmitted to $\overline{w^t}$ through $w_k^t$.
Assuming the coefficient of $v_i^t$ in $w_{(n)}^t$ is $\alpha_n^i$, then the coefficient of $v_i^t$ in $w_k^t$ contributed by $ES_n$ based on the \textit{client aggregation} is $\frac{1}{|S_k|}\alpha_n^i$. According to the definition of $\overline{w^t}$, the coefficient of $v_i^t$ in
$\overline{w^t}$ contributed by $ES_n$, denoted by $\beta_n^i$, would be expressed as
\begin{equation}
	\begin{split}\label{Eq:coefficient1}
		\beta_n^i&=\sum_{k\in C_n}{\frac{1}{|S_k|}\alpha_n^ip_k} \\
		&=\left(\sum_{k\in C_n}{\frac{1}{|S_k|}p_k}\right)\alpha_n^i.
	\end{split}
\end{equation}
Since the coefficient of $v_i^t$ in $w_{(n)}^t$ is $\frac{1}{|S_i|} p_i\frac{1}{\lozenge_n}$, we have $\alpha_n^i=\frac{1}{|S_i|} p_i\frac{1}{\lozenge_n}$ and (\ref{Eq:coefficient1}) can be rewritten as
\begin{equation}
	\begin{split}
		\beta_n^i
		=\left(\sum_{k\in C_n}{\frac{1}{|S_k|}p_k}\right)\frac{1}{|S_i|} p_i\frac{1}{\lozenge_n}.
	\end{split}
\end{equation}
As $v_i^t$ is only involved in the edge models of ESs indexed by $S_i$, the final coefficient of $v_i^t$ in $\overline{w^t}$, denoted by $\gamma_i$, is expressed as
\begin{equation}
	\begin{split}
		\gamma_i
		&=\sum_{n\in S_i}\beta_n^i\\
		&=\sum_{n\in S_i}\left[\left(\sum_{k\in C_n}{\frac{1}{|S_k|}p_k}\right)\frac{1}{|S_i|} p_i\frac{1}{\lozenge_n}\right]\\
		&=\sum_{n\in S_i}\left[\left(\sum_{k\in C_n}{\frac{1}{|S_k|}p_k}\right)\frac{1}{\lozenge_n}\right]\frac{1}{|S_i|} p_i\\
		&\stackrel{(\delta)}{=}\sum_{n\in S_i}\frac{1}{|S_i|} p_i=p_i\\
	\end{split}
\end{equation}
where ($\delta$) comes from (10) in the main paper, i.e., $\lozenge_n=\sum_{i\in C_n}\frac{p_i}{|S_i|}$.
Therefore, it is safe to say that the coefficient of $v_i^{t}$ in $\overline{w^{t}}$ always equals that in $\overline{v^{t}}$, and we can conclude that $\overline{w^{t}} = \overline{v^{t}}$. Which completes the proof.
\subsection{Proof of Lemma $3$} \label{proof4L3}
We first define $\mathbb{E'}[x_i]=\sum_{i=1}^{k}p_ix_i$, then we have
\begin{equation}
	\begin{split}
		\sum_{i=1}^K p_i \left\| \overline{w^t} - w_i^t \right\|^2
		&=\mathbb{E'}\left\| \overline{w^t} - w_i^t \right\|^2\\
		&=\mathbb{E'}\left\| (\overline{w^t}-\overline{w^{T_0}}) - (w_i^t-\overline{w^{T_0}}) \right\|^2\\
		&=\mathbb{E'}\left\| \mathbb{E'}[w_i^t-\overline{w^{T_0}}] - (w_i^t-\overline{w^{T_0}}) \right\|^2,
	\end{split}
\end{equation}
where $T_0$ is the largest multiple of $EG$ not exceeding $t$.
Furthermore, according to the mean-square error inequality (i.e., $\mathbb{E'} \left[ \| x - \mathbb{E'}[x] \|^2 \right] \leq \mathbb{E'} \left[ \| x \|^2 \right]$), we have
\begin{equation}\label{inequalw_hat}
	\begin{split}
		\sum_{i=1}^K p_i \left\| \overline{w^t} - w_i^t \right\|^2&=\mathbb{E'}\left\| \mathbb{E'}[w_i^t-\overline{w^{T_0}}] - (w_i^t-\overline{w^{T_0}}) \right\|^2\\
		&\leq\mathbb{E'}\left\|  (w_i^t-\overline{w^{T_0}}) \right\|^2\\
		&=\sum_{i=1}^K p_i \left\| w_i^t-\overline{w^{T_0}}\right\|^2.
	\end{split}
\end{equation}
Based on (18), the local models $w_i^{T_0}$ are identical across all clients $i \in \{1, \ldots, K\}$, i.e., $w_i^{T_0} = \overline{w^{T_0}}$. Consequently, we have $\sum_{i=1}^K p_i \left\| w_i^t-\overline{w^{T_0}}\right\|^2=\sum_{i=1}^K p_i \left\| w_i^t-w_i^{T_0}\right\|^2$. Substituting this into \eqref{inequalw_hat}, the inequality can be rewritten as
\begin{equation}\label{inequalw}
	\begin{split}
		\sum_{i=1}^K p_i \left\| \overline{w^t} - w_i^t \right\|^2
		\leq\sum_{i=1}^K p_i \left\| w_i^t-w_i^{T_0}\right\|^2.
	\end{split}
\end{equation}
Let $A_i^t=\left\| w_i^t-w_i^{T_0}\right\|^2$. Then, \eqref{inequalw} can be rewritten as
\begin{equation}\label{inequalfinal1_w}
	\begin{split}
		\sum_{i=1}^K p_i \left\| \overline{w^t} - w_i^t \right\|^2
		\leq\sum_{i=1}^K p_i A_i^t,
	\end{split}
\end{equation}
and we have
\begin{equation}\label{Eq:A_i}
	\begin{split}
		A_i^t&=\left\| w_i^t-w_i^{T_0}\right\|^2=\left\| \sum_{t'=T_0+1}^t w_i^{t'}-w_i^{t'-1}\right\|^2\\
		&=\left\| \sum_{t'=T_0+1}^t (v_i^{t'}-w_i^{t'-1})+(w_i^{t'}-v_i^{t'})\right\|^2\\
		&=\left\| \sum_{t'=T_0+1}^t (v_i^{t'}-w_i^{t'-1})+\sum_{t'=T_0+1}^t(w_i^{t'}-v_i^{t'})\right\|^2.
	\end{split}
\end{equation}
Note that $(w_i^{t'} - v_i^{t'})\neq0$ if and only if $E \mid t'$, according to (18).
Thus, the number of non-zero terms in $\sum_{t'=T_0+1}^t(w_i^{t'}-v_i^{t'})$ is $\left\lfloor \frac{t - T_0}{E} \right\rfloor$.
Using the inequality that the squared norm of the sum is bounded by the root mean square, we get
\begin{equation}\label{Ieq:AMR2MS1}
	\begin{split}
		&	\left\| \sum_{t'=T_0+1}^t (v_i^{t'}-w_i^{t'-1})+\sum_{t'=T_0+1}^t(w_i^{t'}-v_i^{t'})\right\|^2 \leq\\
		& \left[(t-T_0)+\left\lfloor \frac{t - T_0}{E} \right\rfloor \right]*\\
		&	\left[\sum_{t'=T_0+1}^t \left\|(v_i^{t'}-w_i^{t'-1})\right\|^2+\sum_{t'=T_0+1}^t \left\|(w_i^{t'}-v_i^{t'})\right\|^2\right].
	\end{split}
\end{equation}
Let $B_i^t=\sum_{t'=T_0+1}^t \left\|(v_i^{t'}-w_i^{t'-1})\right\|^2$ and $C_i^t=\sum_{t'=T_0+1}^t \left\|(w_i^{t'}-v_i^{t'})\right\|^2$. Then, based on \eqref{Eq:A_i} and \eqref{Ieq:AMR2MS1}, we have  
\begin{equation}\label{Ieq:AMR2MS2}
	\begin{split}
		A_i^t
		&	\leq \left[(t-T_0)+\left\lfloor \frac{t - T_0}{E} \right\rfloor \right] \left[B_i^t+C_i^t\right]\\
		&\leq \left[(GE-1)+(G-1) \right] \left[B_i^t+C_i^t\right]\\
		&=\left[(GE+G-2) \right] \left[B_i^t+C_i^t\right].
	\end{split}
\end{equation}
Furthermore, based on inequality \eqref{inequalfinal1_w}, we have
\begin{equation}\label{inequalfinal2_w}
	\begin{split}
		\sum_{i=1}^K p_i \left\| \overline{w^t} - w_i^t \right\|^2
		\leq\sum_{i=1}^K p_i \left[(GE+G-2) \right] \left[B_i^t+C_i^t\right],
	\end{split}
\end{equation}
and thus
\begin{equation}
	\begin{split}
		\mathbb{E}	\sum_{i=1}^K p_i \left\| \overline{w^t} - w_i^t \right\|^2
		\leq\mathbb{E} \sum_{i=1}^K p_i \left[(GE+G-2) \right] \left[B_i^t+C_i^t\right],
	\end{split}
\end{equation}
To bound the right-hand side of the above inequality, we begin by deriving an upper bound for $\mathbb{E} B_i^t$ as follows:
\begin{equation}\label{BoundB_i} 
	\begin{split}
		&\mathbb{E} \sum_{t'=T_0+1}^t \left\|(v_i^{t'}-w_i^{t'-1})\right\|^2\\
		&=\mathbb{E}\sum_{t'=T_0+1}^t  \eta_{t'-1}^2 	\left\|\nabla F_i\left( \mathbf w_i^{t'-1}, \mathbf \xi_i^{t'-1} \right)	\right\|^2\\
		&\leq \mathbb{E} \sum_{t'=T_0+1}^t  \eta_{T_0}^2 	\left\|\nabla F_i\left( \mathbf w_i^{t'-1}, \mathbf \xi_i^{t'-1} \right)	\right\|^2\\
		&=\sum_{t'=T_0+1}^t  \eta_{T_0}^2 \mathbb{E}	\left\|\nabla F_i\left( \mathbf w_i^{t'-1}, \mathbf \xi_i^{t'-1} \right)	\right\|^2\\
		&\leq \sum_{t'=T_0+1}^t  \eta_{t}^2 4^{\left\lceil \frac{t-T_0-1}{E} \right\rceil}	\mathbb{E} \left\|\nabla F_i\left( \mathbf w_i^{t'-1}, \mathbf \xi_i^{t'-1} \right)	\right\|^2\\
		&\leq \sum_{t'=T_0+1}^t  \eta_{t}^2 4^{G}\mathbb{E}	\left\|\nabla F_i\left( \mathbf w_i^{t'-1}, \mathbf \xi_i^{t'-1} \right)	\right\|^2\\
		&\leq \sum_{t'=T_0+1}^t  \eta_{t}^2 4^{G}H^2\leq (GE-1)  \eta_{t}^2 4^{G}H^2,
	\end{split}
\end{equation}
where the first inequality comes from $\eta_{t}$ is non-increasing, the second inequality comes from $\eta_{t}\leq2 \eta_{t+E}$, the third inequality comes from $\left\lceil \frac{t-T_0-1}{E} \right\rceil\leq{G}$, the fourth inequality comes from the Assumption 4, and the final inequality comes from $T_0 \leq t < T_0 + GE$.
Thus,
\begin{equation}\label{Ieq:BoundA-B}
	\begin{split}
		&\mathbb{E} \sum_{i=1}^K p_i \left[(GE+G-2) \right] B_i^t\\
		&\leq\left[(GE+G-2) \right](GE-1)  \eta_{t}^2 4^{G}H^2.
	\end{split}
\end{equation}
Similar to \eqref{Ieq:BoundA-B}, to bound $\mathbb{E} \sum_{i=1}^K p_i \left[(GE+G-2) \right]C_i^t$, we proceed as follows:
\begin{equation}\label{Ieq:BoundA-C-1}
	\begin{split}
		&\mathbb{E} \sum_{i=1}^K p_i \left[(GE+G-2) \right]C_i^t\\
		&=\mathbb{E}\sum_{i=1}^K p_i \left[(GE+G-2) \right]\sum_{t'=T_0+1}^t \left\|(w_i^{t'}-v_i^{t'})\right\|^2\\
		&=\left[(GE+G-2) \right]\mathbb{E}\sum_{i=1}^K p_i \sum_{t'=T_0+1}^t \left\|(w_i^{t'}-v_i^{t'})\right\|^2\\
		&=\left[(GE+G-2) \right]\sum_{t'=T_0+1}^t \sum_{i=1}^K p_i \mathbb{E} \left\|(w_i^{t'}-v_i^{t'})\right\|^2.\\
	\end{split}
\end{equation}
Let $\hat{t}$ denote the set of all step indices $t'$ in $\{T_0 + 1, \ldots, t\}$ such that $E \mid t'$, and let $U = GE + G - 2$. Then, $\hat{t}$ contains exactly $\left\lfloor \frac{t - T_0}{E} \right\rfloor$ elements. Accordingly, we can rewrite (\ref{Ieq:BoundA-C-1}) as
\begin{equation}\label{Ieq:BoundA-C-2} 
	\begin{split}
		&\mathbb{E} \sum_{i=1}^K p_i \left[(GE+G-2) \right]C_i^t\\
		&=U \sum_{t'\in\hat t} \sum_{i=1}^K p_i \mathbb{E} \left\|(w_i^{t'}-v_i^{t'})\right\|^2\\
		&\leq U \sum_{t'\in\hat t} \sum_{i=1}^K p_i \mathbb{E} \left\|(w_i^{t'-E}-v_i^{t'})\right\|^2\\
		&= U\sum_{t'\in\hat t} \sum_{i=1}^K p_i \mathbb{E} \left\|\sum_{t^\star=t'-E+1}^{t'}(v_i^{t^\star}-w_i^{t^\star-1})\right\|^2\\
		&\leq U\sum_{t'\in\hat t} \sum_{i=1}^K p_i  E\sum_{t^\star=t'-E+1}^{t'}\mathbb{E}\left\|(v_i^{t^\star}-w_i^{t^\star-1})\right\|^2\\
		&=U\sum_{t'\in\hat t} \sum_{i=1}^K p_i  E
		\sum_{t^\star=t'-E+1}^{t'}\mathbb{E}\left\|\eta_{t^\star-1}^2 	\left\|\nabla F_i\left( \mathbf w_i^{t^\star-1}, \mathbf \xi_i^{t^\star-1} \right)	\right\|^2\right\|^2\\
		&\leq U\sum_{t'\in\hat t} \sum_{i=1}^K p_i  E\sum_{t^\star=t'-E+1}^{t'} \eta_{t^\star-1}^2 H^2\\
		&\leq U\sum_{t'\in\hat t} \sum_{i=1}^K p_i  E\sum_{t^\star=t'-E+1}^{t'} \eta_{T_0}^2 H^2\\
		&\leq U\sum_{t'\in\hat t} \sum_{i=1}^K p_i  E\sum_{t^\star=t'-E+1}^{t'} \eta_{t}^2 4^{\left\lceil \frac{t-T_0-1}{E} \right\rceil} H^2\\
		&\leq U\sum_{t'\in\hat t} \sum_{i=1}^K p_i  E\sum_{t^\star=t'-E+1}^{t'} \eta_{t}^2 4^{G} H^2\\
		&= U\sum_{t'\in\hat t} \sum_{i=1}^K p_i  E^2 \eta_{t}^2 4^{G} H^2= U\left\lfloor \frac{t - T_0}{E} \right\rfloor  E^2 \eta_{t}^2 4^{G} H^2\\
		&\leq U(G-1)  E^2 \eta_{t}^2 4^{G} H^2,
	\end{split}
\end{equation}
where the first inequality comes from the fact that there is an edge aggregation at step $t' \in \hat{t}$ and $w_i^{t'}$ has a closer connection with $v_i^{t'}$ compared to $w_i^{t' - E}$, the second inequality comes from  the arithmetic mean is less than or equal to the root mean square, the third inequality comes from Assumption 4, the fourth inequality comes from $\eta_{t}$ is non-increasing, the fifth inequality comes from $\eta_{t}\leq2 \eta_{t+E}$, the sixth inequality comes from $\left\lceil \frac{t-T_0-1}{E} \right\rceil\leq{G}$, and the final inequality comes from $\left\lfloor \frac{t - T_0}{E} \right\rfloor\leq(G-1)$.

Finally, we have
\begin{equation}\label{BoundC_i^t}
	\begin{split}
		&\mathbb{E}\sum_{i=1}^K p_i \left\| \overline{w^t} - w_i^t \right\|^2=\mathbb{E}\sum_{i=1}^K p_i A_i^t\\
		&\leq\mathbb{E}\sum_{i=1}^K p_i \left[(GE+G-2) \right] \left[B_i^t+C_i^t\right]\\
		&\leq(GE+G-2)(GE-1)  \eta_{t}^2 4^{G}H^2\\
		&\;\;\;+ (GE+G-2)(G-1)E^2 \eta_{t}^2 4^{G} H^2\\
		&=\eta_{t}^2 4^{G}H^2(GE+G-2)\left[(GE-1)+E^2(G-1)\right],
	\end{split}
\end{equation}
where the first inequality comes from (\ref{Ieq:AMR2MS2}), the second inequality comes from (\ref{Ieq:BoundA-B}) and (\ref{Ieq:BoundA-C-2}). Which completes the proof.
\subsection{Proof of Corollary $2$} \label{proof4C2}
Based on the above analysis, we now have
\begin{equation}\label{proof4T1-2}
	\begin{split}
		\Delta_{t+1}
		&\leq (1 - \eta_t \mu) \Delta_{t}
		+ \eta_t^2 \mathbb{E} \left\| g_t - \overline{{g}^t} \right\|^2 \\
		&\;\;+ 6L \eta_t^2 (F^\star - \sum_{k=1}^N p_k F_k^\star)
		+ 2 \mathbb{E} \sum_{i=1}^K p_i \left\| \overline{{w}^t} - w_i^t \right\|^2,
		\\
		&\leq (1 - \eta_t \mu) \Delta_{t}
		+ \eta_t^2 \sum_{i=1}^K {p_i}^2\sigma_i^2 \\
		&\;\;\;+ 6L \eta_t^2 (F^\star - \sum_{k=1}^N p_k F_k^\star)
		+ 2 \mathbb{E} \sum_{i=1}^K p_i \left\| \overline{{w}^t} - w_i^t \right\|^2\\
		&\leq (1 - \eta_t \mu) \Delta_{t}
		+ \eta_t^2 \sum_{i=1}^K {p_i}^2\sigma_i^2 + 6L \eta_t^2 (F^\star - \sum_{k=1}^N p_k F_k^\star)\\
		&\;\;\;+ 2 \eta_{t}^2 4^{G}H^2(GE+G-2)\left[(GE-1)+E^2(G-1)\right]
	\end{split}
\end{equation}
where the first inequality comes from Corollary $1$, the second inequality comes from Lemma $4$, and the third inequality comes from Lemma $3$ .
Let $ X = \sum_{i=1}^K p_i^2 \sigma_i^2 + 6L(F^\star - \sum_{k=1}^N p_k F_k^\star) + 2^{2G+1}H^2(GE+G-2)[(GE-1)+E^2(G-1)]$, then \eqref{proof4T1-2} can be simplified as
\begin{equation}\label{proof4T1-3}
	\begin{split}
		\Delta_{t+1}
		&\leq (1 - \eta_t \mu) \Delta_{t}
		+ \eta_t^2 X,
	\end{split}
\end{equation}
which completes the proof.
\subsection{Proof of Theorem $1$} \label{proof4T1}
Since $\eta_t = \frac{\beta}{t + \alpha}$ with $\beta > \frac{1}{\mu} > 0$ and $\alpha > 0$, it follows that $\eta_t$ is non-increasing.  
Moreover, since $\eta_0 \leq \min\left\{ \frac{1}{\mu}, \frac{1}{4L} \right\}$, we have $\eta_t \leq \eta_0\leq \frac{1}{4L}$. Additionally, note that $(2\eta_{t+E} - \eta_t) = \frac{2\beta}{t+E+\alpha} - \frac{\beta}{t+\alpha} = \frac{\beta(t+\alpha-E)}{(t+E+\alpha)(t+\alpha)} \geq \frac{\beta(\alpha-E)}{(t+E+\alpha)(t+\alpha)}$. Since $\alpha\geq E>0$, we have $(2\eta_{t+E} - \eta_t)\geq 0$, i.e., $\eta_t \leq 2\eta_{t+E}$. The three properties, together with Assumptions~1–4, satisfy all the conditions required in Corollary~$2$. Consequently, based on (26) in the main paper, we have $\Delta_{t+1} \leq (1 - \eta_t \mu) \Delta_{t} + \eta_t^2 X$, where $X$ is defined therein.

Here, we introduce a crucial assumption: $\Delta_t \leq \frac{Z}{t+\alpha}$, where $Z = \max\{\frac{\beta^2 X}{\beta\mu-1}, \Delta_0 \alpha\}$. This assumption is critical for the subsequent proof of Theorem~$1$ and we will prove it using mathematical induction.
Since $Z = \max\left\{ \frac{\beta^2 X}{\beta \mu - 1},\, \Delta_0 \alpha \right\}$, we have $Z \geq \Delta_0 \alpha$, which implies
$\Delta_0 \leq \frac{Z}{\alpha} = \frac{Z}{0 + \alpha}.
$
Hence, the assumption holds for the base case $t = 0$.
Next, assume that $\Delta_{t_0} \leq \frac{Z}{t_0 + \alpha}$ holds for some $t_0 \geq 0$. Then, for $\Delta_{t_0 + 1}$, we have
\begin{equation}\label{induction}
	\begin{split}
		\Delta_{t_0+1}&\leq(1 - \eta_{t_0} \mu) \Delta_{t_0}+ \eta_{t_0}^2 X\\
		&\leq(1 - \eta_{t_0} \mu)\frac{Z}{t_0+\alpha}+\eta_{t_0}^2 X\\
		&=\left(1 - \frac{\beta}{t_0+\alpha} \mu\right)\frac{Z}{t_0+\alpha}+\frac{\beta^2}{(t+\alpha)^2} X\\
		&=\frac{Z\left(1 - \frac{\beta}{t_0+\alpha} \mu\right)(t_0+\alpha)+\beta^2X}{(t_0+\alpha)^2}\\
		&=\left[\frac{Z\left(t_0+\alpha -\beta \mu\right)+\beta^2X}{(t_0+\alpha)^2}-\frac{Z}{t_0+1+\alpha}\right]\\
		&\;\;\;\;+\frac{Z}{t_0+1+\alpha}.
	\end{split}
\end{equation}
In the inequality above, the first term can be expressed as
\begin{equation}\label{Diffterm}
	\begin{split}
		&\frac{Z\left(t_0+\alpha -\beta \mu\right)+\beta^2X}{(t_0+\alpha)^2}-\frac{Z}{t_0+1+\alpha}\\
		&=Z\left[ \frac{t_0+\alpha -\beta \mu}{(t_0+\alpha)^2}-\frac{1}{t_0+1+\alpha}\right]+\frac{\beta^2X}{(t_0+\alpha)^2}\\
		&=\frac{1}{(t_0+\alpha)^2}\left\{Z\left[t_0+\alpha-\beta\mu-\frac{(t_0+\alpha)^2}{t_0+\alpha+1}\right]+\beta^2X\right\}\\
		&=\frac{1}{(t_0+\alpha)^2}\left[Z\left(\frac{t_0+\alpha}{t_0+\alpha+1}-\beta\mu\right)+\beta^2X\right].
	\end{split}
\end{equation}
Since $0 < \frac{t_0 + \alpha}{t_0 + \alpha + 1} < 1$, it follows that $\beta\mu - \frac{t_0 + \alpha}{t_0 + \alpha + 1} > \beta\mu - 1$, which implies $\frac{\beta^2 X}{\beta \mu - \frac{t_0 + \alpha}{t_0 + \alpha + 1}} < \frac{\beta^2 X}{\beta \mu - 1}$. Given $Z = \max\left\{\frac{\beta^2 X}{\beta \mu - 1}, \Delta_0 \alpha\right\}$, we have $Z \geq \frac{\beta^2 X}{\beta \mu - 1}\geq \frac{\beta^2 X}{\beta \mu - \frac{t_0 + \alpha}{t_0 + \alpha + 1}}$, and thus $\left[Z\left(\frac{t_0+\alpha}{t_0+\alpha+1}-\beta\mu\right)+\beta^2X\right]\leq0$. 
Which shows that the expression in~\eqref{Diffterm} is non-positive. Therefore, based on~\eqref{induction}, we conclude that $\Delta_{t_0 + 1} \leq \frac{Z}{t_0 + 1 + \alpha}$, thereby completing the induction and proving the assumption $\Delta_t \leq \frac{Z}{t+\alpha}$.

Based on Assumption 1, we know that $F(w)$ is a $L$-smooth function, which implies for any $\mathbf w$ and $\mathbf {w'}$, $F(\mathbf{w}) \leq F(\mathbf{w}') + (\mathbf{w} - \mathbf{w}') \nabla F(\mathbf{w}') + \frac{L}{2} \| \mathbf{w} - \mathbf{w}' \|^2$. Let $\mathbf w=\overline{w^t}$ and $\mathbf {w'}=w^\star$, since $\nabla F(w^\star)=0$, we can get

\begin{equation}
	\begin{split}
		F(\overline{w^t}) &\leq F(w^\star) + (\overline{w^t} - w^\star) \nabla F(w^\star) + \frac{L}{2} \| \overline{w^t} - w^\star \|^2\\
		&=F(w^\star) + \frac{L}{2} \| \overline{w^t} - w^\star \|^2,
	\end{split}
\end{equation}
and thus
\begin{equation}
	\begin{split}
		\mathbb{E} \left[F(\overline{w^t}) -F(w^\star)\right]&\leq \mathbb{E} \left[\frac{L}{2} \| \overline{w^t} - w^\star \|^2\right]\\
		&=\frac{L}{2}\mathbb{E} \left[\| \overline{w^t} - w^\star \|^2\right]=\frac{L}{2}\mathbb{E} \left[\Delta_t \right]\\
		&\leq \frac{L}{2} \frac{Z}{t+\alpha},
	\end{split}
\end{equation}
where the last inequality comes from our introduced assumption. 
Which completes the proof.
\hlb \subsection{Proof of Remark $5$} \label{proof4Remark4}
We assume that Assumptions~1 to 4 also hold under the HFL setting. Since Lemma~1 is derived under Assumptions~1 and~2, the same result can be established for conventional HFL as well. 

To show that Lemma~2 also holds under HFL, we first define the \emph{virtual global models} $\overline{{w}^{t}}$ and $\overline{{v}^{t}}$, which are obtained by hypothetically performing {edge aggregation} followed by {cloud aggregation} over all client models $\{\mathbf{w}_i^{t}\}_{i=1}^{K}$ and $\{\mathbf{v}_i^{t}\}_{i=1}^{K}$, respectively. According to the HFL aggregation procedure described in Section~II-B, we have

\begin{equation}\label{HFL_w1}
	\begin{split}
		\overline{{w}^{t}}
		&=\sum_{n=1}^{N}\frac{N_n}{N_{\mathrm{total}}}{w}_{(n)}^{t}
		=\sum_{n=1}^{N}\frac{N_n}{N_{\mathrm{total}}}\sum_{i\in \mathcal{C}_n}\frac{N_i}{N_n}{w}_i^{t}\\
		&=\sum_{n=1}^{N}\frac{1}{N_{\mathrm{total}}}\sum_{i\in \mathcal{C}_n}{N_i}{w}_i^{t}=\sum_{n=1}^{N}\sum_{i\in \mathcal{C}_n}\frac{N_i}{N_{\mathrm{total}}}{w}_i^{t},
	\end{split}
\end{equation}
where $N_i$ is the number of local data held by client $i$, and $N_n = \sum_{i \in \mathcal{C}_n} N_i$ is the number of data covered by $ES_n$. Since in HFL each client $i$ can connect to only one ES, the client sets $\{\mathcal{C}_n\}_{n=1}^{N}$ are pairwise disjoint. Therefore, \eqref{HFL_w1} can be rewritten as
\begin{equation}\label{HFL_w_2}
	\overline{{w}^{t}}
	=\sum_{i=1}^{K}\frac{N_i}{N_{\mathrm{total}}}{w}_i^{t}=\sum_{i=1}^K p_i w_i^t,
\end{equation}
where $p_i \triangleq \frac{N_i}{N_{\mathrm{total}}}$.
Similarly, we obtain
\begin{equation}\label{HFL_v_2}
	\overline{{v}^{t}}
	=\sum_{i=1}^{K}\frac{N_i}{N_{\mathrm{total}}}{v}_i^{t}
	=\sum_{i=1}^{K} p_i {v}_i^{t}.
\end{equation}
Next, similar to the proof in Appendix~A, we show that the equality $\overline{{w}^{t}}=\overline{{v}^{t}}$  can be established by examining the relationship between ${w}_i^{t}$ and ${v}_i^{t}$ under three cases of $t$.

\noindent \textit{Case 1}: $E\nmid t$.
Given that $w_i^t = v_i^t$ for all $i \in \{1, \ldots, K\}$, we can conclude that the aggregated models must also be equal, i.e., $\overline{w^t} = \overline{v^t}$.

\noindent \textit{Case 2}: $EG\mid t$.
Given that $w_i^{t} = \overline{v^t}$ for all $i \in \{1, \ldots, K\}$, we have
$	\overline{w^{t}} = \sum_{i=1}^K p_i w_i^{t} = \sum_{i=1}^K p_i \overline{v^{t}} = \overline{v^{t}}.
$

\noindent \textit{Case 3}: $E\mid t$ and $EG\nmid t$. Given that $w_{(n)}^t=\sum_{i\in \mathcal{C}_n} \frac{v_i^t p_i}{\sum_{i\in \mathcal{C}_n}p_i}$, we have
\begin{equation}\label{HFL_w3}
	\begin{split}
		\overline{{w}^{t}}
		&=\sum_{i=1}^K p_i w_i^t=\sum_{n=1}^N \sum_{i\in \mathcal{C}_n}p_i w_i^t=\sum_{n=1}^N \sum_{i\in \mathcal{C}_n}p_i w_{(n)}^t\\
		&=\sum_{n=1}^N w_{(n)}^t\sum_{i\in \mathcal{C}_n}p_i =\sum_{n=1}^N \sum_{i\in \mathcal{C}_n} \frac{v_i^t p_i}{\sum_{i\in \mathcal{C}_n}p_i}\sum_{i\in \mathcal{C}_n}p_i\\
		&=\sum_{n=1}^N \sum_{i\in \mathcal{C}_n}p_i v_i^t=\sum_{i=1}^K p_i v_i^t=\overline{v^t}.
	\end{split}
\end{equation}
Therefore, we conclude that $\overline{\mathbf{w}}^{t}=\overline{\mathbf{v}}^{t}$, and hence Lemma~2 holds under the HFL setting. Moreover, since Corollary~1 is built upon Lemmas~1 and~2, it also holds under the HFL setting. For Lemma~4, since it is derived under Assumption~3, the same result also holds under the HFL setting.

We denote the upper bound of the HFL term
$\mathbb{E}\sum_{i=1}^{K} p_i \bigl\|\overline{{w}^{t}}-{w}_i^{t}\bigr\|^{2}$
by $M'$, which yields
\begin{equation} \label{HFL_M'}
	\mathbb{E}\sum_{i=1}^{K} p_i \left\|\overline{{w}^{t}}-\mathbf{w}_i^{t}\right\|^{2}
	\leq M'.
\end{equation}
Following the same steps as in Appendix~C, we can derive an analogous recursion:
\begin{equation}\label{Eq:HFL_recursion}
	\Delta_{t+1}
	\leq (1-\eta_t \mu)\Delta_t
	+\eta_t^{2} X',
\end{equation}
where
$X' = \sum_{i=1}^{K} p_i^{2}\sigma_i^{2}
+ 6L\!\left(F^\star-\sum_{k=1}^{K} p_k F_k^\star\right)
+ \frac{2M'}{\eta_t^{2}}$. 
Similarly, we denote by $M$ the corresponding upper bound of the HHFL term
$\mathbb{E}\sum_{i=1}^{K} p_i \left\|\overline{\mathbf{w}}^{t}-\mathbf{w}_i^{t}\right\|^{2}$.
Based on Corollary~2, the associated constant in HHFL can be written as
$X = \sum_{i=1}^{K} p_i^{2}\sigma_i^{2}
+ 6L\!\left(F^\star-\sum_{k=1}^{K} p_k F_k^\star\right)
+ \frac{2M}{\eta_t^{2}}$.
Note that $\sum_{i=1}^{K} p_i^{2}\sigma_i^{2}$ is determined by Assumption~3 and captures the data heterogeneity {within} each client, while
$F^\star-\sum_{k=1}^{K} p_k F_k^\star$ is determined by the data heterogeneity {across} different clients.
Moreover, we assume the same learning rate $\eta_t$ for both HHFL and HFL at any step $t$. 
Therefore, the difference between $X'$ and $X$ is solely determined by the difference between $M'$ and $M$. As shown in Appendix~D, Theorem~1 is obtained via an inductive argument based on Corollary~2, where $X$ appears as a constant throughout the derivation. Therefore, for HFL, a similar result can be derived:
$
\mathbb{E}\!\left[F\!\left(\overline{\mathbf{w}}^{t}\right)\right]-F^\star
\leq \frac{L}{2} \frac{Z'}{t+\alpha},
$
where
$
Z' = \max\!\left\{\frac{\beta^{2}X'}{\beta\mu+1},\,\Delta_{0}\alpha\right\}.
$
Since the difference between $X'$ (for HFL) and $X$ (for HHFL) is determined by $M'$ and $M$, it follows that the difference between $Z'$ and $Z$ is also governed by $M'$ and $M$. In other words, $M'$ and $M$ capture the convergence efficiency of HFL and HHFL: if $M'>M$, then HFL converges more slowly than HHFL; otherwise, HHFL is slower.

Next, we focus on comparing the magnitudes of $M'$ and $M$.
By definition, $M'$ upper bounds the drift term $\mathbb{E}\sum_{i=1}^K p_i\|\overline{{w}^{t}}-w_i^t\|^2$ under HFL,
whereas $M$ upper bounds the same term under HHFL; hence, a larger drift directly leads to a larger value of $M'$ or $M$.
To reason about how this drift differs between the two architectures, we define the disagreement
$D_t\triangleq \sum_{i=1}^K p_i\|\overline{{w}^{t}}-w_i^t\|^2$.
Within each edge round, every client performs $E$ local SGD steps before the next aggregation, so the local models repeatedly accumulate gradient steps and deviate from $\overline{{w}^{t}}$; consequently, $D_t$ increases with larger $E$ (e.g., under bounded gradients, it scales on the order of $\eta_t^2 H^2 E^2$).
Moreover, since cloud aggregation occurs only once every $G$ edge rounds, this deviation has more time to accumulate when $G$ is larger, which further enlarges the drift. The key difference between HFL and HHFL is how much of this accumulated disagreement can be mitigated before the next cloud aggregation: in HFL, each client is associated with exactly one ES and the client sets across ESs are pairwise disjoint, so edge aggregations remain isolated across ESs between two cloud aggregations and thus cross-ES disagreement is only weakly reduced, leading to a larger accumulated drift and hence a larger $M'$.
In contrast, HHFL enables clients located in overlapping regions (i.e., $|S_i|>1$) to act as {bridges} between ESs by averaging multiple ES models for initialization and simultaneously uploading their updated models to multiple ESs at every edge round; this continual cross-ES information mixing suppresses model divergence and mitigates the drift growth caused by large $E$ and $G$.
As a result, the drift bound under HFL increases {more rapidly} with $E$ and $G$ than that under HHFL, so we typically have $M'>M$, and the gap $M'-M$ (and thus the gaps $X'-X$ and $Z'-Z$) becomes more pronounced as $E$ and $G$ increase, ultimately making the advantage of HHFL over HFL more evident. Which completes the proof.
{\hlb\subsection{Analysis of ES--client Transmission Resource Consumption}
	Let $R$ denote the multiplicative factor of the per-round resource consumption for ES--client model transmissions in HHFL relative to HFL. In general, we have $R>1$.  When ES--client transmissions are implemented in a pure unicast manner (i.e., each ES--client transmission is scheduled over orthogonal resources in a one-to-one manner), the resource consumption per round scales near-linearly with the number of ES--client links. This is because each ES--client link carries a uniform payload (i.e., the model parameters), leading to nearly identical resource consumption for each round-trip transmission between them. In HHFL, the number of ES--client links is  $\sum_{i=1}^{K} |S_i|$, whereas it is $K$ in HFL. Therefore, the multiplicative factor of the per-round resource consumption for ES--client model transmissions in HHFL relative to HFL can be approximated by $\frac{\sum_{i=1}^{K} |S_i|}{K}$. However, in our practical HHFL implementation, model transmissions over additional ES–client links leverage the one-to-many nature of the wireless interface. Specifically, clients in overlapping areas can transmit their local models once, while multiple ESs decode the same uplink transmission via CoMP-enabled reception (uplink JR). Compared to the traditional single-link uplink in HFL, such clients do not require additional transmitter-side uplink resources (e.g., PRB-time footprint and transmission energy) to support these concurrent receptions.
	Therefore, the multiplicative factor of the per-round resource consumption for ES--client model transmissions in HHFL relative to HFL under pure unicast transmission can be viewed as a conservative upper bound on $R$, and we have $1 < R <\frac{\sum_{i=1}^{K} |S_i|}{K}$. Moreover, for certain resource metric such as energy consumption, the cost does not necessarily scale linearly with the number of ES--client associations. This is because parallel transmission/reception can share common RF and front-end baseband operations (e.g., a common local oscillator, up/downconversion, and a shared OFDM processing pipeline). Consequently, the marginal cost of supporting additional links can be lower, suggesting that the actual $R$ can be tighter than the above conservative unicast-based bound.
	
	Given our experimental setup where $\sum_{i=1}^{K} |S_i|=75$ and $K=57$, we have $1 < R < \frac{\sum_{i=1}^{K} |S_i|}{K}=\frac{75}{57}\approx1.32$. 
	As shown in Fig.~11 and Fig.~14 (b), HHFL reduces the required training steps/communication rounds to reach the same target accuracy by more than $1.5\times$. Since these convergence speedups exceed 1.32 (and consequently $R$ itself), HHFL maintains a clear advantage in terms of total ES--client model transmission resource consumption required to reach convergence.
}

\bibliographystyle{IEEEtran}

\end{document}